\documentclass[aps,pra,notitlepage, superscriptaddress, twocolumn]{revtex4-2}

\usepackage[T1]{fontenc}
\usepackage{amsmath,amsfonts,amssymb,amsthm,bbm,bm, braket}
\usepackage{wasysym}
\usepackage{comment}
\usepackage{scalerel}
\usepackage{enumerate}
\usepackage{graphicx}
\usepackage{mathtools}
\usepackage{ifthen}
\usepackage{tensor}
\usepackage{tikz}
\usepackage{tikz-network}
\usetikzlibrary{patterns,decorations.pathreplacing}

\usepackage{longtable}
\usepackage{cancel}

\theoremstyle{break}

\usepackage{color}
\definecolor{myred}{RGB}{232,102,102}
\definecolor{myblue}{RGB}{187,187,255}
\definecolor{myorange}{RGB}{255,165,0}
\definecolor{mygrey}{RGB}{105,105,105}
\definecolor{OliveGreen}{RGB}{85,107,47}
\definecolor{NavyBlue}{RGB}{0,0,128}
\definecolor{mygreen}{RGB}{34,139,34}
\definecolor{myY}{RGB}{220,255,203}
\definecolor{myYO}{RGB}{255, 220, 151}
\usepackage{tensor}
\newcommand{\be}{\begin{equation}}
\newcommand{\ee}{\end{equation}}
\newcommand{\ba}{\begin{aligned}}
\newcommand{\ea}{\end{aligned}}
\newcommand{\bw}{\begin{widetext}}
\newcommand{\ew}{\end{widetext}}
\newcommand{\1}{\mathbbm{1}}

\theoremstyle{plain}
\newtheorem{property}{Property}
\theoremstyle{plain}
\newtheorem{lemma}{Lemma}
\theoremstyle{plain}

\usepackage[colorlinks,bookmarks=false,citecolor=NavyBlue,linkcolor=OliveGreen,urlcolor=blue]{hyperref}

\newcommand{\Wgategreen}[2]{
\draw[very thick] (#1-0.5, #2 +0.5) -- (#1+0.5,#2-0.5);
\draw[very thick] (#1-0.5,#2-0.5) -- (#1+0.5,#2+0.5);
\draw[ thick, fill=mygreen, rounded corners=2pt] (#1-0.25,#2+0.25) rectangle (#1+0.25,#2-0.25);
\draw[thick] (#1,#2+0.15) -- (#1+0.15,#2+0.15) -- (#1+0.15,#2);
}

\newcommand{\Wreduced}[2]{
\draw[thin] (#1-0.5, #2 +0.5) -- (#1+0.5,#2-0.5);
\draw[thin] (#1-0.5,#2-0.5) -- (#1+0.5,#2+0.5);
\draw[ thick, fill=myY, rounded corners=2pt] (#1-0.25,#2+0.25) rectangle (#1+0.25,#2-0.25);
\draw[thick] (#1,#2+0.15) -- (#1+0.15,#2+0.15) -- (#1+0.15,#2);
}

\definecolor{myblue1}{RGB}{176,223,229}
\definecolor{myblue2}{RGB}{0,0,128}
\definecolor{myblue3}{RGB}{0,108,255}
\definecolor{myblue4}{RGB}{101,147,245}
\definecolor{myblue5}{RGB}{115,194,251}
\definecolor{myblue6}{RGB}{87,160,211}
\definecolor{myblue7}{RGB}{137,207,240}
\definecolor{myblue8}{RGB}{29,41,81}
\definecolor{myblue9}{RGB}{14,77,146}
\definecolor{myblue10}{RGB}{15,82,186}
\definecolor{myred1}{RGB}{255,36,0}
\definecolor{myred2}{RGB}{205,92,92}
\definecolor{myred3}{RGB}{178,34,34}
\definecolor{myred4}{RGB}{164,90,82}
\definecolor{myred5}{RGB}{255,8,0}
\definecolor{myred6}{RGB}{202,52,51}
\definecolor{myred7}{RGB}{66,13,9}
\definecolor{myred8}{RGB}{141,2,31}
\definecolor{myred9}{RGB}{250,128,114}
\definecolor{myred10}{RGB}{237,41,57}

\definecolor{myyellow1}{RGB}{254,220,86}
\definecolor{myyellow2}{RGB}{255,229,180}
\definecolor{myyellow5}{RGB}{255,195,11}
\definecolor{myyellow6}{RGB}{218,165,32}

\newcommand{\mcirc}{\mathbin{\scalerel*{\fullmoon}{G}}}
\newcommand{\mcircf}{\mathbin{\scalerel*{\newmoon}{G}}}

\newcommand{\lineW}{.8mm}
\newcommand{\sroot}{1.41421356237}

\newcommand{\oGate}[2]{
\draw[thin, fill=myY, rounded corners=0pt] (#1-0.5,#2+0.5) rectangle (#1+0.5,#2-0.5);}

\newcommand{\aGate}[2]{
\draw[thin, fill=myY, rounded corners=0pt] (#1-0.5,#2+0.5) rectangle (#1+0.5,#2-0.5);
\draw[line width=\lineW] (#1-0.5,#2) -- (#1+0.5,#2);}

\newcommand{\cGate}[2]{
\draw[thin, fill=myY, rounded corners=0pt] (#1-0.5,#2+0.5) rectangle (#1+0.5,#2-0.5);
\draw[line width=\lineW] (#1,#2-0.5) -- (#1,#2+0.5);}

\newcommand{\gGate}[2]{
\draw[thin, fill=myY, rounded corners=0pt] (#1-0.5,#2+0.5) rectangle (#1+0.5,#2-0.5);
\draw[line width=\lineW] (#1,#2-0.5) -- (#1,#2+0.5);
\draw[line width=\lineW] (#1-0.5,#2) -- (#1+0.5,#2);}

\newcommand{\doGate}[2]{
\draw[thin, fill=myY, rounded corners=0pt] (#1-0.5,#2+0.5) rectangle (#1+0.5,#2-0.5);
\draw[line width=\lineW] (#1,#2-0.5) -- (#1,#2) -- (#1 +0.5,#2);}

\newcommand{\dtGate}[2]{
\draw[thin, fill=myY, rounded corners=0pt] (#1-0.5,#2+0.5) rectangle (#1+0.5,#2-0.5);
\draw[line width=\lineW] (#1-0.5,#2) -- (#1,#2) -- (#1 ,#2+0.5);}

\newcommand{\eeGate}[2]{
\draw[thin, fill=myY, rounded corners=0pt] (#1-0.5,#2+0.5) rectangle (#1+0.5,#2-0.5);
\draw[line width=\lineW] (#1,#2-0.5) to[out=90,in= 180] (#1+0.5,#2);
\draw[line width=\lineW] (#1-0.5,#2) to[out=0,in=-90] (#1,#2+0.5);}

\newcommand{\acGate}[2]{
\draw[thin, fill=myY, rounded corners=0pt] (#1-0.5,#2+0.5) rectangle (#1+0.5,#2-0.5);
\draw[line width=\lineW] (#1,#2-0.5) -- (#1,#2-0.25) to[out=0,in= 0] (#1,#2+0.25)-- (#1,#2+0.5);
\draw[line width=\lineW] (#1-0.5,#2) -- (#1+0.5,#2);}

\hypersetup{
pdftitle={Chaos and Ergodicity in Extended Quantum Systems with Noisy Driving},
pdfsubject={Statistical mechanics},
pdfauthor={Pavel Kos, Bruno Bertini, and Tomaz Prosen},
pdfkeywords={chaos, ergodicity, quantum many-body systems}
}
\begin{document}

\title{Chaos and Ergodicity in Extended Quantum Systems with Noisy Driving}
\date{\today}

\author{Pavel Kos}
\affiliation{Department of Physics, Faculty of Mathematics and Physics, University of Ljubljana, Jadranska 19, SI-1000 Ljubljana, Slovenia}

\author{Bruno Bertini}
\affiliation{Department of Physics, Faculty of Mathematics and Physics, University of Ljubljana, Jadranska 19, SI-1000 Ljubljana, Slovenia}
\affiliation{Rudolf Peierls Centre for Theoretical Physics, Clarendon Laboratory, Oxford University, Parks Road, Oxford OX1 3PU, United Kingdom}

\author{Toma\v z Prosen}
\affiliation{Department of Physics, Faculty of Mathematics and Physics, University of Ljubljana, Jadranska 19, SI-1000 Ljubljana, Slovenia}

\begin{abstract}
We study the time evolution operator in a family of local quantum circuits with random fields in a fixed direction. We argue that the presence of quantum chaos implies that at large times the time evolution operator becomes effectively a \emph{random matrix} in the many-body Hilbert space. To quantify this phenomenon we compute analytically the squared magnitude of the trace of the evolution operator --- the generalised spectral form factor --- and compare it with the prediction of Random Matrix Theory (RMT). We show that for the systems under consideration the generalised spectral form factor can be expressed in terms of dynamical correlation functions of local observables in the infinite temperature state, linking chaotic and ergodic properties of the systems. This also provides a connection between the many-body Thouless time $\tau_{\rm th}$ --- the time at which the generalised spectral form factor starts following the random matrix theory prediction --- and the conservation laws of the system. Moreover, we explain different scalings of $\tau_{\rm th}$ with the system size, observed for systems with and without the conservation laws.
\end{abstract}

\maketitle

The concept of \emph{chaos} is very natural in classical systems. Its naive formulation in terms of strong sensitivity of the trajectory to the initial conditions, the ``butterfly effect'', is so simple and powerful that has long become an element of the popular culture. During the second half of the twentieth century this concept has been refined, from both the physical and mathematical points of view, leading to a complete theory of chaos in classical dynamical systems~\cite{arnold, arnold2, ott, cornfeld} that can be regarded as one of the greatest achievements of mathematical physics. 

In the quantum realm the situation is much less intuitive due to the absence of well defined trajectories and the linear structure of the unitary evolution. In this context, a key role is played by the spectral correlations of the time-evolution operator. Indeed, as established in a series of seminal works~\cite{CGV80, Berry81, BGS84}, systems with a well defined chaotic classical limit have a spectrum with correlations that coincide with those of an ensemble of random matrices with the same symmetries. The latter property remains well defined also away from the classical limit and has then been taken as a definition of quantum chaos. However, the problem of connecting the spectral statistics with more intuitive dynamical properties of the system remained open. 
 
Over the last decade the problem of characterising chaos in quantum systems received a renewed interest due to seminal results coming for the study of black holes~\cite{HaPr07, SeSu08} and connecting quantum many-body chaos with the \emph{scrambling} of quantum information. In turn, this renaissance also produced new discoveries concerning chaos in extended quantum many-body systems on the lattice~\cite{KLP, Chalker, Chalker2, Chalker3, BKP:kickedIsing, RP20, Chalker4, Chan, JHN:OperatorEntanglement, BP:tripartite, BKP:OEsolitons,BKP:OEergodicandmixing, AlDM19, ClLa20, Keyserlingk, Nahum:operatorspreadingRU, Keselman} and lead to the introduction of useful minimal models like local random unitary circuits~\cite{RandomCircuitsEnt, Chalker} and dual-unitary circuits~\cite{BKP:dualunitary}. For some of these systems it has been possible to compute measures of the spectral statistics~\cite{KLP, RP20, BKP:kickedIsing, Chalker2, Chalker3}, proving that they indeed follow the predictions of Random Matrix Theory. Importantly, however, it has been realised that in generic extended systems with local interactions this happens only for energy levels smaller than a certain scale $E_{\rm th}$ --- known as Thouless Energy --- which bares information on the spatial structure. This energy scale (or the associated Thouless time $\tau_{\rm th}=\hbar/E_{\rm th}$) is believed to display different scalings with the system size depending on the conservation laws of the system. 

In the recent comeback of quantum chaos an important role has been played by driven systems, as they furnish a simpler modelisation of many interesting dynamical phenomena~\cite{RandomCircuitsEnt, JHN:OperatorEntanglement, BP:tripartite, Keyserlingk, Nahum:operatorspreadingRU}. For these systems, in the generic instance of aperiodic driving, the spectral statistics is not well defined (their time-evolution operator is time-dependent) and the chaotic regime has been identified by looking at some features of the quantum many-body dynamics --- seeking a quantum many-body analogue of the butterfly effect. Some of the most studied features have been the spreading of support of local operators (measured, e.g., by out-of-time-ordered correlators~\cite{ShSt14_bh, RoSS15, MaSS16}), the growth of complexity in the classical simulations of the dynamics~\cite{PZ07}, and the scrambling of quantum information~\cite{HQRY:tripartiteinfo}. However, even though all these features are connected to an idea of  `dynamical complexity', they provide different information.  It is unclear what is the minimal set of these features, if any, that a system has to display to be considered chaotic.

In this Letter we follow a different route and regard as ``chaotic'' those driven systems where the time-evolution operator acquires random matrix spectral correlations after a certain initial transient~\cite{SSS:genSFF, genSFF}. This is a direct generalisation of the traditional definition of quantum chaos and the transient is naturally interpreted as the Thouless time~\cite{genSFF}. We present a family of local quantum circuits with random fields in a fixed direction. In these systems --- that have no semi-classical limit --- the time-dependent spectral correlations can be characterised exactly.  In particular, we compute the squared magnitude of the trace of the evolution operator --- which we dub Generalised Spectral Form Factor (GSFF) --- and show that at the leading order in time it is fully specified by the two-point dynamical correlation functions of local operators in the infinite temperature state. This provides an unprecedented direct link between spectral properties and local physics. We use this result to show that the regime of quantum chaos coincides with the ergodic and mixing one (where all dynamical correlations decay in time), and to elucidate the connection between conservation laws and scaling of the Thouless time with the system size.

More specifically, we consider a chain of length $L$ with $2L$ qubits placed at integer and half-integer indexed sites. Thus, the Hilbert space of the system is 
$\mathcal{H}= (\mathbb{C}^2)^{\otimes 2L}$. The time evolution is governed by a brickwork-like local quantum circuit, consisting of unitary matrices (gates) acting on two neighbouring spins (with periodic boundary conditions). We consider the case where the gates are different at each space-time point and represent the time evolution as 
\begin{align}
&\mathbb U(t)=
\begin{tikzpicture}[baseline=(current  bounding  box.center), scale=0.5]
\foreach \i in {3,...,13}{
\draw[gray, dashed] (-12.5+\i,-2.1) -- (-12.5+\i,4.3);
}
\foreach \i in {-1,...,5}{
\draw[gray, dashed] (-9.75,3-\i) -- (.75,3-\i);
}
\foreach \i in{1.5,2.5,3.5}{
\draw[thick] (0.5,2*\i-0.5-3.5) arc (45:-90:0.15);
\draw[thick] (-10+0.5+0,2*\i-0.5-3.5) arc (-135:-270:0.15);
}
\foreach \i in{1.5,2.5}
{
\draw[ thick] (0.5,1+2*\i-0.5-3.5) arc (-45:90:0.15);
\draw[ thick] (-10+0.5,1+2*\i-0.5-3.5) arc (135:270:0.15);
}
\foreach \i in {0,1,2}{
\Text[x=1.25,y=-2+2*\i]{\scriptsize$\i$}
}
\foreach \i in {1,3}{
\Text[x=1.25,y=-2+\i]{\small$\frac{\i}{2}$}
}
\foreach \i in {1,3}{
\Text[x=-7.5+\i+1,y=-2.6]{\small$\frac{\i}{2}$}
}
\foreach \i in {0,1}{
\Text[x=-7.5+2*\i+1,y=-2.6]{\scriptsize${\i}$}
}
\foreach \jj in {0}{
\foreach \i in {0}{
\draw[thick] (-.5-2*\i,-1+\jj) -- (0.525-2*\i,0.025+\jj);
\draw[thick] (-0.525-2*\i,0.025+\jj) -- (0.5-2*\i,-1+\jj);
\draw[thick, fill=myblue1, rounded corners=2pt] (-0.25-2*\i,-0.25+\jj) rectangle (.25-2*\i,-0.75+\jj);
\draw[thick] (-2*\i,-0.35+\jj) -- (-2*\i+0.15,-0.35+\jj) -- (-2*\i+0.15,-0.5+\jj);%
}
\foreach \i in {1}{
\draw[thick] (-.5-2*\i,-1+\jj) -- (0.525-2*\i,0.025+\jj);
\draw[thick] (-0.525-2*\i,0.025+\jj) -- (0.5-2*\i,-1+\jj);
\draw[thick, fill=myblue4, rounded corners=2pt] (-0.25-2*\i,-0.25+\jj) rectangle (.25-2*\i,-0.75+\jj);
\draw[thick] (-2*\i,-0.35+\jj) -- (-2*\i+0.15,-0.35+\jj) -- (-2*\i+0.15,-0.5+\jj);%
}
\foreach \i in {2}{
\draw[thick] (-.5-2*\i,-1+\jj) -- (0.525-2*\i,0.025+\jj);
\draw[thick] (-0.525-2*\i,0.025+\jj) -- (0.5-2*\i,-1+\jj);
\draw[thick, fill=myblue3, rounded corners=2pt] (-0.25-2*\i,-0.25+\jj) rectangle (.25-2*\i,-0.75+\jj);
\draw[thick] (-2*\i,-0.35+\jj) -- (-2*\i+0.15,-0.35+\jj) -- (-2*\i+0.15,-0.5+\jj);%
}
\foreach \i in {3}{
\draw[thick] (-.5-2*\i,-1+\jj) -- (0.525-2*\i,0.025+\jj);
\draw[thick] (-0.525-2*\i,0.025+\jj) -- (0.5-2*\i,-1+\jj);
\draw[thick, fill=myblue7, rounded corners=2pt] (-0.25-2*\i,-0.25+\jj) rectangle (.25-2*\i,-0.75+\jj);
\draw[thick] (-2*\i,-0.35+\jj) -- (-2*\i+0.15,-0.35+\jj) -- (-2*\i+0.15,-0.5+\jj);%
}
\foreach \i in {4}{
\draw[thick] (-.5-2*\i,-1+\jj) -- (0.525-2*\i,0.025+\jj);
\draw[thick] (-0.525-2*\i,0.025+\jj) -- (0.5-2*\i,-1+\jj);
\draw[thick, fill=myblue6, rounded corners=2pt] (-0.25-2*\i,-0.25+\jj) rectangle (.25-2*\i,-0.75+\jj);
\draw[thick] (-2*\i,-0.35+\jj) -- (-2*\i+0.15,-0.35+\jj) -- (-2*\i+0.15,-0.5+\jj);%
}
}
\foreach \jj in {2}{
\foreach \i in {0}{
\draw[thick] (-.5-2*\i,-1+\jj) -- (0.525-2*\i,0.025+\jj);
\draw[thick] (-0.525-2*\i,0.025+\jj) -- (0.5-2*\i,-1+\jj);
\draw[thick, fill=myblue9, rounded corners=2pt] (-0.25-2*\i,-0.25+\jj) rectangle (.25-2*\i,-0.75+\jj);
\draw[thick] (-2*\i,-0.35+\jj) -- (-2*\i+0.15,-0.35+\jj) -- (-2*\i+0.15,-0.5+\jj);%
}
\foreach \i in {1}{
\draw[thick] (-.5-2*\i,-1+\jj) -- (0.525-2*\i,0.025+\jj);
\draw[thick] (-0.525-2*\i,0.025+\jj) -- (0.5-2*\i,-1+\jj);
\draw[thick, fill=myblue10, rounded corners=2pt] (-0.25-2*\i,-0.25+\jj) rectangle (.25-2*\i,-0.75+\jj);
\draw[thick] (-2*\i,-0.35+\jj) -- (-2*\i+0.15,-0.35+\jj) -- (-2*\i+0.15,-0.5+\jj);%
}
\foreach \i in {2}{
\draw[thick] (-.5-2*\i,-1+\jj) -- (0.525-2*\i,0.025+\jj);
\draw[thick] (-0.525-2*\i,0.025+\jj) -- (0.5-2*\i,-1+\jj);
\draw[thick, fill=myblue2, rounded corners=2pt] (-0.25-2*\i,-0.25+\jj) rectangle (.25-2*\i,-0.75+\jj);
\draw[thick] (-2*\i,-0.35+\jj) -- (-2*\i+0.15,-0.35+\jj) -- (-2*\i+0.15,-0.5+\jj);%
}
\foreach \i in {3}{
\draw[thick] (-.5-2*\i,-1+\jj) -- (0.525-2*\i,0.025+\jj);
\draw[thick] (-0.525-2*\i,0.025+\jj) -- (0.5-2*\i,-1+\jj);
\draw[thick, fill=myblue6, rounded corners=2pt] (-0.25-2*\i,-0.25+\jj) rectangle (.25-2*\i,-0.75+\jj);
\draw[thick] (-2*\i,-0.35+\jj) -- (-2*\i+0.15,-0.35+\jj) -- (-2*\i+0.15,-0.5+\jj);%
}
\foreach \i in {4}{
\draw[thick] (-.5-2*\i,-1+\jj) -- (0.525-2*\i,0.025+\jj);
\draw[thick] (-0.525-2*\i,0.025+\jj) -- (0.5-2*\i,-1+\jj);
\draw[thick, fill=myblue7, rounded corners=2pt] (-0.25-2*\i,-0.25+\jj) rectangle (.25-2*\i,-0.75+\jj);
\draw[thick] (-2*\i,-0.35+\jj) -- (-2*\i+0.15,-0.35+\jj) -- (-2*\i+0.15,-0.5+\jj);%
}
}
\foreach \jj in {4}{
\foreach \i in {0}{
\draw[thick] (-.5-2*\i,-1+\jj) -- (0.525-2*\i,0.025+\jj);
\draw[thick] (-0.525-2*\i,0.025+\jj) -- (0.5-2*\i,-1+\jj);
\draw[thick, fill=myblue4, rounded corners=2pt] (-0.25-2*\i,-0.25+\jj) rectangle (.25-2*\i,-0.75+\jj);
\draw[thick] (-2*\i,-0.35+\jj) -- (-2*\i+0.15,-0.35+\jj) -- (-2*\i+0.15,-0.5+\jj);%
}
\foreach \i in {1}{
\draw[thick] (-.5-2*\i,-1+\jj) -- (0.525-2*\i,0.025+\jj);
\draw[thick] (-0.525-2*\i,0.025+\jj) -- (0.5-2*\i,-1+\jj);
\draw[thick, fill=myblue3, rounded corners=2pt] (-0.25-2*\i,-0.25+\jj) rectangle (.25-2*\i,-0.75+\jj);
\draw[thick] (-2*\i,-0.35+\jj) -- (-2*\i+0.15,-0.35+\jj) -- (-2*\i+0.15,-0.5+\jj);%
}
\foreach \i in {2}{
\draw[thick] (-.5-2*\i,-1+\jj) -- (0.525-2*\i,0.025+\jj);
\draw[thick] (-0.525-2*\i,0.025+\jj) -- (0.5-2*\i,-1+\jj);
\draw[thick, fill=myblue2, rounded corners=2pt] (-0.25-2*\i,-0.25+\jj) rectangle (.25-2*\i,-0.75+\jj);
\draw[thick] (-2*\i,-0.35+\jj) -- (-2*\i+0.15,-0.35+\jj) -- (-2*\i+0.15,-0.5+\jj);%
}
\foreach \i in {3}{
\draw[thick] (-.5-2*\i,-1+\jj) -- (0.525-2*\i,0.025+\jj);
\draw[thick] (-0.525-2*\i,0.025+\jj) -- (0.5-2*\i,-1+\jj);
\draw[thick, fill=myblue5, rounded corners=2pt] (-0.25-2*\i,-0.25+\jj) rectangle (.25-2*\i,-0.75+\jj);
\draw[thick] (-2*\i,-0.35+\jj) -- (-2*\i+0.15,-0.35+\jj) -- (-2*\i+0.15,-0.5+\jj);%
}
\foreach \i in {4}{
\draw[thick] (-.5-2*\i,-1+\jj) -- (0.525-2*\i,0.025+\jj);
\draw[thick] (-0.525-2*\i,0.025+\jj) -- (0.5-2*\i,-1+\jj);
\draw[thick, fill=myblue10, rounded corners=2pt] (-0.25-2*\i,-0.25+\jj) rectangle (.25-2*\i,-0.75+\jj);
\draw[thick] (-2*\i,-0.35+\jj) -- (-2*\i+0.15,-0.35+\jj) -- (-2*\i+0.15,-0.5+\jj);%
}
}
\foreach \jj[evaluate=\jj as \j using -2*(ceil(\jj/2)-\jj/2)] in {0}{
\foreach \i in {1}
{
\draw[thick] (.5-2*\i-1*\j,-2-1*\jj) -- (1-2*\i-1*\j,-1.5-\jj);
\draw[thick] (1-2*\i-1*\j,-1.5-1*\jj) -- (1.5-2*\i-1*\j,-2-\jj);
\draw[thick] (.5-2*\i-1*\j,-1-1*\jj) -- (1-2*\i-1*\j,-1.5-\jj);
\draw[thick] (1-2*\i-1*\j,-1.5-1*\jj) -- (1.5-2*\i-1*\j,-1-\jj);
\draw[thick, fill=myblue5, rounded corners=2pt] (0.75-2*\i-1*\j,-1.75-\jj) rectangle (1.25-2*\i-1*\j,-1.25-\jj);
\draw[thick] (-2*\i+1,-1.35-\jj) -- (-2*\i+1.15,-1.35-\jj) -- (-2*\i+1.15,-1.5-\jj);%
}
\foreach \i in {2}
{
\draw[thick] (.5-2*\i-1*\j,-2-1*\jj) -- (1-2*\i-1*\j,-1.5-\jj);
\draw[thick] (1-2*\i-1*\j,-1.5-1*\jj) -- (1.5-2*\i-1*\j,-2-\jj);
\draw[thick] (.5-2*\i-1*\j,-1-1*\jj) -- (1-2*\i-1*\j,-1.5-\jj);
\draw[thick] (1-2*\i-1*\j,-1.5-1*\jj) -- (1.5-2*\i-1*\j,-1-\jj);
\draw[thick, fill=myblue4, rounded corners=2pt] (0.75-2*\i-1*\j,-1.75-\jj) rectangle (1.25-2*\i-1*\j,-1.25-\jj);
\draw[thick] (-2*\i+1,-1.35-\jj) -- (-2*\i+1.15,-1.35-\jj) -- (-2*\i+1.15,-1.5-\jj);%
}
\foreach \i in {3}
{
\draw[thick] (.5-2*\i-1*\j,-2-1*\jj) -- (1-2*\i-1*\j,-1.5-\jj);
\draw[thick] (1-2*\i-1*\j,-1.5-1*\jj) -- (1.5-2*\i-1*\j,-2-\jj);
\draw[thick] (.5-2*\i-1*\j,-1-1*\jj) -- (1-2*\i-1*\j,-1.5-\jj);
\draw[thick] (1-2*\i-1*\j,-1.5-1*\jj) -- (1.5-2*\i-1*\j,-1-\jj);
\draw[thick, fill=myblue3, rounded corners=2pt] (0.75-2*\i-1*\j,-1.75-\jj) rectangle (1.25-2*\i-1*\j,-1.25-\jj);
\draw[thick] (-2*\i+1,-1.35-\jj) -- (-2*\i+1.15,-1.35-\jj) -- (-2*\i+1.15,-1.5-\jj);%
}
\foreach \i in {4}
{
\draw[thick] (.5-2*\i-1*\j,-2-1*\jj) -- (1-2*\i-1*\j,-1.5-\jj);
\draw[thick] (1-2*\i-1*\j,-1.5-1*\jj) -- (1.5-2*\i-1*\j,-2-\jj);
\draw[thick] (.5-2*\i-1*\j,-1-1*\jj) -- (1-2*\i-1*\j,-1.5-\jj);
\draw[thick] (1-2*\i-1*\j,-1.5-1*\jj) -- (1.5-2*\i-1*\j,-1-\jj);
\draw[thick, fill=myblue2, rounded corners=2pt] (0.75-2*\i-1*\j,-1.75-\jj) rectangle (1.25-2*\i-1*\j,-1.25-\jj);
\draw[thick] (-2*\i+1,-1.35-\jj) -- (-2*\i+1.15,-1.35-\jj) -- (-2*\i+1.15,-1.5-\jj);%
}
\foreach \i in {5}
{
\draw[thick] (.5-2*\i-1*\j,-2-1*\jj) -- (1-2*\i-1*\j,-1.5-\jj);
\draw[thick] (1-2*\i-1*\j,-1.5-1*\jj) -- (1.5-2*\i-1*\j,-2-\jj);
\draw[thick] (.5-2*\i-1*\j,-1-1*\jj) -- (1-2*\i-1*\j,-1.5-\jj);
\draw[thick] (1-2*\i-1*\j,-1.5-1*\jj) -- (1.5-2*\i-1*\j,-1-\jj);
\draw[thick, fill=myblue1, rounded corners=2pt] (0.75-2*\i-1*\j,-1.75-\jj) rectangle (1.25-2*\i-1*\j,-1.25-\jj);
\draw[thick] (-2*\i+1,-1.35-\jj) -- (-2*\i+1.15,-1.35-\jj) -- (-2*\i+1.15,-1.5-\jj);%
}
}
\foreach \jj[evaluate=\jj as \j using -2*(ceil(\jj/2)-\jj/2)] in {-2}{
\foreach \i in {1}
{
\draw[thick] (.5-2*\i-1*\j,-2-1*\jj) -- (1-2*\i-1*\j,-1.5-\jj);
\draw[thick] (1-2*\i-1*\j,-1.5-1*\jj) -- (1.5-2*\i-1*\j,-2-\jj);
\draw[thick] (.5-2*\i-1*\j,-1-1*\jj) -- (1-2*\i-1*\j,-1.5-\jj);
\draw[thick] (1-2*\i-1*\j,-1.5-1*\jj) -- (1.5-2*\i-1*\j,-1-\jj);
\draw[thick, fill=myblue8, rounded corners=2pt] (0.75-2*\i-1*\j,-1.75-\jj) rectangle (1.25-2*\i-1*\j,-1.25-\jj);
\draw[thick] (-2*\i+1,-1.35-\jj) -- (-2*\i+1.15,-1.35-\jj) -- (-2*\i+1.15,-1.5-\jj);%
}
\foreach \i in {2}
{
\draw[thick] (.5-2*\i-1*\j,-2-1*\jj) -- (1-2*\i-1*\j,-1.5-\jj);
\draw[thick] (1-2*\i-1*\j,-1.5-1*\jj) -- (1.5-2*\i-1*\j,-2-\jj);
\draw[thick] (.5-2*\i-1*\j,-1-1*\jj) -- (1-2*\i-1*\j,-1.5-\jj);
\draw[thick] (1-2*\i-1*\j,-1.5-1*\jj) -- (1.5-2*\i-1*\j,-1-\jj);
\draw[thick, fill=myblue10, rounded corners=2pt] (0.75-2*\i-1*\j,-1.75-\jj) rectangle (1.25-2*\i-1*\j,-1.25-\jj);
\draw[thick] (-2*\i+1,-1.35-\jj) -- (-2*\i+1.15,-1.35-\jj) -- (-2*\i+1.15,-1.5-\jj);%
}
\foreach \i in {3}
{
\draw[thick] (.5-2*\i-1*\j,-2-1*\jj) -- (1-2*\i-1*\j,-1.5-\jj);
\draw[thick] (1-2*\i-1*\j,-1.5-1*\jj) -- (1.5-2*\i-1*\j,-2-\jj);
\draw[thick] (.5-2*\i-1*\j,-1-1*\jj) -- (1-2*\i-1*\j,-1.5-\jj);
\draw[thick] (1-2*\i-1*\j,-1.5-1*\jj) -- (1.5-2*\i-1*\j,-1-\jj);
\draw[thick, fill=myblue6, rounded corners=2pt] (0.75-2*\i-1*\j,-1.75-\jj) rectangle (1.25-2*\i-1*\j,-1.25-\jj);
\draw[thick] (-2*\i+1,-1.35-\jj) -- (-2*\i+1.15,-1.35-\jj) -- (-2*\i+1.15,-1.5-\jj);%
}
\foreach \i in {4}
{
\draw[thick] (.5-2*\i-1*\j,-2-1*\jj) -- (1-2*\i-1*\j,-1.5-\jj);
\draw[thick] (1-2*\i-1*\j,-1.5-1*\jj) -- (1.5-2*\i-1*\j,-2-\jj);
\draw[thick] (.5-2*\i-1*\j,-1-1*\jj) -- (1-2*\i-1*\j,-1.5-\jj);
\draw[thick] (1-2*\i-1*\j,-1.5-1*\jj) -- (1.5-2*\i-1*\j,-1-\jj);
\draw[thick, fill=myblue9, rounded corners=2pt] (0.75-2*\i-1*\j,-1.75-\jj) rectangle (1.25-2*\i-1*\j,-1.25-\jj);
\draw[thick] (-2*\i+1,-1.35-\jj) -- (-2*\i+1.15,-1.35-\jj) -- (-2*\i+1.15,-1.5-\jj);%
}
\foreach \i in {5}
{
\draw[thick] (.5-2*\i-1*\j,-2-1*\jj) -- (1-2*\i-1*\j,-1.5-\jj);
\draw[thick] (1-2*\i-1*\j,-1.5-1*\jj) -- (1.5-2*\i-1*\j,-2-\jj);
\draw[thick] (.5-2*\i-1*\j,-1-1*\jj) -- (1-2*\i-1*\j,-1.5-\jj);
\draw[thick] (1-2*\i-1*\j,-1.5-1*\jj) -- (1.5-2*\i-1*\j,-1-\jj);
\draw[thick, fill=myblue1, rounded corners=2pt] (0.75-2*\i-1*\j,-1.75-\jj) rectangle (1.25-2*\i-1*\j,-1.25-\jj);
\draw[thick] (-2*\i+1,-1.35-\jj) -- (-2*\i+1.15,-1.35-\jj) -- (-2*\i+1.15,-1.5-\jj);%
}
}
\foreach \jj[evaluate=\jj as \j using -2*(ceil(\jj/2)-\jj/2)] in {-4}{
\foreach \i in {1}
{
\draw[thick] (.5-2*\i-1*\j,-2-1*\jj) -- (1-2*\i-1*\j,-1.5-\jj);
\draw[thick] (1-2*\i-1*\j,-1.5-1*\jj) -- (1.5-2*\i-1*\j,-2-\jj);
\draw[thick] (.5-2*\i-1*\j,-1-1*\jj) -- (1-2*\i-1*\j,-1.5-\jj);
\draw[thick] (1-2*\i-1*\j,-1.5-1*\jj) -- (1.5-2*\i-1*\j,-1-\jj);
\draw[thick, fill=myblue6, rounded corners=2pt] (0.75-2*\i-1*\j,-1.75-\jj) rectangle (1.25-2*\i-1*\j,-1.25-\jj);
\draw[thick] (-2*\i+1,-1.35-\jj) -- (-2*\i+1.15,-1.35-\jj) -- (-2*\i+1.15,-1.5-\jj);%
}
\foreach \i in {2}
{
\draw[thick] (.5-2*\i-1*\j,-2-1*\jj) -- (1-2*\i-1*\j,-1.5-\jj);
\draw[thick] (1-2*\i-1*\j,-1.5-1*\jj) -- (1.5-2*\i-1*\j,-2-\jj);
\draw[thick] (.5-2*\i-1*\j,-1-1*\jj) -- (1-2*\i-1*\j,-1.5-\jj);
\draw[thick] (1-2*\i-1*\j,-1.5-1*\jj) -- (1.5-2*\i-1*\j,-1-\jj);
\draw[thick, fill=myblue3, rounded corners=2pt] (0.75-2*\i-1*\j,-1.75-\jj) rectangle (1.25-2*\i-1*\j,-1.25-\jj);
\draw[thick] (-2*\i+1,-1.35-\jj) -- (-2*\i+1.15,-1.35-\jj) -- (-2*\i+1.15,-1.5-\jj);%
}
\foreach \i in {3}
{
\draw[thick] (.5-2*\i-1*\j,-2-1*\jj) -- (1-2*\i-1*\j,-1.5-\jj);
\draw[thick] (1-2*\i-1*\j,-1.5-1*\jj) -- (1.5-2*\i-1*\j,-2-\jj);
\draw[thick] (.5-2*\i-1*\j,-1-1*\jj) -- (1-2*\i-1*\j,-1.5-\jj);
\draw[thick] (1-2*\i-1*\j,-1.5-1*\jj) -- (1.5-2*\i-1*\j,-1-\jj);
\draw[thick, fill=myblue7, rounded corners=2pt] (0.75-2*\i-1*\j,-1.75-\jj) rectangle (1.25-2*\i-1*\j,-1.25-\jj);
\draw[thick] (-2*\i+1,-1.35-\jj) -- (-2*\i+1.15,-1.35-\jj) -- (-2*\i+1.15,-1.5-\jj);%
}
\foreach \i in {4}
{
\draw[thick] (.5-2*\i-1*\j,-2-1*\jj) -- (1-2*\i-1*\j,-1.5-\jj);
\draw[thick] (1-2*\i-1*\j,-1.5-1*\jj) -- (1.5-2*\i-1*\j,-2-\jj);
\draw[thick] (.5-2*\i-1*\j,-1-1*\jj) -- (1-2*\i-1*\j,-1.5-\jj);
\draw[thick] (1-2*\i-1*\j,-1.5-1*\jj) -- (1.5-2*\i-1*\j,-1-\jj);
\draw[thick, fill=myblue2, rounded corners=2pt] (0.75-2*\i-1*\j,-1.75-\jj) rectangle (1.25-2*\i-1*\j,-1.25-\jj);
\draw[thick] (-2*\i+1,-1.35-\jj) -- (-2*\i+1.15,-1.35-\jj) -- (-2*\i+1.15,-1.5-\jj);%
}
\foreach \i in {5}
{
\draw[thick] (.5-2*\i-1*\j,-2-1*\jj) -- (1-2*\i-1*\j,-1.5-\jj);
\draw[thick] (1-2*\i-1*\j,-1.5-1*\jj) -- (1.5-2*\i-1*\j,-2-\jj);
\draw[thick] (.5-2*\i-1*\j,-1-1*\jj) -- (1-2*\i-1*\j,-1.5-\jj);
\draw[thick] (1-2*\i-1*\j,-1.5-1*\jj) -- (1.5-2*\i-1*\j,-1-\jj);
\draw[thick, fill=myblue7, rounded corners=2pt] (0.75-2*\i-1*\j,-1.75-\jj) rectangle (1.25-2*\i-1*\j,-1.25-\jj);
\draw[thick] (-2*\i+1,-1.35-\jj) -- (-2*\i+1.15,-1.35-\jj) -- (-2*\i+1.15,-1.5-\jj);%
}
}
\Text[x=-7.8,y=-2.6]{$\cdots$}
\Text[x=-2,y=-2.6]{$\cdots$}
\Text[x=0,y=-2.6]{\small $\tfrac{L}{2}-1$}
\Text[x=-9.8,y=-2.6]{\small $-\tfrac{L}{2}$}
\Text[x=-4,y=-3.5]{$x$}
\Text[x=1.25,y=4]{\small $t$}
\Text[x=1.25,y=3.2]{$\vdots$}
\end{tikzpicture},
\label{eq:quantumcircuit}
\end{align}
where we depicted two-site gates as
\be
\label{eq:Ugate}
U_{x,\tau}=\begin{tikzpicture}[baseline=(current  bounding  box.center), scale=.7]
\draw[ thick] (-4.25,0.5) -- (-3.25,-0.5);
\draw[ thick] (-4.25,-0.5) -- (-3.25,0.5);
\draw[ thick, fill=myblue, rounded corners=2pt] (-4,0.25) rectangle (-3.5,-0.25);
\draw[thick] (-3.75,0.15) -- (-3.75+0.15,0.15) -- (-3.75+0.15,0);
\Text[x=-4.25,y=-0.75]{}
\end{tikzpicture},\quad
x\in\frac{1}{2}\mathbb Z_{2L},\; \tau\in\frac{1}{2}\mathbb Z_{2t},\quad x+ \tau \in\mathbb Z,
\ee
and different colours denote different matrices. Note that we adopt the convention of time running upwards.
We remark that this setting is in fact quite general. It can be thought of as generated by a disordered local Hamiltonian, which changes at each half-integer time due to some external driving. This formulation of quantum evolution is widely used, for instance, in the context of quantum simulators~\cite{google}.

The main quantity of interest for this paper is the GSFF defined as
\be
K_{g}(t) \equiv \braket{|{\rm tr_{sector}}\, \mathbb U(t)|^2 }, \qquad\qquad t>0\,.
\label{eq:GSFF}
\ee
where the trace is reduced to a common eigenspace of $\mathbb U(t)$ and all its commuting symmetries, and $\braket{\cdot}$ denotes an average of some sort (either a moving time-average or an average over an ensemble of similar systems). Such an average is necessary because the distribution of $|{\rm tr}\, \mathbb U(t)|^2$ in an ensemble of systems does \emph{not} generically become infinitely peaked even in the limit of infinitely many degrees of freedom~\cite{nonSA}. 
From the definition \eqref{eq:GSFF}, we see that $K_g(t)$ with unrestricted trace can be interpreted as the survival probability (or Loschmidt echo) for a random initial state, which is another 
chaos indicator~\cite{CH,LEReview}. 

As mentioned before, here we regard $\mathbb U(t)$ as ``chaotic'' if there exist a scale $\tau_{\rm th}$ such that
\be
K_g(t) \simeq 1=\mathbb E_{\rm CUE}[|{\rm tr}\, U|^2]\,,\qquad \text{for}\quad t \gg \tau_{\rm th},
\label{eq:QCC}
\ee
where $\simeq$ denotes asymptotic equality in the leading and possibly subleading order in $\tau_{\rm th}/t$ and, since our system is not time-reversal invariant, we considered the average over the circular unitary ensemble (CUE)~\cite{Mehtabook, Haake}. Eq.~\eqref{eq:QCC} allows us to illustrate that \eqref{eq:GSFF} is very different from the (conventional) spectral form factor considered in periodically driven systems~\cite{Haake, COE:SFF, Mehtabook}. Indeed, here the full time-evolution operator $\mathbb U(t)$ is expected to behave like a random matrix at large times, while in the periodically driven case $\mathbb U(t)$ behaves as the $t$-th power of a random matrix. This means that $K_g(t)$ should be compared with the conventional spectral form factor at time $t=1$ and that the asymptotic value in \eqref{eq:QCC} has nothing to do with the asymptotic value (exponentially large in the volume) at which the conventional spectral form factor relaxes for times larger than the inverse level spacing (which are, again, exponentially large in the volume).

In our setting \eqref{eq:quantumcircuit}, a natural way to produce ensembles of different systems is to introduce noise in the local gate $U_{x,t}$. Since we are interested in generic drivings, we look at time-dependent noise and, to avoid any bias, we choose it to be independently distributed in space and time. Specifically, following~\cite{KBP:CorrelationsPerturbed}, we take random gates of the form 
\be
U_{x,\tau} = (e^{i \phi_{x,\tau} \sigma^z} \otimes e^{i \phi_{\scriptscriptstyle x+{1}/{2},\tau} \sigma^{z}})\cdot U,
\label{eq:Ubar}
\ee
where $U$ is a fixed $U(4)$ matrix, $\phi_{x,t}$ are independent random variables uniformly distributed over $[-\pi,\pi]$ and $\otimes$ denotes the tensor product between two neighbouring sites. From the physical point of view the choice \eqref{eq:Ubar} describes a homogeneous spin-1/2 chain where the time evolution is periodic but each spin is subject to white noise produced by a random magnetic field in the z-direction~\cite{notecit}.

For gates of the form \eqref{eq:Ubar}, the average $\braket{\cdot}$ in \eqref{eq:GSFF} can be implemented locally by placing $\mathbb U(t)^\dag$ on top of $\mathbb U(t)$ in such a way that each gate lies on top of its conjugate (i.e. `folding' the circuit, see, e.g., Ref.~\cite{KBP:CorrelationsPerturbed}). 
Specifically, the average projects to a subspace spanned by diagonal operators ($\ket{\mcirc}\equiv \ket{\1}$ and $\ket{\mcircf}= \ket{\sigma^z}$) and allows us to write \eqref{eq:GSFF} as 
\be
K_g(t) = 
\begin{tikzpicture}[baseline=(current  bounding  box.center), scale=0.5]
\def\eps{0};
\def\shift{11}
\def\shifty{-2.5}
\Text[x=\shift-3, y=-4]{}
\foreach \i in {1,...,4}{
\draw[ thick, opacity =0.2,dashed] (2*\i+2-1.5+0.25,-2.5-0.1) -- (2*\i+2-1.5+0.25,3.5-0.1);
\draw[ thick, opacity =0.2,dashed] (2*\i+2-0.5-0.25,-2.5-0.1) -- (2*\i+2-0.5-0.25,3.5-0.1);
}
\foreach \i in{3,...,4}
{
\draw[ thick] (\shift-1.5,2*\i-0.5-3.5+\shifty) arc (45:-90:0.15);
\draw[ thick] (\shift-10+0.5+0,2*\i-0.5-3.5+\shifty) arc (-135:-270:0.15);}
\foreach \i in{2,...,4}
{
\draw[ thick] (\shift-1.5,1+2*\i-0.5-3.5+\shifty) arc (-45:90:0.15);
\draw[ thick] (\shift-10+0.5,1+2*\i-0.5-3.5+\shifty) arc (135:270:0.15);
}
\foreach \i in {1,...,4}
{
\draw[ thick] (\shift-.5-2*\i,1+\shifty) -- (\shift+0.5-2*\i,0+\shifty);
\draw[ thick] (\shift-0.5-2*\i,0+\shifty) -- (\shift+0.5-2*\i,1+\shifty);
\draw[ thick, fill=myY, rounded corners=2pt] (\shift-0.25-2*\i,0.25+\shifty) rectangle (\shift+.25-2*\i,0.75+\shifty);
\draw[thick] (\shift-2*\i,0.65+\shifty) -- (\shift+.15-2*\i,.65+\shifty) -- (\shift+.15-2*\i,0.5+\shifty);

\draw[ thick] (2*\i+2-1.5,3.5) arc (135:-0:0.15);
\draw[ thick] (2*\i+2-2.5,3.5) arc (-325:-180:0.15);
\draw[ thick] (2*\i+2-1.5,-2.5) arc (-45:180:-0.15);
\draw[ thick] (2*\i+2-0.5,-2.5) arc (45:-180:0.15);
}
\foreach \i in {2,...,5}
{
\draw[ thick] (\shift+.5-2*\i,6+\shifty) -- (\shift+1-2*\i,5.5+\shifty);
\draw[ thick] (\shift+1.5-2*\i,6+\shifty) -- (\shift+1-2*\i,5.5+\shifty);
}
\foreach \jj[evaluate=\jj as \j using -2*(ceil(\jj/2)-\jj/2)] in {0,...,3}
\foreach \i in {2,...,5}
{
\draw[ thick] (\shift+.5-2*\i-1*\j,2+1*\jj+\shifty) -- (\shift+1-2*\i-1*\j,1.5+\jj+\shifty);
\draw[ thick] (\shift+1-2*\i-1*\j,1.5+1*\jj+\shifty) -- (\shift+1.5-2*\i-1*\j,2+\jj+\shifty);
}
\foreach \i in {1,...,4}
{
\draw[ thick] (\shift-.5-2*\i,1+\shifty) -- (\shift+0.5-2*\i,0+\shifty);
\draw[ thick] (\shift-0.5-2*\i,0+\shifty) -- (\shift+0.5-2*\i,1+\shifty);
\draw[ thick, fill=myY, rounded corners=2pt] (\shift-0.25-2*\i,0.25+\shifty) rectangle (\shift+.25-2*\i,0.75+\shifty);
\draw[thick] (\shift-2*\i,0.65+\shifty) -- (\shift+.15-2*\i,.65+\shifty) -- (\shift+.15-2*\i,0.5+\shifty);
}
\foreach \jj[evaluate=\jj as \j using -2*(ceil(\jj/2)-\jj/2)] in {0,...,4}
\foreach \i in {2,...,5}
{
\draw[ thick] (\shift+.5-2*\i-1*\j,1+1*\jj+\shifty) -- (\shift+1-2*\i-1*\j,1.5+\jj+\shifty);
\draw[ thick] (\shift+1-2*\i-1*\j,1.5+1*\jj+\shifty) -- (\shift+1.5-2*\i-1*\j,1+\jj+\shifty);
\draw[ thick, fill=myY, rounded corners=2pt] (\shift+0.75-2*\i-1*\j,1.75+\jj+\shifty) rectangle (\shift+1.25-2*\i-1*\j,1.25+\jj+\shifty);
\draw[thick] (\shift+1-2*\i-1*\j,1.65+1*\jj+\shifty) -- (\shift+1.15-2*\i-1*\j,1.65+1*\jj+\shifty) -- (\shift+1.15-2*\i-1*\j,1.5+1*\jj+\shifty);
} 
\end{tikzpicture}\, ,
\label{eq:SFF0}
\vspace{-5mm}
\ee
where top and bottom wires at the same position are connected because of the trace.
Above we introduced the non-unitary `averaged gate', written in the local basis 
$\{\ket{\mcirc},\ket{\mcircf}\}$:
\begin{equation}
w = \begin{tikzpicture}[baseline=(current  bounding  box.center), scale=.8]
\def\eps{0.5};
\Wreduced{-3.75}{0};
\Text[x=-3.75,y=-0.6]{}
\end{tikzpicture}
= \begin{pmatrix}
    1 & 0 & 0 & 0\\
0 & \varepsilon_1 & a & b\\
0 & c & \varepsilon_2 & d\\
0 & e & f &  g
    \end{pmatrix}\!,
\label{eq:reducedgates}    
\end{equation}
with $9$ real parameters in $[-1,1]$ depending on the choice of $U$ in \eqref{eq:Ubar}~\cite{SM}. Remarkably, for any $U$, $w$ becomes bistochastic after a Hadamard transformation on the single wires~\cite{KBP:CorrelationsPerturbed}. 
This means that for the choice \eqref{eq:Ubar} of noise, $K_g(t,L)$ can be interpreted as the state-averaged return probability in a classical stochastic Markov process built as a brickwork circuit with the gate $w$.

Let us now evaluate the first two orders in the asymptotic expansion of \eqref{eq:SFF0} for large times. Conceptually, this will parallel similar derivations carried out in the periodically driven case, in both single-particle~\cite{Altshuler, Smilansky} and many-body~\cite{Chalker3, RP20, genSFF, ChaosvsMBL, Chan, KLP, Chalker2, genSFF} contexts. Indeed, even though $K_g(t)$ will generically relax to $1$ and not to $t$ (cf. Eq.~\eqref{eq:QCC}), in both cases the leading correction is exponential and the relaxation timescale can be interpreted as a Thouless time.

To proceed we now expand the trace \eqref{eq:SFF0} in the computational basis $\{\ket{e_i^m}\}$ where $m=0,\ldots,2L$ denotes the particle number (number of quasiparticles $\mcircf$) and $i=1,\ldots, \binom{2L}{m}$ labels states in a fixed $m$ sector. Assuming that there are no conserved charges we have
\be
\!\! K_{g}(t) =  \sum_{m=0}^{2L} K^{(m)}_g(t) = 1+ K^{(1)}_g(t) + \dots + K^{(2L)}_{g}(t),
\label{eq:GSFFave}
\ee
where we defined, in `1st quantization notation':
\be
K^{(m)}_g(t) \equiv \sum_{\{x_j\}_{j=1}^m}^{x_i>x_j : i>j}
\braket{\mcircf_{x_1}\!\!\!\cdots\!\mcircf_{x_m}\!\! | \!\mcircf_{x_1}\!\!\!\cdots\mcircf_{x_m}\!\!(t)}_L\, ,
\label{eq:GSFFmave}
\ee
and used that $K_g^{(0)}(t)=1$. We see that $K_g^{(m)}(t)$ is expressed as the sum of the averaged autocorrelation functions of the extended operators ${\sigma^z_{x_1}\cdots\sigma^z_{x_m}}$ (with ${x_1 < x_2  \cdots < x_m}$, and ${x_j\in\mathbb Z_{2L}/2}$) in finite volume $L$ (cf. Ref.~\cite{KBP:CorrelationsPerturbed}).

Let us now focus on a special family of reduced gates \eqref{eq:reducedgates}: those with either no \emph{splittings} (${f=e=0}$) or no \emph{mergers} (${b=d=0}$) and with non-negative weights. For this family of gates we can invoke the following property (proven in Sec.~\ref{app:bound} of the Supplemental Material (SM)~\cite{SM}:
\begin{property}
\label{prop:P1}
The averaged dynamical correlations $\braket{\mcircf_{x_1}\cdots\mcircf_{x_m} | \!\mcircf_{y_1}\cdots\mcircf_{y_m}\!(t) }_L$ are bounded from above by  
\be
\!\!\!{\rm max}\left(1, \frac{g}{\varepsilon_1 \varepsilon_2+ a c}\right)^{(m-1)t}\!\!\sum_{\sigma\in S_{m}}\prod_{i=1}^{m} {\braket{\mcircf_{x_i} | \!\mcircf_{y_{\sigma(i)}}\!(t) }_L}\,,
\ee
where $S_{m}$ is the permutation group of $m$ elements. 
\end{property}
\noindent Moreover, we also have:
\begin{property}
\label{prop:P2}
The two-point functions have the following asymptotic expansion in $t$
\be
\!\!\braket{\mcircf_{x} | \!\mcircf_{y}\!(t) }_L \simeq \frac{C_{\eta_x,\eta_y}}{C_{0,0}+C_{1,1}} \frac{\lambda^t}{L} + a_{\eta_x}^{2t}\delta_{t-(x-y)\,\text{\rm mod}\,L} \,,
\label{eq:C1asy}
\ee 
where ${\eta_j={2j\, {\rm mod}\, 2}}$, $\delta_0=1, \delta_{x\neq 0}=0$, ${a_0=a}$, ${a_{1}=c}$,
\be
\!\!\!\!\!\lambda=\frac{1}{4}\left( (a+c)+\sqrt{4 \varepsilon_1 \varepsilon_2+(a-c)^2}\right)^2
\label{eq:lambda}
\ee
while $C_{\eta_x,\eta_y}$ are constant amplitudes (${C_{0,0}}$ and ${C_{1,1}}$ are reported in Sec.~\ref{sec:AsyCorr} of the SM). 
\end{property}
\noindent An instructive way to obtain the expansion \eqref{eq:C1asy} is to note that the correlations in finite volume can be written as  
\be
\braket{\mcircf_x | \!\mcircf_y\!\!(t) }_L = \sum_{w=-\left \lfloor{t/L}\right \rfloor}^{\left \lfloor{t/L}\right \rfloor}  \braket{\mcircf_{x+w L} | \mcircf_y(t) }_{\infty}\; , 
\label{eq:corrFinite}
\ee
where $\braket{\mcircf_x | \!\mcircf_y\!\!(t) }_{\infty}$ are the infinite volume correlations known exactly from Ref.~\cite{KBP:CorrelationsPerturbed}. This form follows from the observation that for no splittings (mergers) the only contributions to the correlation come from continuous paths (the \emph{skeleton diagrams}~\cite{KBP:CorrelationsPerturbed}) connecting the endpoints, and wrapping around the cylindrical worldsheet along the space direction an arbitrary number of times. The maximal number of wrappings is restricted by the maximal speed of propagation. Then, Eq.~\eqref{eq:C1asy} follows directly plugging in the asymptotic form
\be
\!\!\!\!\braket{\mcircf_{x} | \!\mcircf_y\!\!(t) }_{\infty} \!\simeq \delta_{t-(x-y)} a_{\eta_x}^{2t} + \frac{ \lambda^t C_{\eta_x,\eta_y}}{\sqrt{t}} e^{- \frac{\scriptstyle (x-y- \bar{\zeta} t)^2}{\scriptstyle 4 D t}}\!,
\label{eq:CarrAsyGauss}
\ee
[where the diffusion constant is given by ${D=[4\pi (C_{0,0}+C_{1,1})^2]^{-1}}$ and the fluid velocity $\bar\zeta$ is defined in Sec.~\ref{sec:AsyCorr} of the SM], and turning the sum over $wL/t$ into an integral for $t\gg L$. Alternatively, Eq.~\eqref{eq:C1asy} can also be derived by diagonalising an effective Markov operator, see Sec.~\ref{app:Markov} of the SM.  

Using the asymptotic form~\eqref{eq:C1asy} for two-point correlations and Property~\ref{prop:P1} we find (see Sec.~\ref{app:boundK} of the SM)
\be
\sum_{m=2}^{2L} K^{(m)}_g(t) < C L^2 \lambda^{2t} {\rm max}\left(1, \frac{g^{t}}{(\varepsilon_1 \varepsilon_2+ a c)^{t}}\right)\,.
\label{eq:K2bound}
\ee
This leads us to our \textbf{first main result}: for large times and $\lambda\, {\rm max}[1, g/(\varepsilon_1 \varepsilon_2+ a c)]<1$, the GSFF is fully determined by correlation functions of \emph{local observables}
\be
K_{g}(t) \simeq 1+ K_g^{(1)}(t)\simeq 1+ \lambda^t+  (a^{2t}+c^{2t})L  \delta_{t\,\text{mod}\,L}\,.
\label{eq:Kasy}
\ee
In particular, since $\lambda > {\rm max}(a^2,c^2)$, we find 
 \be
K_g(t) \simeq 1+ e^{-t/\tau_{\rm th}},\qquad \tau_{\rm th}^{-1} = -{ \log \lambda} \; .
 \label{eq:ThTime}
 \ee
Note that in this case $\tau_{\rm th}$ is the exponent governing the decay of two-point correlations in infinite volume. Note also that there is no $L$ dependence in $\tau_{\rm th}$, in contrast to $\log L$ dependence found in several examples of extended systems, see e.g. \cite{KLP, Chalker2, genSFF}.  

Eq.~\eqref{eq:Kasy} shows excellent agreement with the exact numerical evaluation of $K_{g}(t)$, see Fig.~\ref{fig:1ptApprox} for a representative example. Moreover, our numerical observations suggest that the bound \eqref{eq:K2bound} is too conservative and Eq.~\eqref{eq:Kasy} holds whenever $\lambda <1$, namely whenever the averaged two-point correlations decay exponentially. 

When some of the gate's parameters (\ref{eq:reducedgates}) are negative, the Gaussian asymptotic form \eqref{eq:CarrAsyGauss} is not valid. We calculate $K_g^{(1)}(t) \simeq \lambda^t$ by diagonalising an effective Markov operator, see Sec.~\ref{app:Markov} of the SM ($\lambda$ can be different from the one in \eqref{eq:lambda}).
Moreover, we again bound the other contributions as in \eqref{eq:K2bound} (with a minor modification, see Sec.~\ref{app:boundK} of the SM). 

\begin{figure}
    \centering
    \includegraphics[scale=0.75]{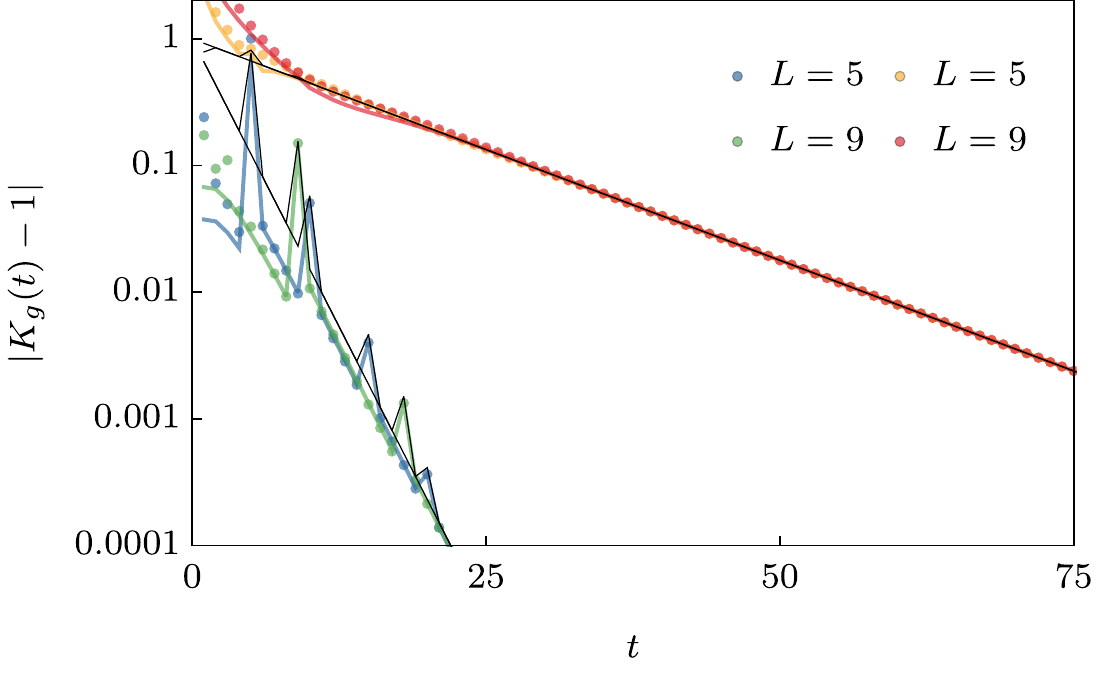}
    \caption{
Deviations of GSFF $K_g(t)$ from the RMT prediction for two different gates with no splittings, $b=d=0$. Symbols denote the exact numerical results for up to $18$ sites ($L=9$). Solid color line depicts $K_g^{(1)}(t)$ computed according to Eqs.~\eqref{eq:GSFFmave} and~\eqref{eq:corrFinite}. Notice that the $\tau_{\rm th }$ does not scale with $L$.  The solid black lines show the asymptotic from Eq.~\eqref{eq:Kasy}. 
We wrote the gate's parameters in Table~\ref{tab:U1gates} of Sec.~\ref{app:numerics} of the SM.
    }
    \label{fig:1ptApprox}
\end{figure}

Let us now consider a special case for which \eqref{eq:K2bound} does not provide a useful bound (because $\lambda=1$). Namely, the case of averaged gates with a conservation law. This situation has been extensively studied in the recent literature~\cite{Khemani, Rakovszky, Chalker3, genSFF} and can be realised in our setting by considering a gate $U$ (and hence $U_{x,\tau}$ in Eq.~\eqref{eq:Ubar}) that conserves the magnetisation in the $z$ direction. This leads to the following averaged gate~\cite{SM}
\be
w_{\rm U(1)} = \begin{pmatrix}
    1 & 0 & 0 & 0\\
0 & \cos^2 2J & \sin^2 2J & 0\\
0 & \sin^2 2J &\cos^2 2J & 0\\
0 & 0 & 0 &  1
    \end{pmatrix}\!,\,\quad J\in[0,{\pi}/{4}]\,.
    \label{eq:redCon}
\ee
Note that the time-evolution operator generated by this gate is integrable: it is an example of Floquet XXX model at a non-unitary point~\cite{Lenart}. Interestingly, a similar Floquet XXX model was obtained in Ref.~\cite{Chalker3} after averaging a ${\rm U}(1)$-symmetric Floquet Haar random circuit. Finally, we remark that a similar reduced gate for driven systems has been studied in Ref.~\cite{genSFF}. 

Since the magnetisation is conserved, the trace in \eqref{eq:GSFF} is reduced to a single magnetisation sector. This means that instead of $K_g(t)$ in Eq.~\eqref{eq:GSFFave} we should consider a single term $K_g^{(m)}(t)$ with fixed $m=0,1,\ldots,2L$. Moreover, we observe that, apart from the two trivial sectors $m=0$ and $m=2L$ where the GSFF is one, all $K_g^{(m)}(t)$ decay to one with the same exponent, see Fig.~\ref{fig:SFFcon}.
\begin{figure}
    \includegraphics[scale=0.75]{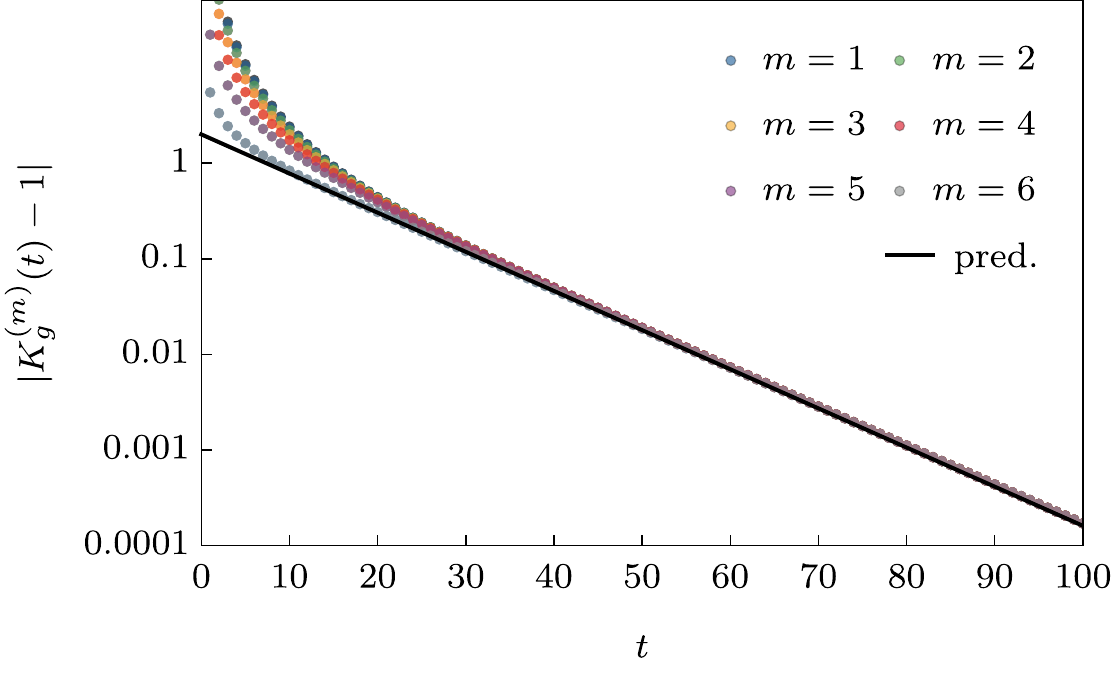}
    \caption{Deviations of the GSFF $|K^{(m)}_g(t)-1|$ with conservation laws. We show the results for different magnetisation sectors at $L=7$ (14 sites) and $J=0.3$. The black line is the prediction (\ref{eq:SFFCon}). 
    }
    \label{fig:SFFcon}
\end{figure}
This can be understood directly from the Bethe-Ansatz solution (see, e.g., the supplemental material of Ref.~\cite{Chalker3}). Indeed, by looking at the finite volume eigenstates one finds that the lowest excitations (those with eigenvalue of the Markov operator which is the closest to one), are one-magnon excitations (as opposed to bound states or scattering states of many magnons). Since the one-magnon states are highest weight states of the representation of SU(2) with $S_z=L-1$, their descendants (obtained by multiple applications of the lowering operators $S^-$) appear in all sectors $m=1,...,2L-1$. Therefore all sectors have the same Thouless time, which can be deduced from the $m=1$ sector.

For large times, the averaged two-point function for $m=1$ takes a simple diffusive form 
\be
\braket{\mcircf_{x} | \mcircf_0(t) }_{\infty} \simeq \frac{1}{2 \sqrt{4 \pi t D}} e^{-\frac{\scriptstyle x^2}{\scriptstyle 4 D t}}\,,
\label{eq:asyFinCon}
\ee
where ${D = (\tan^2 2J)}/4$ is the diffusion constant and we neglected exponentially small corrections with $L$-independent exponents because we expect a $L$-dependent Thouless time. Using again Eq.~\eqref{eq:corrFinite} we have 
\begin{align}
K^{(1)}_g (t) \simeq \frac{1}{2 \sqrt{4 \pi t D}} \!\!\! \sum_{w=-\left \lfloor{t/L}\right \rfloor}^{\left \lfloor{t/L}\right \rfloor}  e^{-\frac{\scriptstyle w^2 L^2}{\scriptstyle 4 D t}} \, .
\end{align}
Extending the summation to $\pm\infty$~\cite{note} and utilising the Poisson summation formula we get
\begin{align}
\!\!\!\!K^{(1)}_g (t)\simeq\!\!\!\sum_{n=-\infty}^{\infty} \!\!\!\! e^{-\frac{\scriptstyle 4 \pi^2 D t n^2}{\scriptstyle L^2}}\!\!\simeq \!1\!+\!2 e^{-{t}/{\tau_{\rm th}}},\,\,\tau_{\rm th}\! =\!  \frac{L^2}{4 \pi^2 D}\,.
\label{eq:SFFCon}
\end{align}
Note that the Thouless time depends on $L^2/D$, in agreement with previous observations in chaotic systems with diffusive conservation laws in both single-particle~\cite{Altshuler, Smilansky} and many-body~\cite{Chalker3, RP20, genSFF, ChaosvsMBL, Chan} contexts. Our derivation gives a straightforward illustration of the origin of this scaling.   

Another interesting limiting case is when, in addition to no splittings (or merges), at least one of $\varepsilon_1$ and $\varepsilon_2$ vanishes (note that ${\varepsilon_1=\varepsilon_2=0}$ if and only if the gate $U$ is dual-unitary~\cite{BKP:dualunitary, KBP:CorrelationsPerturbed}). In this case 
${K_g(t)= 1+ (a^{2t}+c^{2t})L  \delta_{t\,\text{mod}\,L} + ...}$, and the GSFF admits a closed-form expression (see Sec.~\ref{app:SWAPlike} of the SM). 
The model is chaotic when all $a,c,g$ differ from  $\pm 1$. In contrast, if the above conditions does not hold, $K_{g}(t)$ with unrestricted trace does not decay to the RMT result. This signals new commuting symmetries and possibly non-chaotic behaviour. For instance, for ${a=c=g=1}$ (corresponding to the SWAP gate) and unrestricted trace we find $K_{g}(t)|_{\rm SWAP} = 4^{{\rm gcd}(t,L)}$. Here ${\rm gcd}(t,L)$ is the greatest common divisor of $L$ and $t$. This result is manifestly larger than the RMT result. 

\begin{figure}[t!]
    \includegraphics[scale=0.55]{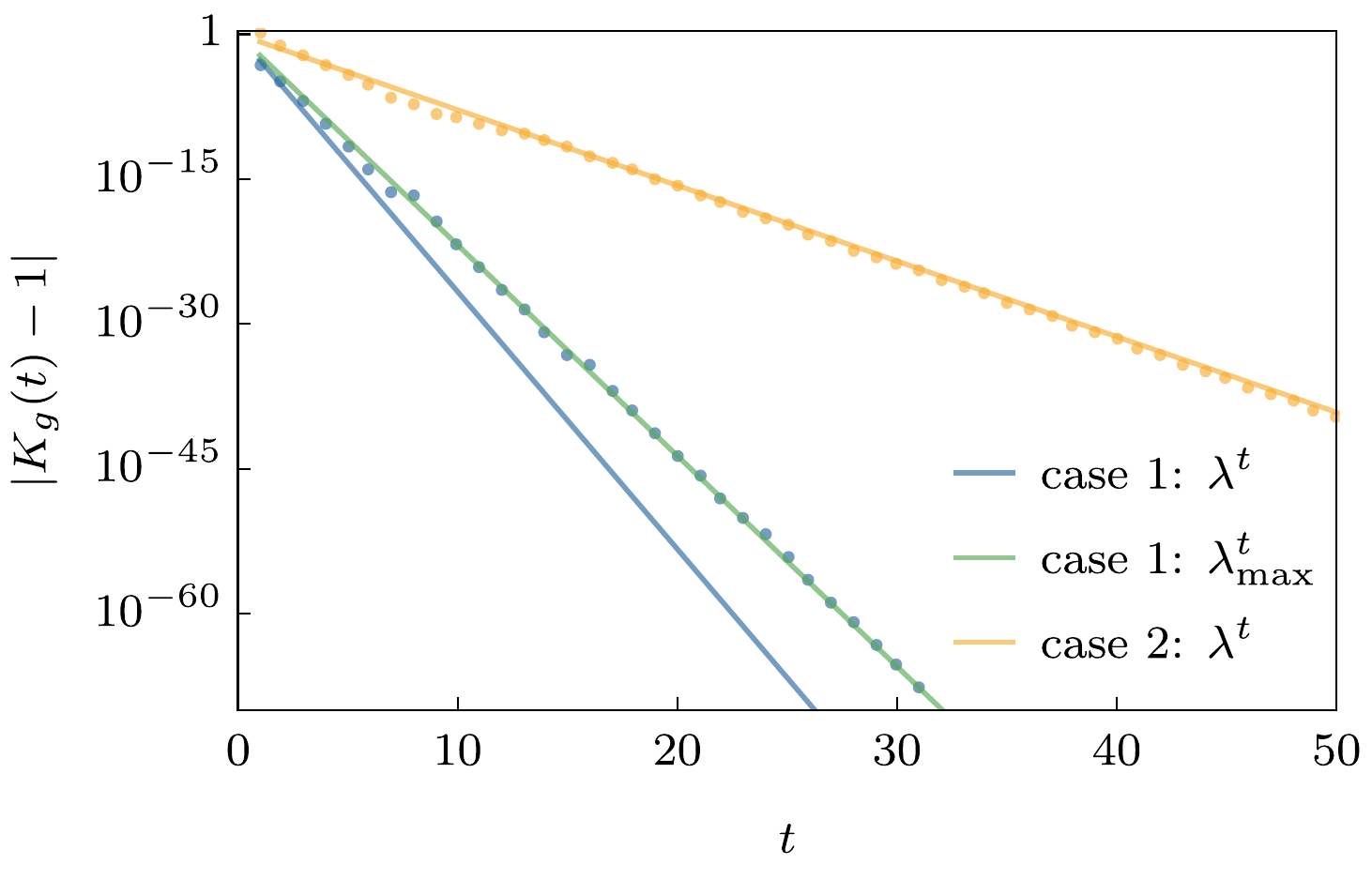}
    \caption{Symbols show the  numerical results for deviations of the spectral form factor $|K(t)-1|$ with allowed splittings and merges. 
     Solid line show the decay of the skeleton correlation functions $\lambda^t$ from \eqref{eq:lambda} and the true decay of the two-point correlation function $\lambda_{\rm max}^t$.
    In the case 2, the two-point correlation function is well described by the skeleton contributions (see \cite{KBP:CorrelationsPerturbed} for when this holds). In contrast, the case 1 exhibits a slower decay of the deviations than given by \ref{eq:lambda}, agreeing with the slower decay of the correlation functions.
    We obtained $\lambda_{\rm max}$ from direct numerical evaluation of the two-point correlation functions in infinite volume. 
The data shown is for the last two gates in Table \ref{tab:U1gates} of Sec.~\ref{app:numerics} of the SM, $L=8$.}
    \label{fig:Splittings}
\end{figure}

In the general case, when both mergers and splittings are allowed, there is a phase transition in the decay exponent of infinite volume correlations~\cite{KBP:CorrelationsPerturbed}. In particular, there is a region in parameter space (see Eq.~(41) in Ref.~\cite{KBP:CorrelationsPerturbed}) where the decay of quasiparticles is still governed by $\lambda$ in Eq.~\eqref{eq:lambda}, while for parameters out of this region the exponent changes. Moreover, all $K_g^{(m)}(t)$ will decay with the same exponent (since the number of particles can change during the time evolution, all $K_g^{(m)}(t)$ contain the slow-decaying configurations). However, this means that the decay exponent can again be determined from two-point functions of local operators and that $\tau_{\rm th}=-1/\log\lambda_{\rm max}$, where $\lambda_{\rm max}=\lim_{t\to\infty} ({\rm max}_x \braket{\mcircf_x | \mcircf_0(t) }_{\infty})^{1/t}$. This is in agreement with our numerical experiments, as shown Fig. \ref{fig:Splittings} for a representative example.

In conclusion, we studied the GSFF in a class of local quantum circuits with random fields, expressing it in terms of (averaged) dynamical correlations of local observables. By means of this correspondence we showed that in the regime where the correlations decay exponentially in time (known as ergodic and mixing in ergodicity theory) the GSFF approaches the prediction of random matrix theory over the same time-scale. Moreover, we proved that the GSFF approaches the prediction of random matrix theory also in the presence of a conservation law, if the correlations take a diffusive form. In this case the timescale is proportional to the system size squared divided by the diffusion constant. Finally, we showed that when the correlations do not decay, the GSFF does not approach the random matrix theory prediction. The correspondence between quantum chaotic and quantum ergodic and mixing regimes is expected on general grounds~\cite{KI_JPA,Mirko,Casati,LEReview} and provides an intuitive understanding of quantum chaos. Our results in a specific setting provide a rigorous proof of such a correspondence, and pave the way for its quantitative understanding in more general settings. Moreover, interpreting the U(1)-noise averaged GSFF as a state-averaged return probability for a general bistochastic brickwork Markov circuit provides an analogous correspondence in classical stochastic systems.

{\em Acknowledgement:} The work is supported by the EU Horizon 2020 program through the ERC Advanced Grant OMNES No. 694544, and by the Slovenian Research Agency (ARRS) under the Programme P1-0402. BB was also supported by the Royal Society through the University Research Fellowship No. 201102.

\onecolumngrid

\pagebreak

\newcounter{equationSM}
\newcounter{figureSM}
\newcounter{tableSM}
\stepcounter{equationSM}
\setcounter{equation}{0}
\setcounter{figure}{0}
\setcounter{table}{0}
\makeatletter
\renewcommand{\theequation}{\textsc{sm}-\arabic{equation}}
\renewcommand{\thefigure}{\textsc{sm}-\arabic{figure}}
\renewcommand{\thetable}{\textsc{sm}-\arabic{table}}

\begin{center}
{\large{\bf Supplemental Material for\\
 ``Chaos and Ergodicity in Extended Quantum Systems with Noisy Driving''}}
\end{center}
Here we report some useful information complementing the main text. In particular
\begin{itemize}

\item[-] In Section~\ref{app:ReducedGate} we show how to derive the representation~\eqref{eq:SFF0} for the Generalised Spectral Form Factor;
\item[-] In Section~\ref{app:bound} we prove Property~\ref{prop:P1};
\item[-] In Section~\ref{sec:AsyCorr} we compute the asymptotics of the averaged two-point functions of local operators in infinite volume;
\item[-] In Section~\ref{app:Markov} we compute $K_g^{(1)}(t)$ using an effective Markov operator;
\item[-] In Section~\ref{app:boundK} we establish the bound in Eq.~\eqref{eq:K2bound};
\item[-] In Section~\ref{app:SWAPlike} we compute $K_{g}(t)$ with unrestricted trace for $f=e=\varepsilon_{1}=0$;
\item[-] In Section~\ref{app:numerics} we report the parameters of the gates used in our numerical experiments;
\end{itemize}

{
\section{Derivation of Eq.~6}
\label{app:ReducedGate}
In this appendix we show how to write the GSFF in terms of reduced gates, both in the presence~\eqref{eq:redCon} and in the absence~\eqref{eq:SFF0} of conservation laws. 

We begin by writing the GSFF $K_g(t)= \braket{{\rm tr}\mathbb U(t)\,\, {\rm tr}\mathbb U^\dagger(t) }$ in the circuit representation 
\be
\label{eq:fullGSFF}
K_g(t) = \Biggl \langle \,\, \begin{tikzpicture}[baseline=(current  bounding  box.center), scale=0.49]
\foreach \i in {1,...,5}{
\draw[thick, dotted] (2*\i+2-12.5+0.255,-1.75-0.1) -- (2*\i+2-12.5+0.255,4.25-0.1);
\draw[thick, dotted] (2*\i+2-11.5-0.255,-1.75-0.1) -- (2*\i+2-11.5-0.255,4.25-0.1);}

\foreach \i in {1,...,5}{
\draw[thick] (2*\i+2-11.5,4) arc (-45:175:0.15);
\draw[thick] (2*\i+2-11.5,-2) arc (315:180:0.15);
\draw[thick] (2*\i+2-0.5-12,-2) arc (-135:0:0.15);
}
\foreach \i in {2,...,6}{
\draw[thick] (2*\i+2-2.5-12,4) arc (225:0:0.15);
}
\foreach \i in {0,1,2}{
\draw[thick, dotted] (-9.5,2*\i-1.745) -- (0.4,2*\i-1.745);
\draw[thick, dotted] (-9.5,2*\i-1.255) -- (0.4,2*\i-1.255);
}
\foreach \i in{1.5,2.5,3.5}{
\draw[thick] (0.5,2*\i-0.5-3.5) arc (45:-90:0.15);
\draw[thick] (-10+0.5+0,2*\i-0.5-3.5) arc (45:270:0.15);
}
\foreach \i in{0.5,1.5,2.5}
{
\draw[ thick] (0.5,1+2*\i-0.5-3.5) arc (-45:90:0.15);
\draw[ thick] (-10+0.5,1+2*\i-0.5-3.5) arc (315:90:0.15);
}
\foreach \jj in {0}{
\foreach \i in {0}{
\draw[thick] (-.5-2*\i,-1+\jj) -- (0.525-2*\i,0.025+\jj);
\draw[thick] (-0.525-2*\i,0.025+\jj) -- (0.5-2*\i,-1+\jj);
\draw[thick, fill=myblue1, rounded corners=2pt] (-0.25-2*\i,-0.25+\jj) rectangle (.25-2*\i,-0.75+\jj);
\draw[thick] (-2*\i,-0.35+\jj) -- (-2*\i+0.15,-0.35+\jj) -- (-2*\i+0.15,-0.5+\jj);%
}
\foreach \i in {1}{
\draw[thick] (-.5-2*\i,-1+\jj) -- (0.525-2*\i,0.025+\jj);
\draw[thick] (-0.525-2*\i,0.025+\jj) -- (0.5-2*\i,-1+\jj);
\draw[thick, fill=myblue4, rounded corners=2pt] (-0.25-2*\i,-0.25+\jj) rectangle (.25-2*\i,-0.75+\jj);
\draw[thick] (-2*\i,-0.35+\jj) -- (-2*\i+0.15,-0.35+\jj) -- (-2*\i+0.15,-0.5+\jj);%
}
\foreach \i in {2}{
\draw[thick] (-.5-2*\i,-1+\jj) -- (0.525-2*\i,0.025+\jj);
\draw[thick] (-0.525-2*\i,0.025+\jj) -- (0.5-2*\i,-1+\jj);
\draw[thick, fill=myblue3, rounded corners=2pt] (-0.25-2*\i,-0.25+\jj) rectangle (.25-2*\i,-0.75+\jj);
\draw[thick] (-2*\i,-0.35+\jj) -- (-2*\i+0.15,-0.35+\jj) -- (-2*\i+0.15,-0.5+\jj);%
}
\foreach \i in {3}{
\draw[thick] (-.5-2*\i,-1+\jj) -- (0.525-2*\i,0.025+\jj);
\draw[thick] (-0.525-2*\i,0.025+\jj) -- (0.5-2*\i,-1+\jj);
\draw[thick, fill=myblue7, rounded corners=2pt] (-0.25-2*\i,-0.25+\jj) rectangle (.25-2*\i,-0.75+\jj);
\draw[thick] (-2*\i,-0.35+\jj) -- (-2*\i+0.15,-0.35+\jj) -- (-2*\i+0.15,-0.5+\jj);%
}
\foreach \i in {4}{
\draw[thick] (-.5-2*\i,-1+\jj) -- (0.525-2*\i,0.025+\jj);
\draw[thick] (-0.525-2*\i,0.025+\jj) -- (0.5-2*\i,-1+\jj);
\draw[thick, fill=myblue6, rounded corners=2pt] (-0.25-2*\i,-0.25+\jj) rectangle (.25-2*\i,-0.75+\jj);
\draw[thick] (-2*\i,-0.35+\jj) -- (-2*\i+0.15,-0.35+\jj) -- (-2*\i+0.15,-0.5+\jj);%
}
}
\foreach \jj in {2}{
\foreach \i in {0}{
\draw[thick] (-.5-2*\i,-1+\jj) -- (0.525-2*\i,0.025+\jj);
\draw[thick] (-0.525-2*\i,0.025+\jj) -- (0.5-2*\i,-1+\jj);
\draw[thick, fill=myblue9, rounded corners=2pt] (-0.25-2*\i,-0.25+\jj) rectangle (.25-2*\i,-0.75+\jj);
\draw[thick] (-2*\i,-0.35+\jj) -- (-2*\i+0.15,-0.35+\jj) -- (-2*\i+0.15,-0.5+\jj);%
}
\foreach \i in {1}{
\draw[thick] (-.5-2*\i,-1+\jj) -- (0.525-2*\i,0.025+\jj);
\draw[thick] (-0.525-2*\i,0.025+\jj) -- (0.5-2*\i,-1+\jj);
\draw[thick, fill=myblue10, rounded corners=2pt] (-0.25-2*\i,-0.25+\jj) rectangle (.25-2*\i,-0.75+\jj);
\draw[thick] (-2*\i,-0.35+\jj) -- (-2*\i+0.15,-0.35+\jj) -- (-2*\i+0.15,-0.5+\jj);%
}
\foreach \i in {2}{
\draw[thick] (-.5-2*\i,-1+\jj) -- (0.525-2*\i,0.025+\jj);
\draw[thick] (-0.525-2*\i,0.025+\jj) -- (0.5-2*\i,-1+\jj);
\draw[thick, fill=myblue2, rounded corners=2pt] (-0.25-2*\i,-0.25+\jj) rectangle (.25-2*\i,-0.75+\jj);
\draw[thick] (-2*\i,-0.35+\jj) -- (-2*\i+0.15,-0.35+\jj) -- (-2*\i+0.15,-0.5+\jj);%
}
\foreach \i in {3}{
\draw[thick] (-.5-2*\i,-1+\jj) -- (0.525-2*\i,0.025+\jj);
\draw[thick] (-0.525-2*\i,0.025+\jj) -- (0.5-2*\i,-1+\jj);
\draw[thick, fill=myblue6, rounded corners=2pt] (-0.25-2*\i,-0.25+\jj) rectangle (.25-2*\i,-0.75+\jj);
\draw[thick] (-2*\i,-0.35+\jj) -- (-2*\i+0.15,-0.35+\jj) -- (-2*\i+0.15,-0.5+\jj);%
}
\foreach \i in {4}{
\draw[thick] (-.5-2*\i,-1+\jj) -- (0.525-2*\i,0.025+\jj);
\draw[thick] (-0.525-2*\i,0.025+\jj) -- (0.5-2*\i,-1+\jj);
\draw[thick, fill=myblue7, rounded corners=2pt] (-0.25-2*\i,-0.25+\jj) rectangle (.25-2*\i,-0.75+\jj);
\draw[thick] (-2*\i,-0.35+\jj) -- (-2*\i+0.15,-0.35+\jj) -- (-2*\i+0.15,-0.5+\jj);%
}
}
\foreach \jj in {4}{
\foreach \i in {0}{
\draw[thick] (-.5-2*\i,-1+\jj) -- (0.525-2*\i,0.025+\jj);
\draw[thick] (-0.525-2*\i,0.025+\jj) -- (0.5-2*\i,-1+\jj);
\draw[thick, fill=myblue4, rounded corners=2pt] (-0.25-2*\i,-0.25+\jj) rectangle (.25-2*\i,-0.75+\jj);
\draw[thick] (-2*\i,-0.35+\jj) -- (-2*\i+0.15,-0.35+\jj) -- (-2*\i+0.15,-0.5+\jj);%
}
\foreach \i in {1}{
\draw[thick] (-.5-2*\i,-1+\jj) -- (0.525-2*\i,0.025+\jj);
\draw[thick] (-0.525-2*\i,0.025+\jj) -- (0.5-2*\i,-1+\jj);
\draw[thick, fill=myblue3, rounded corners=2pt] (-0.25-2*\i,-0.25+\jj) rectangle (.25-2*\i,-0.75+\jj);
\draw[thick] (-2*\i,-0.35+\jj) -- (-2*\i+0.15,-0.35+\jj) -- (-2*\i+0.15,-0.5+\jj);%
}
\foreach \i in {2}{
\draw[thick] (-.5-2*\i,-1+\jj) -- (0.525-2*\i,0.025+\jj);
\draw[thick] (-0.525-2*\i,0.025+\jj) -- (0.5-2*\i,-1+\jj);
\draw[thick, fill=myblue2, rounded corners=2pt] (-0.25-2*\i,-0.25+\jj) rectangle (.25-2*\i,-0.75+\jj);
\draw[thick] (-2*\i,-0.35+\jj) -- (-2*\i+0.15,-0.35+\jj) -- (-2*\i+0.15,-0.5+\jj);%
}
\foreach \i in {3}{
\draw[thick] (-.5-2*\i,-1+\jj) -- (0.525-2*\i,0.025+\jj);
\draw[thick] (-0.525-2*\i,0.025+\jj) -- (0.5-2*\i,-1+\jj);
\draw[thick, fill=myblue5, rounded corners=2pt] (-0.25-2*\i,-0.25+\jj) rectangle (.25-2*\i,-0.75+\jj);
\draw[thick] (-2*\i,-0.35+\jj) -- (-2*\i+0.15,-0.35+\jj) -- (-2*\i+0.15,-0.5+\jj);%
}
\foreach \i in {4}{
\draw[thick] (-.5-2*\i,-1+\jj) -- (0.525-2*\i,0.025+\jj);
\draw[thick] (-0.525-2*\i,0.025+\jj) -- (0.5-2*\i,-1+\jj);
\draw[thick, fill=myblue10, rounded corners=2pt] (-0.25-2*\i,-0.25+\jj) rectangle (.25-2*\i,-0.75+\jj);
\draw[thick] (-2*\i,-0.35+\jj) -- (-2*\i+0.15,-0.35+\jj) -- (-2*\i+0.15,-0.5+\jj);%
}
}
\foreach \jj[evaluate=\jj as \j using -2*(ceil(\jj/2)-\jj/2)] in {0}{
\foreach \i in {1}
{
\draw[thick] (.5-2*\i-1*\j,-2-1*\jj) -- (1-2*\i-1*\j,-1.5-\jj);
\draw[thick] (1-2*\i-1*\j,-1.5-1*\jj) -- (1.5-2*\i-1*\j,-2-\jj);
\draw[thick] (.5-2*\i-1*\j,-1-1*\jj) -- (1-2*\i-1*\j,-1.5-\jj);
\draw[thick] (1-2*\i-1*\j,-1.5-1*\jj) -- (1.5-2*\i-1*\j,-1-\jj);
\draw[thick, fill=myblue5, rounded corners=2pt] (0.75-2*\i-1*\j,-1.75-\jj) rectangle (1.25-2*\i-1*\j,-1.25-\jj);
\draw[thick] (-2*\i+1,-1.35-\jj) -- (-2*\i+1.15,-1.35-\jj) -- (-2*\i+1.15,-1.5-\jj);%
}
\foreach \i in {2}
{
\draw[thick] (.5-2*\i-1*\j,-2-1*\jj) -- (1-2*\i-1*\j,-1.5-\jj);
\draw[thick] (1-2*\i-1*\j,-1.5-1*\jj) -- (1.5-2*\i-1*\j,-2-\jj);
\draw[thick] (.5-2*\i-1*\j,-1-1*\jj) -- (1-2*\i-1*\j,-1.5-\jj);
\draw[thick] (1-2*\i-1*\j,-1.5-1*\jj) -- (1.5-2*\i-1*\j,-1-\jj);
\draw[thick, fill=myblue4, rounded corners=2pt] (0.75-2*\i-1*\j,-1.75-\jj) rectangle (1.25-2*\i-1*\j,-1.25-\jj);
\draw[thick] (-2*\i+1,-1.35-\jj) -- (-2*\i+1.15,-1.35-\jj) -- (-2*\i+1.15,-1.5-\jj);%
}
\foreach \i in {3}
{
\draw[thick] (.5-2*\i-1*\j,-2-1*\jj) -- (1-2*\i-1*\j,-1.5-\jj);
\draw[thick] (1-2*\i-1*\j,-1.5-1*\jj) -- (1.5-2*\i-1*\j,-2-\jj);
\draw[thick] (.5-2*\i-1*\j,-1-1*\jj) -- (1-2*\i-1*\j,-1.5-\jj);
\draw[thick] (1-2*\i-1*\j,-1.5-1*\jj) -- (1.5-2*\i-1*\j,-1-\jj);
\draw[thick, fill=myblue3, rounded corners=2pt] (0.75-2*\i-1*\j,-1.75-\jj) rectangle (1.25-2*\i-1*\j,-1.25-\jj);
\draw[thick] (-2*\i+1,-1.35-\jj) -- (-2*\i+1.15,-1.35-\jj) -- (-2*\i+1.15,-1.5-\jj);%
}
\foreach \i in {4}
{
\draw[thick] (.5-2*\i-1*\j,-2-1*\jj) -- (1-2*\i-1*\j,-1.5-\jj);
\draw[thick] (1-2*\i-1*\j,-1.5-1*\jj) -- (1.5-2*\i-1*\j,-2-\jj);
\draw[thick] (.5-2*\i-1*\j,-1-1*\jj) -- (1-2*\i-1*\j,-1.5-\jj);
\draw[thick] (1-2*\i-1*\j,-1.5-1*\jj) -- (1.5-2*\i-1*\j,-1-\jj);
\draw[thick, fill=myblue2, rounded corners=2pt] (0.75-2*\i-1*\j,-1.75-\jj) rectangle (1.25-2*\i-1*\j,-1.25-\jj);
\draw[thick] (-2*\i+1,-1.35-\jj) -- (-2*\i+1.15,-1.35-\jj) -- (-2*\i+1.15,-1.5-\jj);%
}
\foreach \i in {5}
{
\draw[thick] (.5-2*\i-1*\j,-2-1*\jj) -- (1-2*\i-1*\j,-1.5-\jj);
\draw[thick] (1-2*\i-1*\j,-1.5-1*\jj) -- (1.5-2*\i-1*\j,-2-\jj);
\draw[thick] (.5-2*\i-1*\j,-1-1*\jj) -- (1-2*\i-1*\j,-1.5-\jj);
\draw[thick] (1-2*\i-1*\j,-1.5-1*\jj) -- (1.5-2*\i-1*\j,-1-\jj);
\draw[thick, fill=myblue1, rounded corners=2pt] (0.75-2*\i-1*\j,-1.75-\jj) rectangle (1.25-2*\i-1*\j,-1.25-\jj);
\draw[thick] (-2*\i+1,-1.35-\jj) -- (-2*\i+1.15,-1.35-\jj) -- (-2*\i+1.15,-1.5-\jj);%
}
}
\foreach \jj[evaluate=\jj as \j using -2*(ceil(\jj/2)-\jj/2)] in {-2}{
\foreach \i in {1}
{
\draw[thick] (.5-2*\i-1*\j,-2-1*\jj) -- (1-2*\i-1*\j,-1.5-\jj);
\draw[thick] (1-2*\i-1*\j,-1.5-1*\jj) -- (1.5-2*\i-1*\j,-2-\jj);
\draw[thick] (.5-2*\i-1*\j,-1-1*\jj) -- (1-2*\i-1*\j,-1.5-\jj);
\draw[thick] (1-2*\i-1*\j,-1.5-1*\jj) -- (1.5-2*\i-1*\j,-1-\jj);
\draw[thick, fill=myblue8, rounded corners=2pt] (0.75-2*\i-1*\j,-1.75-\jj) rectangle (1.25-2*\i-1*\j,-1.25-\jj);
\draw[thick] (-2*\i+1,-1.35-\jj) -- (-2*\i+1.15,-1.35-\jj) -- (-2*\i+1.15,-1.5-\jj);%
}
\foreach \i in {2}
{
\draw[thick] (.5-2*\i-1*\j,-2-1*\jj) -- (1-2*\i-1*\j,-1.5-\jj);
\draw[thick] (1-2*\i-1*\j,-1.5-1*\jj) -- (1.5-2*\i-1*\j,-2-\jj);
\draw[thick] (.5-2*\i-1*\j,-1-1*\jj) -- (1-2*\i-1*\j,-1.5-\jj);
\draw[thick] (1-2*\i-1*\j,-1.5-1*\jj) -- (1.5-2*\i-1*\j,-1-\jj);
\draw[thick, fill=myblue10, rounded corners=2pt] (0.75-2*\i-1*\j,-1.75-\jj) rectangle (1.25-2*\i-1*\j,-1.25-\jj);
\draw[thick] (-2*\i+1,-1.35-\jj) -- (-2*\i+1.15,-1.35-\jj) -- (-2*\i+1.15,-1.5-\jj);%
}
\foreach \i in {3}
{
\draw[thick] (.5-2*\i-1*\j,-2-1*\jj) -- (1-2*\i-1*\j,-1.5-\jj);
\draw[thick] (1-2*\i-1*\j,-1.5-1*\jj) -- (1.5-2*\i-1*\j,-2-\jj);
\draw[thick] (.5-2*\i-1*\j,-1-1*\jj) -- (1-2*\i-1*\j,-1.5-\jj);
\draw[thick] (1-2*\i-1*\j,-1.5-1*\jj) -- (1.5-2*\i-1*\j,-1-\jj);
\draw[thick, fill=myblue6, rounded corners=2pt] (0.75-2*\i-1*\j,-1.75-\jj) rectangle (1.25-2*\i-1*\j,-1.25-\jj);
\draw[thick] (-2*\i+1,-1.35-\jj) -- (-2*\i+1.15,-1.35-\jj) -- (-2*\i+1.15,-1.5-\jj);%
}
\foreach \i in {4}
{
\draw[thick] (.5-2*\i-1*\j,-2-1*\jj) -- (1-2*\i-1*\j,-1.5-\jj);
\draw[thick] (1-2*\i-1*\j,-1.5-1*\jj) -- (1.5-2*\i-1*\j,-2-\jj);
\draw[thick] (.5-2*\i-1*\j,-1-1*\jj) -- (1-2*\i-1*\j,-1.5-\jj);
\draw[thick] (1-2*\i-1*\j,-1.5-1*\jj) -- (1.5-2*\i-1*\j,-1-\jj);
\draw[thick, fill=myblue9, rounded corners=2pt] (0.75-2*\i-1*\j,-1.75-\jj) rectangle (1.25-2*\i-1*\j,-1.25-\jj);
\draw[thick] (-2*\i+1,-1.35-\jj) -- (-2*\i+1.15,-1.35-\jj) -- (-2*\i+1.15,-1.5-\jj);%
}
\foreach \i in {5}
{
\draw[thick] (.5-2*\i-1*\j,-2-1*\jj) -- (1-2*\i-1*\j,-1.5-\jj);
\draw[thick] (1-2*\i-1*\j,-1.5-1*\jj) -- (1.5-2*\i-1*\j,-2-\jj);
\draw[thick] (.5-2*\i-1*\j,-1-1*\jj) -- (1-2*\i-1*\j,-1.5-\jj);
\draw[thick] (1-2*\i-1*\j,-1.5-1*\jj) -- (1.5-2*\i-1*\j,-1-\jj);
\draw[thick, fill=myblue1, rounded corners=2pt] (0.75-2*\i-1*\j,-1.75-\jj) rectangle (1.25-2*\i-1*\j,-1.25-\jj);
\draw[thick] (-2*\i+1,-1.35-\jj) -- (-2*\i+1.15,-1.35-\jj) -- (-2*\i+1.15,-1.5-\jj);%
}
}
\foreach \jj[evaluate=\jj as \j using -2*(ceil(\jj/2)-\jj/2)] in {-4}{
\foreach \i in {1}
{
\draw[thick] (.5-2*\i-1*\j,-2-1*\jj) -- (1-2*\i-1*\j,-1.5-\jj);
\draw[thick] (1-2*\i-1*\j,-1.5-1*\jj) -- (1.5-2*\i-1*\j,-2-\jj);
\draw[thick] (.5-2*\i-1*\j,-1-1*\jj) -- (1-2*\i-1*\j,-1.5-\jj);
\draw[thick] (1-2*\i-1*\j,-1.5-1*\jj) -- (1.5-2*\i-1*\j,-1-\jj);
\draw[thick, fill=myblue6, rounded corners=2pt] (0.75-2*\i-1*\j,-1.75-\jj) rectangle (1.25-2*\i-1*\j,-1.25-\jj);
\draw[thick] (-2*\i+1,-1.35-\jj) -- (-2*\i+1.15,-1.35-\jj) -- (-2*\i+1.15,-1.5-\jj);%
}
\foreach \i in {2}
{
\draw[thick] (.5-2*\i-1*\j,-2-1*\jj) -- (1-2*\i-1*\j,-1.5-\jj);
\draw[thick] (1-2*\i-1*\j,-1.5-1*\jj) -- (1.5-2*\i-1*\j,-2-\jj);
\draw[thick] (.5-2*\i-1*\j,-1-1*\jj) -- (1-2*\i-1*\j,-1.5-\jj);
\draw[thick] (1-2*\i-1*\j,-1.5-1*\jj) -- (1.5-2*\i-1*\j,-1-\jj);
\draw[thick, fill=myblue3, rounded corners=2pt] (0.75-2*\i-1*\j,-1.75-\jj) rectangle (1.25-2*\i-1*\j,-1.25-\jj);
\draw[thick] (-2*\i+1,-1.35-\jj) -- (-2*\i+1.15,-1.35-\jj) -- (-2*\i+1.15,-1.5-\jj);%
}
\foreach \i in {3}
{
\draw[thick] (.5-2*\i-1*\j,-2-1*\jj) -- (1-2*\i-1*\j,-1.5-\jj);
\draw[thick] (1-2*\i-1*\j,-1.5-1*\jj) -- (1.5-2*\i-1*\j,-2-\jj);
\draw[thick] (.5-2*\i-1*\j,-1-1*\jj) -- (1-2*\i-1*\j,-1.5-\jj);
\draw[thick] (1-2*\i-1*\j,-1.5-1*\jj) -- (1.5-2*\i-1*\j,-1-\jj);
\draw[thick, fill=myblue7, rounded corners=2pt] (0.75-2*\i-1*\j,-1.75-\jj) rectangle (1.25-2*\i-1*\j,-1.25-\jj);
\draw[thick] (-2*\i+1,-1.35-\jj) -- (-2*\i+1.15,-1.35-\jj) -- (-2*\i+1.15,-1.5-\jj);%
}
\foreach \i in {4}
{
\draw[thick] (.5-2*\i-1*\j,-2-1*\jj) -- (1-2*\i-1*\j,-1.5-\jj);
\draw[thick] (1-2*\i-1*\j,-1.5-1*\jj) -- (1.5-2*\i-1*\j,-2-\jj);
\draw[thick] (.5-2*\i-1*\j,-1-1*\jj) -- (1-2*\i-1*\j,-1.5-\jj);
\draw[thick] (1-2*\i-1*\j,-1.5-1*\jj) -- (1.5-2*\i-1*\j,-1-\jj);
\draw[thick, fill=myblue2, rounded corners=2pt] (0.75-2*\i-1*\j,-1.75-\jj) rectangle (1.25-2*\i-1*\j,-1.25-\jj);
\draw[thick] (-2*\i+1,-1.35-\jj) -- (-2*\i+1.15,-1.35-\jj) -- (-2*\i+1.15,-1.5-\jj);%
}
\foreach \i in {5}
{
\draw[thick] (.5-2*\i-1*\j,-2-1*\jj) -- (1-2*\i-1*\j,-1.5-\jj);
\draw[thick] (1-2*\i-1*\j,-1.5-1*\jj) -- (1.5-2*\i-1*\j,-2-\jj);
\draw[thick] (.5-2*\i-1*\j,-1-1*\jj) -- (1-2*\i-1*\j,-1.5-\jj);
\draw[thick] (1-2*\i-1*\j,-1.5-1*\jj) -- (1.5-2*\i-1*\j,-1-\jj);
\draw[thick, fill=myblue7, rounded corners=2pt] (0.75-2*\i-1*\j,-1.75-\jj) rectangle (1.25-2*\i-1*\j,-1.25-\jj);
\draw[thick] (-2*\i+1,-1.35-\jj) -- (-2*\i+1.15,-1.35-\jj) -- (-2*\i+1.15,-1.5-\jj);%
}
}
\def\shiftx{0}
\def\shifty{-11}
\begin{scope}[yshift=+2.15cm]
\foreach \i in {1,...,5}{
\draw[ thick, dotted] (2*\i+2-1.485+0.25+\shiftx-11,-2.5-0.1+\shifty+2.5) -- (2*\i+2-1.485+0.25+\shiftx-11,3.5-0.1+\shifty+2.5);
\draw[ thick, dotted] (2*\i+2-0.525-0.25+\shiftx-11,-2.5-0.1+\shifty+2.5) -- (2*\i+2-0.525-0.25+\shiftx-11,3.5-0.1+\shifty+2.5);
}
\foreach \i in {0,1,2}{
\draw[ thick, dotted] (1.75+\shiftx-11,2*\i-1.25+\shifty+2.5) -- (11.5+\shiftx-11,2*\i-1.25+\shifty+2.5);
\draw[ thick, dotted] (1.5+\shiftx-11,2*\i-.76+\shifty+2.5) -- (11.5+\shiftx-11,2*\i-.76+\shifty+2.5);
}
\foreach \i in {1,...,5}
{
\draw[ thick] (2*\i+2-1.5-11+\shiftx,3.5+\shifty+2.5) arc (135:-0:0.15);
\draw[ thick] (2*\i+2-.5-11+\shiftx,3.5+\shifty+2.5) arc (-325:-180:0.15);
\draw[ thick] (2*\i+2-1.5-11+\shiftx,-2.5+\shifty+2.5) arc (-45:180:-0.15);
\draw[ thick] (2*\i+2-0.5-11+\shiftx,-2.5+\shifty+2.5) arc (45:-180:0.15);
}
\foreach \i in {3,...,5}
{
\draw[ thick] (\shiftx+.5,2*\i-0.5-3.5+\shifty) arc (45:-90:0.15);
\draw[ thick] (\shiftx-10+0.5+0,2*\i-0.5-3.5+\shifty) arc (45:280:0.15);
}
\foreach \i in{2,...,4}
{
\draw[ thick] (\shiftx+.5,1+2*\i-0.5-3.5+\shifty) arc (-45:90:0.15);
\draw[ thick] (\shiftx-10+0.5,1+2*\i-0.5-3.5+\shifty) arc (-45:-280:0.15);
}
\foreach \i in {1,...,5}
{
\draw[ thick] (\shiftx+.5-2*\i,6+\shifty) -- (\shiftx+1-2*\i,5.5+\shifty);
\draw[ thick] (\shiftx+1.5-2*\i,6+\shifty) -- (\shiftx+1-2*\i,5.5+\shifty);
}
\foreach \jj[evaluate=\jj as \j using -2*(ceil(\jj/2)-\jj/2)] in {0,...,3}
\foreach \i in {1,...,5}
{
\draw[ thick] (\shiftx+.5-2*\i-1*\j,2+1*\jj+\shifty) -- (\shiftx+1-2*\i-1*\j,1.5+\jj+\shifty);
\draw[ thick] (\shiftx+1-2*\i-1*\j,1.5+1*\jj+\shifty) -- (\shiftx+1.5-2*\i-1*\j,2+\jj+\shifty);
}
\def\shiftx{-1}
\def\shifty{-9}
\foreach \jj in {0+\shifty}{
\foreach \i in {0+\shiftx}{
\draw[thick] (-.5-2*\i,-1+\jj) -- (0.525-2*\i,0.025+\jj);
\draw[thick] (-0.525-2*\i,0.025+\jj) -- (0.5-2*\i,-1+\jj);
\draw[thick, fill=myred1, rounded corners=2pt] (-0.25-2*\i,-0.25+\jj) rectangle (.25-2*\i,-0.75+\jj);
\draw[thick] (-2*\i,-0.35+\jj) -- (-2*\i+0.15,-0.35+\jj) -- (-2*\i+0.15,-0.5+\jj);%
}
\foreach \i in {1+\shiftx}{
\draw[thick] (-.5-2*\i,-1+\jj) -- (0.525-2*\i,0.025+\jj);
\draw[thick] (-0.525-2*\i,0.025+\jj) -- (0.5-2*\i,-1+\jj);
\draw[thick, fill=myred4, rounded corners=2pt] (-0.25-2*\i,-0.25+\jj) rectangle (.25-2*\i,-0.75+\jj);
\draw[thick] (-2*\i,-0.35+\jj) -- (-2*\i+0.15,-0.35+\jj) -- (-2*\i+0.15,-0.5+\jj);%
}
\foreach \i in {2+\shiftx}{
\draw[thick] (-.5-2*\i,-1+\jj) -- (0.525-2*\i,0.025+\jj);
\draw[thick] (-0.525-2*\i,0.025+\jj) -- (0.5-2*\i,-1+\jj);
\draw[thick, fill=myred3, rounded corners=2pt] (-0.25-2*\i,-0.25+\jj) rectangle (.25-2*\i,-0.75+\jj);
\draw[thick] (-2*\i,-0.35+\jj) -- (-2*\i+0.15,-0.35+\jj) -- (-2*\i+0.15,-0.5+\jj);%
}
\foreach \i in {3+\shiftx}{
\draw[thick] (-.5-2*\i,-1+\jj) -- (0.525-2*\i,0.025+\jj);
\draw[thick] (-0.525-2*\i,0.025+\jj) -- (0.5-2*\i,-1+\jj);
\draw[thick, fill=myred7, rounded corners=2pt] (-0.25-2*\i,-0.25+\jj) rectangle (.25-2*\i,-0.75+\jj);
\draw[thick] (-2*\i,-0.35+\jj) -- (-2*\i+0.15,-0.35+\jj) -- (-2*\i+0.15,-0.5+\jj);%
}
\foreach \i in {4+\shiftx}{
\draw[thick] (-.5-2*\i,-1+\jj) -- (0.525-2*\i,0.025+\jj);
\draw[thick] (-0.525-2*\i,0.025+\jj) -- (0.5-2*\i,-1+\jj);
\draw[thick, fill=myred6, rounded corners=2pt] (-0.25-2*\i,-0.25+\jj) rectangle (.25-2*\i,-0.75+\jj);
\draw[thick] (-2*\i,-0.35+\jj) -- (-2*\i+0.15,-0.35+\jj) -- (-2*\i+0.15,-0.5+\jj);%
}
}
\foreach \jj in {2+\shifty}{
\foreach \i in {0+\shiftx}{
\draw[thick] (-.5-2*\i,-1+\jj) -- (0.525-2*\i,0.025+\jj);
\draw[thick] (-0.525-2*\i,0.025+\jj) -- (0.5-2*\i,-1+\jj);
\draw[thick, fill=myred9, rounded corners=2pt] (-0.25-2*\i,-0.25+\jj) rectangle (.25-2*\i,-0.75+\jj);
\draw[thick] (-2*\i,-0.35+\jj) -- (-2*\i+0.15,-0.35+\jj) -- (-2*\i+0.15,-0.5+\jj);%
}
\foreach \i in {1+\shiftx}{
\draw[thick] (-.5-2*\i,-1+\jj) -- (0.525-2*\i,0.025+\jj);
\draw[thick] (-0.525-2*\i,0.025+\jj) -- (0.5-2*\i,-1+\jj);
\draw[thick, fill=myred10, rounded corners=2pt] (-0.25-2*\i,-0.25+\jj) rectangle (.25-2*\i,-0.75+\jj);
\draw[thick] (-2*\i,-0.35+\jj) -- (-2*\i+0.15,-0.35+\jj) -- (-2*\i+0.15,-0.5+\jj);%
}
\foreach \i in {2+\shiftx}{
\draw[thick] (-.5-2*\i,-1+\jj) -- (0.525-2*\i,0.025+\jj);
\draw[thick] (-0.525-2*\i,0.025+\jj) -- (0.5-2*\i,-1+\jj);
\draw[thick, fill=myred2, rounded corners=2pt] (-0.25-2*\i,-0.25+\jj) rectangle (.25-2*\i,-0.75+\jj);
\draw[thick] (-2*\i,-0.35+\jj) -- (-2*\i+0.15,-0.35+\jj) -- (-2*\i+0.15,-0.5+\jj);%
}
\foreach \i in {3+\shiftx}{
\draw[thick] (-.5-2*\i,-1+\jj) -- (0.525-2*\i,0.025+\jj);
\draw[thick] (-0.525-2*\i,0.025+\jj) -- (0.5-2*\i,-1+\jj);
\draw[thick, fill=myred6, rounded corners=2pt] (-0.25-2*\i,-0.25+\jj) rectangle (.25-2*\i,-0.75+\jj);
\draw[thick] (-2*\i,-0.35+\jj) -- (-2*\i+0.15,-0.35+\jj) -- (-2*\i+0.15,-0.5+\jj);%
}
\foreach \i in {4+\shiftx}{
\draw[thick] (-.5-2*\i,-1+\jj) -- (0.525-2*\i,0.025+\jj);
\draw[thick] (-0.525-2*\i,0.025+\jj) -- (0.5-2*\i,-1+\jj);
\draw[thick, fill=myred7, rounded corners=2pt] (-0.25-2*\i,-0.25+\jj) rectangle (.25-2*\i,-0.75+\jj);
\draw[thick] (-2*\i,-0.35+\jj) -- (-2*\i+0.15,-0.35+\jj) -- (-2*\i+0.15,-0.5+\jj);%
}
}
\foreach \jj in {4+\shifty}{
\foreach \i in {0+\shiftx}{
\draw[thick] (-.5-2*\i,-1+\jj) -- (0.525-2*\i,0.025+\jj);
\draw[thick] (-0.525-2*\i,0.025+\jj) -- (0.5-2*\i,-1+\jj);
\draw[thick, fill=myred4, rounded corners=2pt] (-0.25-2*\i,-0.25+\jj) rectangle (.25-2*\i,-0.75+\jj);
\draw[thick] (-2*\i,-0.35+\jj) -- (-2*\i+0.15,-0.35+\jj) -- (-2*\i+0.15,-0.5+\jj);%
}
\foreach \i in {1+\shiftx}{
\draw[thick] (-.5-2*\i,-1+\jj) -- (0.525-2*\i,0.025+\jj);
\draw[thick] (-0.525-2*\i,0.025+\jj) -- (0.5-2*\i,-1+\jj);
\draw[thick, fill=myred3, rounded corners=2pt] (-0.25-2*\i,-0.25+\jj) rectangle (.25-2*\i,-0.75+\jj);
\draw[thick] (-2*\i,-0.35+\jj) -- (-2*\i+0.15,-0.35+\jj) -- (-2*\i+0.15,-0.5+\jj);%
}
\foreach \i in {2+\shiftx}{
\draw[thick] (-.5-2*\i,-1+\jj) -- (0.525-2*\i,0.025+\jj);
\draw[thick] (-0.525-2*\i,0.025+\jj) -- (0.5-2*\i,-1+\jj);
\draw[thick, fill=myred2, rounded corners=2pt] (-0.25-2*\i,-0.25+\jj) rectangle (.25-2*\i,-0.75+\jj);
\draw[thick] (-2*\i,-0.35+\jj) -- (-2*\i+0.15,-0.35+\jj) -- (-2*\i+0.15,-0.5+\jj);%
}
\foreach \i in {3+\shiftx}{
\draw[thick] (-.5-2*\i,-1+\jj) -- (0.525-2*\i,0.025+\jj);
\draw[thick] (-0.525-2*\i,0.025+\jj) -- (0.5-2*\i,-1+\jj);
\draw[thick, fill=myred5, rounded corners=2pt] (-0.25-2*\i,-0.25+\jj) rectangle (.25-2*\i,-0.75+\jj);
\draw[thick] (-2*\i,-0.35+\jj) -- (-2*\i+0.15,-0.35+\jj) -- (-2*\i+0.15,-0.5+\jj);%
}
\foreach \i in {4+\shiftx}{
\draw[thick] (-.5-2*\i,-1+\jj) -- (0.525-2*\i,0.025+\jj);
\draw[thick] (-0.525-2*\i,0.025+\jj) -- (0.5-2*\i,-1+\jj);
\draw[thick, fill=myred10, rounded corners=2pt] (-0.25-2*\i,-0.25+\jj) rectangle (.25-2*\i,-0.75+\jj);
\draw[thick] (-2*\i,-0.35+\jj) -- (-2*\i+0.15,-0.35+\jj) -- (-2*\i+0.15,-0.5+\jj);%
}
}
\def\shiftx{1}
\def\shifty{-9}
\foreach \jj[evaluate=\jj as \j using -2*(ceil(\jj/2)-\jj/2)] in {0+\shifty}{
\foreach \i in {1+\shiftx}
{
\draw[thick] (.5-2*\i,-2-1*\jj) -- (1-2*\i,-1.5-\jj);
\draw[thick] (1-2*\i,-1.5-1*\jj) -- (1.5-2*\i,-2-\jj);
\draw[thick] (.5-2*\i,-1-1*\jj) -- (1-2*\i,-1.5-\jj);
\draw[thick] (1-2*\i,-1.5-1*\jj) -- (1.5-2*\i,-1-\jj);
\draw[thick, fill=myred5, rounded corners=2pt] (0.75-2*\i,-1.75-\jj) rectangle (1.25-2*\i,-1.25-\jj);
\draw[thick] (-2*\i+1,-1.35-\jj) -- (-2*\i+1.15,-1.35-\jj) -- (-2*\i+1.15,-1.5-\jj);%
}
\foreach \i in {2+\shiftx}
{
\draw[thick] (.5-2*\i,-2-1*\jj) -- (1-2*\i,-1.5-\jj);
\draw[thick] (1-2*\i,-1.5-1*\jj) -- (1.5-2*\i,-2-\jj);
\draw[thick] (.5-2*\i,-1-1*\jj) -- (1-2*\i,-1.5-\jj);
\draw[thick] (1-2*\i,-1.5-1*\jj) -- (1.5-2*\i,-1-\jj);
\draw[thick, fill=myred4, rounded corners=2pt] (0.75-2*\i,-1.75-\jj) rectangle (1.25-2*\i,-1.25-\jj);
\draw[thick] (-2*\i+1,-1.35-\jj) -- (-2*\i+1.15,-1.35-\jj) -- (-2*\i+1.15,-1.5-\jj);%
}
\foreach \i in {3+\shiftx}
{
\draw[thick] (.5-2*\i,-2-1*\jj) -- (1-2*\i,-1.5-\jj);
\draw[thick] (1-2*\i,-1.5-1*\jj) -- (1.5-2*\i,-2-\jj);
\draw[thick] (.5-2*\i,-1-1*\jj) -- (1-2*\i,-1.5-\jj);
\draw[thick] (1-2*\i,-1.5-1*\jj) -- (1.5-2*\i,-1-\jj);
\draw[thick, fill=myred3, rounded corners=2pt] (0.75-2*\i,-1.75-\jj) rectangle (1.25-2*\i,-1.25-\jj);
\draw[thick] (-2*\i+1,-1.35-\jj) -- (-2*\i+1.15,-1.35-\jj) -- (-2*\i+1.15,-1.5-\jj);%
}
\foreach \i in {4+\shiftx}
{
\draw[thick] (.5-2*\i,-2-1*\jj) -- (1-2*\i,-1.5-\jj);
\draw[thick] (1-2*\i,-1.5-1*\jj) -- (1.5-2*\i,-2-\jj);
\draw[thick] (.5-2*\i,-1-1*\jj) -- (1-2*\i,-1.5-\jj);
\draw[thick] (1-2*\i,-1.5-1*\jj) -- (1.5-2*\i,-1-\jj);
\draw[thick, fill=myred2, rounded corners=2pt] (0.75-2*\i,-1.75-\jj) rectangle (1.25-2*\i,-1.25-\jj);
\draw[thick] (-2*\i+1,-1.35-\jj) -- (-2*\i+1.15,-1.35-\jj) -- (-2*\i+1.15,-1.5-\jj);%
}
\foreach \i in {5+\shiftx}
{
\draw[thick] (.5-2*\i,-2-1*\jj) -- (1-2*\i,-1.5-\jj);
\draw[thick] (1-2*\i,-1.5-1*\jj) -- (1.5-2*\i,-2-\jj);
\draw[thick] (.5-2*\i,-1-1*\jj) -- (1-2*\i,-1.5-\jj);
\draw[thick] (1-2*\i,-1.5-1*\jj) -- (1.5-2*\i,-1-\jj);
\draw[thick, fill=myred1, rounded corners=2pt] (0.75-2*\i,-1.75-\jj) rectangle (1.25-2*\i,-1.25-\jj);
\draw[thick] (-2*\i+1,-1.35-\jj) -- (-2*\i+1.15,-1.35-\jj) -- (-2*\i+1.15,-1.5-\jj);%
}
}
\foreach \jj[evaluate=\jj as \j using -2*(ceil(\jj/2)-\jj/2)] in {-2+\shifty}{
\foreach \i in {1+\shiftx}
{
\draw[thick] (.5-2*\i,-2-1*\jj) -- (1-2*\i,-1.5-\jj);
\draw[thick] (1-2*\i,-1.5-1*\jj) -- (1.5-2*\i,-2-\jj);
\draw[thick] (.5-2*\i,-1-1*\jj) -- (1-2*\i,-1.5-\jj);
\draw[thick] (1-2*\i,-1.5-1*\jj) -- (1.5-2*\i,-1-\jj);
\draw[thick, fill=myred8, rounded corners=2pt] (0.75-2*\i,-1.75-\jj) rectangle (1.25-2*\i,-1.25-\jj);
\draw[thick] (-2*\i+1,-1.35-\jj) -- (-2*\i+1.15,-1.35-\jj) -- (-2*\i+1.15,-1.5-\jj);%
}
\foreach \i in {2+\shiftx}
{
\draw[thick] (.5-2*\i,-2-1*\jj) -- (1-2*\i,-1.5-\jj);
\draw[thick] (1-2*\i,-1.5-1*\jj) -- (1.5-2*\i,-2-\jj);
\draw[thick] (.5-2*\i,-1-1*\jj) -- (1-2*\i,-1.5-\jj);
\draw[thick] (1-2*\i,-1.5-1*\jj) -- (1.5-2*\i,-1-\jj);
\draw[thick, fill=myred10, rounded corners=2pt] (0.75-2*\i,-1.75-\jj) rectangle (1.25-2*\i,-1.25-\jj);
\draw[thick] (-2*\i+1,-1.35-\jj) -- (-2*\i+1.15,-1.35-\jj) -- (-2*\i+1.15,-1.5-\jj);%
}
\foreach \i in {3+\shiftx}
{
\draw[thick] (.5-2*\i,-2-1*\jj) -- (1-2*\i,-1.5-\jj);
\draw[thick] (1-2*\i,-1.5-1*\jj) -- (1.5-2*\i,-2-\jj);
\draw[thick] (.5-2*\i,-1-1*\jj) -- (1-2*\i,-1.5-\jj);
\draw[thick] (1-2*\i,-1.5-1*\jj) -- (1.5-2*\i,-1-\jj);
\draw[thick, fill=myred6, rounded corners=2pt] (0.75-2*\i,-1.75-\jj) rectangle (1.25-2*\i,-1.25-\jj);
\draw[thick] (-2*\i+1,-1.35-\jj) -- (-2*\i+1.15,-1.35-\jj) -- (-2*\i+1.15,-1.5-\jj);%
}
\foreach \i in {4+\shiftx}
{
\draw[thick] (.5-2*\i,-2-1*\jj) -- (1-2*\i,-1.5-\jj);
\draw[thick] (1-2*\i,-1.5-1*\jj) -- (1.5-2*\i,-2-\jj);
\draw[thick] (.5-2*\i,-1-1*\jj) -- (1-2*\i,-1.5-\jj);
\draw[thick] (1-2*\i,-1.5-1*\jj) -- (1.5-2*\i,-1-\jj);
\draw[thick, fill=myred9, rounded corners=2pt] (0.75-2*\i,-1.75-\jj) rectangle (1.25-2*\i,-1.25-\jj);
\draw[thick] (-2*\i+1,-1.35-\jj) -- (-2*\i+1.15,-1.35-\jj) -- (-2*\i+1.15,-1.5-\jj);%
}
\foreach \i in {5+\shiftx}
{
\draw[thick] (.5-2*\i,-2-1*\jj) -- (1-2*\i,-1.5-\jj);
\draw[thick] (1-2*\i,-1.5-1*\jj) -- (1.5-2*\i,-2-\jj);
\draw[thick] (.5-2*\i,-1-1*\jj) -- (1-2*\i,-1.5-\jj);
\draw[thick] (1-2*\i,-1.5-1*\jj) -- (1.5-2*\i,-1-\jj);
\draw[thick, fill=myred1, rounded corners=2pt] (0.75-2*\i,-1.75-\jj) rectangle (1.25-2*\i,-1.25-\jj);
\draw[thick] (-2*\i+1,-1.35-\jj) -- (-2*\i+1.15,-1.35-\jj) -- (-2*\i+1.15,-1.5-\jj);%
}
}
\foreach \jj[evaluate=\jj as \j using -2*(ceil(\jj/2)-\jj/2)] in {-4+\shifty}{
\foreach \i in {1+\shiftx}
{
\draw[thick] (.5-2*\i,-2-1*\jj) -- (1-2*\i,-1.5-\jj);
\draw[thick] (1-2*\i,-1.5-1*\jj) -- (1.5-2*\i,-2-\jj);
\draw[thick] (.5-2*\i,-1-1*\jj) -- (1-2*\i,-1.5-\jj);
\draw[thick] (1-2*\i,-1.5-1*\jj) -- (1.5-2*\i,-1-\jj);
\draw[thick, fill=myred6, rounded corners=2pt] (0.75-2*\i,-1.75-\jj) rectangle (1.25-2*\i,-1.25-\jj);
\draw[thick] (-2*\i+1,-1.35-\jj) -- (-2*\i+1.15,-1.35-\jj) -- (-2*\i+1.15,-1.5-\jj);%
}
\foreach \i in {2+\shiftx}
{
\draw[thick] (.5-2*\i,-2-1*\jj) -- (1-2*\i,-1.5-\jj);
\draw[thick] (1-2*\i,-1.5-1*\jj) -- (1.5-2*\i,-2-\jj);
\draw[thick] (.5-2*\i,-1-1*\jj) -- (1-2*\i,-1.5-\jj);
\draw[thick] (1-2*\i,-1.5-1*\jj) -- (1.5-2*\i,-1-\jj);
\draw[thick, fill=myred3, rounded corners=2pt] (0.75-2*\i,-1.75-\jj) rectangle (1.25-2*\i,-1.25-\jj);
\draw[thick] (-2*\i+1,-1.35-\jj) -- (-2*\i+1.15,-1.35-\jj) -- (-2*\i+1.15,-1.5-\jj);%
}
\foreach \i in {3+\shiftx}
{
\draw[thick] (.5-2*\i,-2-1*\jj) -- (1-2*\i,-1.5-\jj);
\draw[thick] (1-2*\i,-1.5-1*\jj) -- (1.5-2*\i,-2-\jj);
\draw[thick] (.5-2*\i,-1-1*\jj) -- (1-2*\i,-1.5-\jj);
\draw[thick] (1-2*\i,-1.5-1*\jj) -- (1.5-2*\i,-1-\jj);
\draw[thick, fill=myred7, rounded corners=2pt] (0.75-2*\i,-1.75-\jj) rectangle (1.25-2*\i,-1.25-\jj);
\draw[thick] (-2*\i+1,-1.35-\jj) -- (-2*\i+1.15,-1.35-\jj) -- (-2*\i+1.15,-1.5-\jj);%
}
\foreach \i in {4+\shiftx}
{
\draw[thick] (.5-2*\i,-2-1*\jj) -- (1-2*\i,-1.5-\jj);
\draw[thick] (1-2*\i,-1.5-1*\jj) -- (1.5-2*\i,-2-\jj);
\draw[thick] (.5-2*\i,-1-1*\jj) -- (1-2*\i,-1.5-\jj);
\draw[thick] (1-2*\i,-1.5-1*\jj) -- (1.5-2*\i,-1-\jj);
\draw[thick, fill=myred2, rounded corners=2pt] (0.75-2*\i,-1.75-\jj) rectangle (1.25-2*\i,-1.25-\jj);
\draw[thick] (-2*\i+1,-1.35-\jj) -- (-2*\i+1.15,-1.35-\jj) -- (-2*\i+1.15,-1.5-\jj);%
}
\foreach \i in {5+\shiftx}
{
\draw[thick] (.5-2*\i,-2-1*\jj) -- (1-2*\i,-1.5-\jj);
\draw[thick] (1-2*\i,-1.5-1*\jj) -- (1.5-2*\i,-2-\jj);
\draw[thick] (.5-2*\i,-1-1*\jj) -- (1-2*\i,-1.5-\jj);
\draw[thick] (1-2*\i,-1.5-1*\jj) -- (1.5-2*\i,-1-\jj);
\draw[thick, fill=myred7, rounded corners=2pt] (0.75-2*\i,-1.75-\jj) rectangle (1.25-2*\i,-1.25-\jj);
\draw[thick] (-2*\i+1,-1.35-\jj) -- (-2*\i+1.15,-1.35-\jj) -- (-2*\i+1.15,-1.5-\jj);%
}
}
\end{scope}
\end{tikzpicture}\,\, \Biggr\rangle,
\ee
where we introduced 
\begin{align}
 U_{x,\tau'}^\dag & = 
\begin{tikzpicture}[baseline=(current  bounding  box.center), scale=.7]
\draw[ thick] (-4.25,0.5) -- (-3.25,-0.5);
\draw[ thick] (-4.25,-0.5) -- (-3.25,0.5);
\draw[ thick, fill=myred, rounded corners=2pt] (-4,0.25) rectangle (-3.5,-0.25);
\draw[thick] (-3.75,0.15) -- (-3.75+0.15,0.15) -- (-3.75+0.15,0);
\Text[x=-4.25,y=-0.75]{}
\end{tikzpicture}\, = (U_{x,\tau})^\dagger,
\end{align}
with $\tau'=t-\tau-\frac{1}{2}$.
Bending the upper part (blue gates) underneath the lower one (red gates) and using that the average factorises at each space-time point (positions of the gates in the two copies match) we find
\be
K_g(t) = 
\begin{tikzpicture}[baseline=(current  bounding  box.center), scale=0.55]
\def\eps{0};
\def\shift{11}
\def\shifty{-2.5}
\foreach \i in {1,...,5}{
\draw[very thick, opacity =0.2,dashed] (2*\i+2-1.5+0.25,-2.5-0.1) -- (2*\i+2-1.5+0.25,3.5-0.1);
\draw[very thick, opacity =0.2,dashed] (2*\i+2-0.5-0.25,-2.5-0.1) -- (2*\i+2-0.5-0.25,3.5-0.1);
}
\foreach \i in{3,...,4}{
\draw[very thick] (\shift-1.5+2,2*\i-0.5-3.5+\shifty) arc (45:-90:0.15);
\draw[very thick] (\shift-10+0.5+0,2*\i-0.5-3.5+\shifty) arc (-135:-270:0.15);}

\foreach \i in{2,...,4}{
\draw[very thick] (\shift-1.5+2,1+2*\i-0.5-3.5+\shifty) arc (-45:90:0.15);
\draw[very thick] (\shift-10+0.5,1+2*\i-0.5-3.5+\shifty) arc (135:270:0.15);}
\foreach \i in {0,...,4}
{
\draw[very thick] (\shift-.5-2*\i,1+\shifty) -- (\shift+0.5-2*\i,0+\shifty);
\draw[very thick] (\shift-0.5-2*\i,0+\shifty) -- (\shift+0.5-2*\i,1+\shifty);
\draw[ thick, fill=mygreen, rounded corners=2pt] (\shift-0.25-2*\i,0.25+\shifty) rectangle (\shift+.25-2*\i,0.75+\shifty);
\draw[thick] (\shift-2*\i,0.65+\shifty) -- (\shift+.15-2*\i,.65+\shifty) -- (\shift+.15-2*\i,0.5+\shifty);
}
\foreach \i in {1,...,5}
{
\draw[very thick] (2*\i+2-1.5,3.5) arc (135:-0:0.15);
\draw[very thick] (2*\i+2-2.5,3.5) arc (-325:-180:0.15);
\draw[very thick] (2*\i+2-1.5,-2.5) arc (-45:180:-0.15);
\draw[very thick] (2*\i+2-0.5,-2.5) arc (45:-180:0.15);
}
\foreach \i in {1,...,4}{
\draw[very thick] (\shift+.5-2*\i,6+\shifty) -- (\shift+1-2*\i,5.5+\shifty);
\draw[very thick] (\shift+1.5-2*\i,6+\shifty) -- (\shift+1-2*\i,5.5+\shifty);}

\foreach \jj[evaluate=\jj as \j using -2*(ceil(\jj/2)-\jj/2)] in {0,...,4}
\foreach \i in {1,...,5}
{
\draw[very thick] (\shift+.5-2*\i-1*\j,2+1*\jj+\shifty) -- (\shift+1-2*\i-1*\j,1.5+\jj+\shifty);
\draw[very thick] (\shift+1-2*\i-1*\j,1.5+1*\jj+\shifty) -- (\shift+1.5-2*\i-1*\j,2+\jj+\shifty);
}
%
\foreach \jj[evaluate=\jj as \j using -2*(ceil(\jj/2)-\jj/2)] in {0,...,4}
\foreach \i in {1,...,5}
{
\draw[very thick] (\shift+.5-2*\i-1*\j,1+1*\jj+\shifty) -- (\shift+1-2*\i-1*\j,1.5+\jj+\shifty);
\draw[very thick] (\shift+1-2*\i-1*\j,1.5+1*\jj+\shifty) -- (\shift+1.5-2*\i-1*\j,1+\jj+\shifty);
\draw[ thick, fill=mygreen, rounded corners=2pt] (\shift+0.75-2*\i-1*\j,1.75+\jj+\shifty) rectangle (\shift+1.25-2*\i-1*\j,1.25+\jj+\shifty);
\draw[thick] (\shift+1-2*\i-1*\j,1.65+1*\jj+\shifty) -- (\shift+1.15-2*\i-1*\j,1.65+1*\jj+\shifty) -- (\shift+1.15-2*\i-1*\j,1.5+1*\jj+\shifty);
} 
\end{tikzpicture}\, ,
\label{eq:SFF0SM}
\vspace{-5mm}
\ee
where we introduced the averaged ``double gates'' 
\begin{equation}
\label{eq:doublegate}
W_{x,\tau} =
\begin{tikzpicture}[baseline=(current  bounding  box.center), scale=.7]
\def\eps{0.5};
\Wgategreen{-3.75}{0};
\Text[x=-3.75,y=-0.6]{}
\end{tikzpicture}
=\braket{
\begin{tikzpicture}[baseline=(current  bounding  box.center), scale=.7]
\draw[thick] (-1.65,0.65) -- (-0.65,-0.35);
\draw[thick] (-1.65,-0.35) -- (-0.65,0.65);
\draw[ thick, fill=myblue, rounded corners=2pt] (-1.4,0.4) rectangle (-.9,-0.1);
\draw[thick] (-1.15,0) -- (-1,0) -- (-1,0.15);
\draw[thick] (-2.25,0.5) -- (-1.25,-0.5);
\draw[thick] (-2.25,-0.5) -- (-1.25,0.5);
\draw[ thick, fill=myred, rounded corners=2pt] (-2,0.25) rectangle (-1.5,-0.25);
\draw[thick] (-1.75,0.15) -- (-1.6,0.15) -- (-1.6,0);
\Text[x=-2.25,y=-0.6]{}
\end{tikzpicture}}
= \braket{U^{\dag}_{x,\tau}\otimes_c U_{x,\tau}^{T}}.
\end{equation}
Here the tensor product $\otimes_c$ denotes the tensor product between the two copies of the systems corresponding to the backward and forward time evolution (respectively the one given by red and blue gates in \eqref{eq:fullGSFF}).

Let us now consider the local gate given by Eq.~\eqref{eq:Ubar}. Noting that each single site gate with random field can be split into two random gates $e^{i \bar{\phi}_{x,\tau}\sigma^{z}} = e^{i \phi_{x,\tau}^\prime \sigma^{z}} \cdot e^{i \phi_{x,\tau}\sigma^{z}}$ ($\bar{\phi}_{x,\tau}$ denotes the random strength of the original field in Eq.~\eqref{eq:Ubar} and $\phi_{x,\tau},\phi_{x,\tau}^\prime$ are new random fields) we can write 
\be
U_{x,\tau} = (e^{i \phi_{x,\tau} \sigma^z} \otimes_s e^{i \phi_{\scriptscriptstyle x+{1}/{2},\tau} \sigma^{z}})\cdot U \cdot (e^{i \phi_{x,\tau}^\prime \sigma^z} \otimes_s e^{i \phi_{\scriptscriptstyle x+{1}/{2},\tau}^\prime \sigma^{z}}),
\label{eq:UbarSM}
\ee
where $\otimes_s$ denotes the tensor product between neighboring sites. Plugging in \eqref{eq:doublegate} we find 
\begin{align}
W=& \langle\left[
(e^{-i \phi_{x,\tau}^\prime \sigma^z} \!\!\otimes_s\! e^{-i \phi_{\scriptscriptstyle x+{1}/{2},\tau}^\prime \sigma^{z}}) U^\dagger  (e^{-i \phi_{x,\tau}\sigma^z} \!\!\otimes_s\! e^{-i \phi_{\scriptscriptstyle x+{1}/{2},\tau} \sigma^{z}}) \right]\notag\\
&\qquad\qquad\otimes_c \!\!\left[
(e^{i \phi_{x,\tau}^\prime \sigma^z} \!\!\otimes_s\! e^{i \phi_{\scriptscriptstyle x+{1}/{2},\tau}^\prime \sigma^{z}}) U^T (e^{i \phi_{x,\tau} \sigma^z} \!\!\otimes_s\! e^{i \phi_{\scriptscriptstyle x+{1}/{2},\tau} \sigma^{z}})\right]\rangle.
\end{align}
We omitted the subscript $_{x,\tau}$ in $W$, as the gate does not depend on the position for our choice of $U_{x,\tau}$. This follows from the fact that we average over $\phi_{x,\tau}, \phi_{x,\tau}^\prime,\phi_{\scriptscriptstyle x+{1}/{2},\tau}, \phi_{\scriptscriptstyle x+{1}/{2},\tau}^\prime$.
Defining $z(\phi_{x,\tau}):=e^{-i \phi_{ x,\tau}\sigma^z} \otimes_c e^{i \phi_{ x,\tau} \sigma^{z}}$ and simplifying the tensor products we get
%
\begin{align}
W=\! \braket{
\left[z(\phi_{x,\tau}^\prime) \otimes_s z(\phi_{\scriptscriptstyle x+{1}/{2},\tau}^\prime) \right]\cdot
\left[ U^\dagger \otimes_c U^T \right] \cdot
\left[z(\phi_{x,\tau}) \otimes_s z(\phi_{\scriptscriptstyle x+{1}/{2},\tau})
 \right] }.
\end{align}
%
Next we use that average factorizes and that
\begin{align}
\braket{z(\phi_{x,\tau})}=\braket{ e^{i \phi_{x,\tau} \sigma^z} \otimes_c e^{-i \phi_{x,\tau} \sigma^z}} = \frac{\1\otimes_c\1+\sigma^z\otimes_c \sigma^z}{2}\equiv P_z\,
\end{align}
is a projector. In particular $P_z$ projects the 4-dimensional space of states of a wire onto a 2-dimensional subspace, which dub `reduced local space'. Concretely, it projects the single site basis 
\be
\{\ket{\1},\ket{ \sigma^x},\ket{\sigma^y},\ket{ \sigma^z}\},
\ee 
to the basis of diagonal operators 
\be
\{\ket{\1}, \ket{\sigma^z}\}=\{\ket{\mcirc}, \ket{\mcircf}\},
\ee 
where we defined vectorised operators as follows 
\be
\ket{O} = \sum_{s_1,s_2\in\{0,1\}} [O]_{s_1,s_2}\ket{s_1}\otimes_c\ket{s_2},
\ee
where $\{\ket{0},\ket{1}\}$ is a basis of the local Hilbert space. 
Therefore, we define a \emph{reduced} folded gate $w$, which is $W$ written in this reduced Hilbert space   
\be w := \left((P_z\otimes_s P_z)\cdot (U^\dagger \otimes_c U^T) \cdot (P_z \otimes_s P_z)\right)_{\rm red} =
\begin{tikzpicture}[baseline={([yshift=-.5ex]current bounding box.center)}, scale=.7]
\Wreduced{0}{0};
\end{tikzpicture}\;,
\label{eq:reducedExplicit}
\ee
where  $_{\rm red}$ denoteds an operator that acts on the reduced Hilbert space of diagonal operators. $w$ thus acts on two sites with two degrees of freedom and is a four by four matrix. 

Next, we use an explicit parametrization of general four by four unitary matrix for our gate $U$, resulting in explicitly parametrized reduced gate $w$ from Eq.~\eqref{eq:reducedgates} of the main text. This is done explicitly in Appendix B of~\cite{KBP:CorrelationsPerturbed}, which we do not repeat here. 

Let us comment what happens in the presence of a conservation law. Specifically, by demanding that the local gate $U$ conserve magnetization (in the z-direction), we can explicitly parametrize it as
\be
U_{\rm U(1)} = \begin{pmatrix}
    1 & 0 & 0 & 0\\
0 & \cos J & i \sin J & 0\\
0 & i \sin J &\cos J & 0\\
0 & 0 & 0 &  1
    \end{pmatrix}\!,\,\qquad\qquad J\in[0,{\pi}/{4}]\,.
\ee
Inserting this parametrization into equation~\eqref{eq:reducedExplicit} results in Eq.~\eqref{eq:redCon} of the main text.

\section{Proof of Property~\ref{prop:P1}}
\label{app:bound}
In this appendix we prove Property~\ref{prop:P1}, namely we show that for no splittings $e=f=0$ (or no mergers $b=d=0$) and non-negative weights the 2$m$-point correlation functions $\braket{\mcircf_{x_1}\cdots\mcircf_{x_m} | \!\mcircf_{y_1}\cdots\mcircf_{y_m}\!(t) }_L$ are bounded by 
\be
\braket{\mcircf_{x_1}\cdots\mcircf_{x_m} | \!\mcircf_{y_1}\cdots\mcircf_{y_m}\!(t) }_L < \left({\rm max} (1, \frac{g}{\varepsilon_1 \varepsilon_2+ a c}) \right)^{(m-1)t} \sum_{\sigma\in S_{m}}\prod_{i=1}^{m} \braket{\mcircf_{x_i} | \!\mcircf_{y_{\sigma(i)}}\!(t) }_L\, ,
\label{eq:Cmbound}
\ee
where $S_{m}$ is the permutation group of $m$ elements.

We start by expressing the correlation functions as the sum of contributions from allowed configurations, as discussed in Section V A of \cite{KBP:CorrelationsPerturbed}. Inserting a resolution of the identity $\1 = \ket{\mcirc}\!\!\bra{\mcirc}+  \ket{\mcircf}\!\!\bra{\mcircf}$ at each reduced-operator wire, we can explicitly decompose each 
\be
\braket{\mcircf_{x_1}\cdots\mcircf_{x_m} | \!\mcircf_{y_1}\cdots\mcircf_{y_m}\!(t) }_L = \begin{tikzpicture}[baseline=(current  bounding  box.center), scale=0.65]
\def\eps{0};
\def\shift{11}
\def\shifty{-2.5}
\foreach \i in{3,...,4}
{
\draw[thin] (\shift-1.5,2*\i-0.5-3.5+\shifty) arc (45:-90:0.15);
\draw[thin] (\shift-10+0.5+0,2*\i-0.5-3.5+\shifty) arc (-135:-270:0.15);}
\foreach \i in{2,...,4}
{
\draw[thin] (\shift-1.5,1+2*\i-0.5-3.5+\shifty) arc (-45:90:0.15);
\draw[thin] (\shift-10+0.5,1+2*\i-0.5-3.5+\shifty) arc (135:270:0.15);
}
\foreach \i in {1,...,4}{
\Wreduced{\shift-2*\i}{\shifty+0.5}
}
\foreach \jj[evaluate=\jj as \j using -2*(ceil(\jj/2)-\jj/2)] in {0,...,4}
\foreach \i in {2,...,5}
{
\Wreduced{\shift+1.-2*\i-1*\j}{1.5+\jj+\shifty}
}
\draw[ thick, fill=black] (\shift-8.5,\shifty) circle (0.1cm); 
\draw[ thick, fill=white] (\shift-7.5,\shifty) circle (0.1cm);
\draw[ thick, fill=white] (\shift-6.5,\shifty) circle (0.1cm); 
\draw[ thick, fill=black] (\shift-5.5,\shifty) circle (0.1cm); 
\draw[ thick, fill=white] (\shift-4.5,\shifty) circle (0.1cm); 
\draw[ thick, fill=white] (\shift-3.5,\shifty) circle (0.1cm); 
\draw[ thick, fill=black] (\shift-2.5,\shifty) circle (0.1cm); 
\draw[ thick, fill=black] (\shift-1.5,\shifty) circle (0.1cm);

\draw[ thick, fill=white] (\shift-8.5,\shifty+6) circle (0.1cm); 
\draw[ thick, fill=black] (\shift-7.5,\shifty+6) circle (0.1cm);
\draw[ thick, fill=white] (\shift-6.5,\shifty+6) circle (0.1cm); 
\draw[ thick, fill=black] (\shift-5.5,\shifty+6) circle (0.1cm); 
\draw[ thick, fill=white] (\shift-4.5,\shifty+6) circle (0.1cm); 
\draw[ thick, fill=black] (\shift-3.5,\shifty+6) circle (0.1cm); 
\draw[ thick, fill=white] (\shift-2.5,\shifty+6) circle (0.1cm); 
\draw[ thick, fill=black] (\shift-9.5,\shifty+6) circle (0.1cm); 
\Text[x=6,y=-3]{$\cdots$}
\Text[x=2.5,y=-3]{$x_1$}
\Text[x=10,y=-3]{$x_m$}
\Text[x=8.5,y=-3]{$x_{m-1}$}
\Text[x=6,y=3.85]{$\cdots$}
\Text[x=3.5,y=3.85]{$y_1$}
\Text[x=1.5,y=3.85]{$y_m$}
\Text[x=7.5,y=3.85]{$y_{m-1}$}
\end{tikzpicture},
\label{eq:aveCorr}
\ee
into the sum of $2^{4L t}$ terms. Configurations are expressed in terms of ``tiles'' where we connect the particles $\mcircf$ with solid lines and ignore the vacancies $\mcirc$ \cite{KBP:CorrelationsPerturbed}. For example,
\be
\begin{tikzpicture}[baseline=(current  bounding  box.center), scale=0.8]
\Text[x=-5.5,y=0.0, anchor = center]{$g$}
\Text[x=-5,y=0.0, anchor = center]{$=$}
\Text[x=-2.5,y=0.0, anchor = center]{$=$}
\begin{scope}[rotate around={-0:(-3.75,0)}]
\draw (-4.25,0.5) -- (-3.25,-0.5);
\draw (-4.25,-0.5) -- (-3.25,0.5);
\draw[ thick, fill=myY, rounded corners=2pt] (-4,0.25) rectangle (-3.5,-0.25);
\draw[thick] (-3.75,0.15) -- (-3.6,0.15) -- (-3.6,0);
\draw[thick, fill=black] (-4.25,0.5) circle (0.1cm); 
\draw[thick, fill=black] (-3.25,0.5) circle (0.1cm); 
\draw[thick, fill=black] (-4.25,-0.5) circle (0.1cm); 
\draw[thick, fill=black] (-3.25,-0.5) circle (0.1cm); 
\end{scope}
\end{tikzpicture}
\begin{tikzpicture}[baseline=(current  bounding  box.center), scale=0.8]
\begin{scope}[shift={(1.5*\sroot,\sroot/2)},rotate=45] \gGate{0}{0} \end{scope}
\end{tikzpicture}.
\label{eq:introducingTiles}
\ee
The complete set of allowed tiles (corresponding to non-zero coefficients of the gate \eqref{eq:reducedgates}), when there are no splittings or mergers, is given by:
\begin{align}
\begin{tikzpicture}[baseline=(current  bounding  box.center), scale=0.7]
\begin{scope}[shift={(0,1.5)},rotate=45]\oGate{0}{0} \end{scope}
\begin{scope}[shift={(3,1.5)},rotate=45]\aGate{0}{0} \end{scope}
\begin{scope}[shift={(6,1.5)},rotate=45]\cGate{0}{0} \end{scope}
\begin{scope}[shift={(0,0)},rotate=45]\doGate{0}{0} \end{scope}
\begin{scope}[shift={(3,0)},rotate=45]\dtGate{0}{0} \end{scope}
\begin{scope}[shift={(6,0)},rotate=45]\gGate{0}{0} \end{scope}
\Text[x=1.25,y=1.5, anchor = center]{$=1$}
\Text[x={3+1.25},y=1.5, anchor = center]{$=a$}
\Text[x={6+1.25},y=1.5, anchor = center]{$=c$}
\Text[x={1.25+6},y=-0, anchor = center]{$=g$}
\Text[x={0+1.25},y=0.0, anchor = center]{$=\varepsilon_1$}
\Text[x={3+1.25},y=0., anchor = center]{$=\varepsilon_2$}
\end{tikzpicture} \, .
\label{eq:tiles}
\end{align}
Note that, if there are no splittings(mergers), and we need to end up with the same number of $\mcircf$, we can not have any mergers(splittings). We start with $m=2$. The four-point correlation function can be expressed as the sum of the weights of configurations $C_i^4$
\be
\braket{\mcircf_{x_1}\mcircf_{x_2} | \!\mcircf_{y_1}\mcircf_{y_2}\!(t) }_L = \sum_i w(C_i^4)\, ,
\ee
where configurations $C_i^4$ have four fixed boundary conditions at $x_1 ,x_2; y_1, y_2$. 
For example:
\begin{align}
C_e^4 =
\begin{tikzpicture}[baseline={([yshift=-.6ex]current bounding box.center)},  scale=0.65]
\begin{scope}[shift={(0,0)},rotate=45] \oGate{0}{0} \end{scope}
\begin{scope}[shift={(\sroot,0)},rotate=45] \doGate{0}{0} \end{scope}
\begin{scope}[shift={(2*\sroot,0)},rotate=45] \cGate{0}{0} \end{scope}
\begin{scope}[shift={(\sroot/2,\sroot/2)},rotate=45] \oGate{0}{0} \end{scope}
\begin{scope}[shift={(1.5*\sroot,\sroot/2)},rotate=45] \gGate{0}{0} \end{scope}
\begin{scope}[shift={(2.5*\sroot,\sroot/2)},rotate=45] \oGate{0}{0} \end{scope}
\begin{scope}[shift={(0,\sroot)},rotate=45] \oGate{0}{0} \end{scope}
\begin{scope}[shift={(\sroot,\sroot)},rotate=45] \doGate{0}{0} \end{scope}
\begin{scope}[shift={(2*\sroot,\sroot)},rotate=45] \aGate{0}{0} \end{scope}
\begin{scope}[shift={(\sroot/2,1.5* \sroot)},rotate=45] \oGate{0}{0} \end{scope}
\begin{scope}[shift={(1.5*\sroot,1.5*\sroot)},rotate=45] \dtGate{0}{0} \end{scope}
\begin{scope}[shift={(2.5*\sroot,1.5*\sroot)},rotate=45] \dtGate{0}{0} \end{scope}
\end{tikzpicture} \, ,
\end{align}
where $x_1=y_1= 3/2$ and $x_2=y_2= 5/2$. The $w(C)$ is the weight of the configuration, which is a product of the weights of all tiles. For instance, the weight of the tile $w(\text{g tile})=g$ and $w(C_e^4) = g a c \varepsilon_1^2 \varepsilon_2^2$. 

Next, we define a map $\mathcal{F}$, which maps configuration to a set of configurations
\be
\mathcal{F}(C_i^4) = \{ \tilde{C}^4_{i,k} \}_{k=1}^{2^n} \, ,
\ee
where $n$ is the number of $g$ tiles in $C_i^4$. It maps each tile $g$ to either tile $a c$ or tile $\varepsilon_1 \varepsilon_2$
resulting in $2^n$ different configurations
\be
\mathcal{F}(
\begin{tikzpicture}[baseline={([yshift=-.6ex]current bounding box.center)},  scale=0.65]
\begin{scope}[shift={(0,0)},rotate=45] \gGate{0}{0} \end{scope}
\end{tikzpicture})
=\left\{
\begin{tikzpicture}[baseline={([yshift=-.6ex]current bounding box.center)},  scale=0.65]
\begin{scope}[shift={(0,0)},rotate=45] \eeGate{0}{0} \end{scope}
\end{tikzpicture},
\begin{tikzpicture}[baseline={([yshift=-.6ex]current bounding box.center)},  scale=0.65]
\begin{scope}[shift={(0,0)},rotate=45] \acGate{0}{0} \end{scope}
\end{tikzpicture} \right\} \,.
\ee
Configurations denoted by tilde $\tilde{}$ have ``new tiles'' $ac$ and  $\varepsilon_1 \varepsilon_2$ and no $g$ tiles.
\begin{align}
\mathcal{F}(C_e^4) = \left\{
\begin{tikzpicture}[baseline={([yshift=-.6ex]current bounding box.center)},  scale=0.65]
\begin{scope}[shift={(0,0)},rotate=45] \oGate{0}{0} \end{scope}
\begin{scope}[shift={(\sroot,0)},rotate=45] \doGate{0}{0} \end{scope}
\begin{scope}[shift={(2*\sroot,0)},rotate=45] \cGate{0}{0} \end{scope}
\begin{scope}[shift={(\sroot/2,\sroot/2)},rotate=45] \oGate{0}{0} \end{scope}
\begin{scope}[shift={(1.5*\sroot,\sroot/2)},rotate=45] \acGate{0}{0} \end{scope}
\begin{scope}[shift={(2.5*\sroot,\sroot/2)},rotate=45] \oGate{0}{0} \end{scope}
\begin{scope}[shift={(0,\sroot)},rotate=45] \oGate{0}{0} \end{scope}
\begin{scope}[shift={(\sroot,\sroot)},rotate=45] \doGate{0}{0} \end{scope}
\begin{scope}[shift={(2*\sroot,\sroot)},rotate=45] \aGate{0}{0} \end{scope}
\begin{scope}[shift={(\sroot/2,1.5* \sroot)},rotate=45] \oGate{0}{0} \end{scope}
\begin{scope}[shift={(1.5*\sroot,1.5*\sroot)},rotate=45] \dtGate{0}{0} \end{scope}
\begin{scope}[shift={(2.5*\sroot,1.5*\sroot)},rotate=45] \dtGate{0}{0} \end{scope}
\end{tikzpicture}
,
\begin{tikzpicture}[baseline={([yshift=-.6ex]current bounding box.center)},  scale=0.65]
\begin{scope}[shift={(0,0)},rotate=45] \oGate{0}{0} \end{scope}
\begin{scope}[shift={(\sroot,0)},rotate=45] \doGate{0}{0} \end{scope}
\begin{scope}[shift={(2*\sroot,0)},rotate=45] \cGate{0}{0} \end{scope}
\begin{scope}[shift={(\sroot/2,\sroot/2)},rotate=45] \oGate{0}{0} \end{scope}
\begin{scope}[shift={(1.5*\sroot,\sroot/2)},rotate=45] \eeGate{0}{0} \end{scope}
\begin{scope}[shift={(2.5*\sroot,\sroot/2)},rotate=45] \oGate{0}{0} \end{scope}
\begin{scope}[shift={(0,\sroot)},rotate=45] \oGate{0}{0} \end{scope}
\begin{scope}[shift={(\sroot,\sroot)},rotate=45] \doGate{0}{0} \end{scope}
\begin{scope}[shift={(2*\sroot,\sroot)},rotate=45] \aGate{0}{0} \end{scope}
\begin{scope}[shift={(\sroot/2,1.5* \sroot)},rotate=45] \oGate{0}{0} \end{scope}
\begin{scope}[shift={(1.5*\sroot,1.5*\sroot)},rotate=45] \dtGate{0}{0} \end{scope}
\begin{scope}[shift={(2.5*\sroot,1.5*\sroot)},rotate=45] \dtGate{0}{0} \end{scope}
\end{tikzpicture}
\right\} \,.
\end{align}
For this map we can prove the following 
\begin{lemma}
$\mathcal{F}$ is surjective, namely $\mathcal{F}(x) \cap \mathcal{F}(y) = \{ \}$, for $x \neq y$.
\end{lemma}
\begin{proof}
Since $x \neq y$, the configurations differ in at least one tile. Moreover, because a tile $g$ is the only tile with two incoming lines, they need to differ in a tile different from $g$. $\mathcal{F}$ does not change tiles different from $g$, therefore the resulting configurations from different configurations differ in the initially different tile. The configurations from the same set are different by construction.
\end{proof}

Next, notice that we can express the sum of the weights of tilde configurations with weights of configurations without the tilde
\be
\sum_{k=1}^{2^n} w(\tilde{C}^4_{i,k}) 
 = 
\left(\frac{\varepsilon_1 \varepsilon_2+ a c}{g}\right)^n w(C_i^4)\,,
\ee
where we used the factorisation of the weights. To see that, notice that we can factor out from the sum the weights of all tiles different from $g$ in $C_i$ and tiles at the same positions in tilde configurations. Then we are left only with a product of $g$ tiles, with each of them maped to $(\varepsilon_1 \varepsilon_2+ a c)$.

For instance, a configuration with of two $g$ tiles
\be
w({C_i^4} ) =  g g 
\ee
maps to
\be
\sum_{k=1}^4 w(\tilde{C}_{i,k}^4 ) = (\varepsilon_1 \varepsilon_2 + a c)(\varepsilon_1 \varepsilon_2 + a c)\, .
\ee
Since each tiling $\tilde{C}^4_{i,k}$ is different and corresponds to a different term in the diagrams of  
$\left(\braket{\mcircf_{x_1} | \!\mcircf_{y_1}\!(t) }_L\!\braket{\mcircf_{x_2} | \!\mcircf_{y_2}\!(t) }_L\!+\!\braket{\mcircf_{x_1}| \!\mcircf_{y_2}\!(t) }_L\!\braket{\mcircf_{x_2} | \!\mcircf_{y_1}\!(t) }_L\right)$, it follows that
\be
\sum_i \sum_{k=1}^{2^n} w(\tilde{C}^4_{i,k}) \leq \left(\braket{\mcircf_{x_1} | \!\mcircf_{y_1}\!(t) }_L\!\braket{\mcircf_{x_2} | \!\mcircf_{y_2}\!(t) }_L\!+\!\braket{\mcircf_{x_1}| \!\mcircf_{y_2}\!(t) }_L\!\braket{\mcircf_{x_2} | \!\mcircf_{y_1}\!(t) }_L\right)\, . 
\ee
Finally, using that the maximal possible number of $g$ tiles is $t$ (for $L>1$), we see
\be
\!\!\!\!\!\braket{\mcircf_{x_1}\mcircf_{x_2} | \!\mcircf_{y_1}\mcircf_{y_2}\!(t) }_L \leq \left[{\rm max}\left(1, \frac{g^t}{(\varepsilon_1 \varepsilon_2+ a c)^t}\right)\right]\!\! \left(\braket{\mcircf_{x_1} | \!\mcircf_{y_1}\!(t) }_L\!\braket{\mcircf_{x_2} | \!\mcircf_{y_2}\!(t) }_L\!+\!\braket{\mcircf_{x_1}| \!\mcircf_{y_2}\!(t) }_L\!\braket{\mcircf_{x_2} | \!\mcircf_{y_1}\!(t) }_L\right)\,. 
\ee
Similarly, for $2<m<2L$ we use that the maximal possible number of $g$ tiles is $(m-1)t$ to find \eqref{eq:Cmbound}.

\section{Asymptotics of infinite-volume correlation functions}
\label{sec:AsyCorr}
\subsubsection{Asymptotics between integer indexed points}
Let us start considering the correlation functions between integer indexed points. First we note that they can be expressed in terms of the  ordinary hypergeometric function and the Jacobi polynomials.
\begin{align}
\braket{\mcircf_x | \mcircf_0(t) }_{\infty} &= \displaystyle  \sum_{n=1}^{t}
\left(\frac{\varepsilon_1 \varepsilon_2}{a c}\right)^n \binom{t+x}{n} \binom{ t-x-1}{n-1} a^{t+x} c^{t  - x} + \delta_{t,x} a^{2t}\,  
\nonumber\\
&=a^{t+x} c^{t  - x} \frac{\varepsilon_1  \varepsilon_2 (t+x) \,
}{a c} \phantom{}_2F_1\left(1-t-x,1-t+x;2;\frac{\varepsilon_1  \varepsilon_2}{a c}\right) + \delta_{t,x} a^{2t} \, ,
\label{eq:corr1}
\end{align}
where $_2F_1$ is Gaussian or ordinary hypergeometric function. The expression can be expressed using Jacobi polynomials $P^{(\alpha,\beta)}_n$ as
\be
\braket{\mcircf_x | \mcircf_0(t) }_{\infty} = a^{t+x} c^{t  - x} \;z \; P^{(1,-2t)}_{t+x-1}(1-2z)+ \delta_{t,x} a^{2t} \, ,
\ee
where we introduced $z={\varepsilon_1 \varepsilon_2}/{a c}$. 

We now proceed to develop the asymptotic expansion of \eqref{eq:corr1}. We begin by writing it in terms of the ray variable $\zeta = x/t \in [ -1,1 ]$
\be 
\braket{\mcircf_{\zeta t} | \mcircf_0(t) }_{\infty} = \delta_{\zeta-1} a^{2t} + \sum_{n=1}^{{\rm min}(t(1+\zeta),t(1-\zeta))} P_n(\zeta, t),
\ee
with 
\be
P_n(\zeta,t)= z^n \binom{t(1+\zeta)}{n} \binom{ t(1-\zeta)-1}{n-1} a^{t(1+\zeta)} c^{t(1-\zeta)}\,,
\ee 
and $z={\varepsilon_1 \varepsilon_2}/{a c}$. We obtain the asymptotic form by firstly using the Stirling's approximation for $P_n$. Then we write the Taylor expansion for $\log P_n$ in $n$ around the maximum, substitute the sum over $n$ by the integral and integrate over $n$. Finally, we write the Taylor expansion for $\log \braket{\mcircf_{\zeta t} | \mcircf_0(t) }_{\infty}$ in terms of $\zeta$ to second order around the maximum. Therefore, we obtain Gaussian form of the correlations, which accurately describes the asymptotic behaviour. 

Stirling's approximation for $P_n$ yields
\be
P_n(\zeta, t) \approx z^n a^{t(1+\zeta)} c^{t(1-\zeta)} \sqrt{\frac{t^2 (1-\zeta^2)- t(1+\zeta)}{(2 \pi)^2 (n^2-n) ((t-n)^2-(t \zeta)^2)}} 
\frac{(t(1+\zeta))^{t(1+\zeta)} (t(1-\zeta)-1)^{t(1-\zeta)-1}
}{n^n (n-1)^{n-1} (t(1+\zeta)-n)^{t(1+\zeta)-n} (t(1-\zeta)-n)^{t(1-\zeta)-n} }\, .
\ee
Next we find the maximum $\bar{n}$, by demanding that the derivative of $\log P_n(\zeta, t)$ vanishes. We obtain the equation
\be
\log \frac{z(t(1+\zeta)-\bar{n}) (t(1-\zeta)-\bar{n})}{\bar{n}(\bar{n}-1)} =- \frac{1}{2} \left(\frac{1}{\zeta t-\bar{n}+t}+\frac{1}{-\zeta t-\bar{n}+t}+\frac{1-2 \bar{n}}{(\bar{n}-1) \bar{n}}\right),
\label{eq:barn}
\ee
which we solve up to $1/t$ order. We obtain
\be
\bar{n} = v t + n_0 + \mathcal{O}(1/t), \qquad
 v = \frac{z - \sqrt{z+\zeta^2(z^2-z)} }{z-1} ,
\ee
where we took the relevant root of the quadratic equation, for which $0 \geq v \geq 1$.
The expression for $n_0$ contains the first sub-leading terms from LHS of Eq. (\ref{eq:barn}) and the leading order from RHS.
\be
n_0 = \frac{1}{2}\frac{v^2-v}{v+\zeta^2-1} = \frac{1}{2} \frac{z+\zeta^2(z^2-z)-\sqrt{z+\zeta^2(z^2-z)}}{z-1+\zeta^2(z-1)^2} \, .
\ee
Next, we derive $S_{\bar{n}} =- \frac{\partial^2 \log P_n}{\partial^2} |_{\bar{n}}$
\be
 S_{\bar{n}} =  \frac{2 z}{v^3 t} (1-v-\zeta^2) \, .
\ee
The simplified value at the saddle point $ P_{\bar{n}}(\zeta, t)$ reads
\be  
P_{\bar{n}}(\zeta, t) = \frac{\sqrt{z} }{2 \pi v t} \sqrt{\frac{1+\zeta}{1-\zeta}}
 \Big[ a c z \frac{1-\zeta^2}{v^2} \; 
 \left(\frac{a}{c}\, \frac{(1+\zeta)(1-\zeta-v)}{(1-\zeta)(1+\zeta-v)} \right)^\zeta  \Big]^t \, .
\ee
Next, we approximate the sum with the integral and obtain
\be 
\braket{\mcircf_{\zeta t} | \mcircf_0(t) }_{\infty} = \delta_{\zeta-1} a^{2t} +\int_0^t \; dn \; P_{\bar{n}}(\zeta, t) e^{-S_{\bar{n}} (n-\bar{n})^2/2}
= \delta_{\zeta-1} a^{2t} + \sqrt{2 \pi / S_{\bar{n}}} P_{\bar{n}}(\zeta, t) \, .
\ee
The asymptotic form is:
\begin{align}
\braket{\mcircf_{\zeta t} | \mcircf_0(t) }_{\infty} &= \delta_{\zeta-1} a^{2t} 
+
 \frac{ 1 }{ 2}  \sqrt{\frac{v }{\pi t} \frac{ (1+\zeta) }{   (1-\zeta) (1-v-\zeta^2))} }
 \Big[ a c z \frac{1-\zeta^2}{v^2} \; 
 \left(\frac{a}{c}\, \frac{(1+\zeta)(1-\zeta-v)}{(1-\zeta)(1+\zeta-v)} \right)^\zeta  \Big]^t 
  \, .
\end{align} 

The expression is rather involved, but it is well approximated by a Gaussian function of $\zeta$. To find it we first determine the maximum versus  $\zeta$. Taking the first derivative equal to $0$ leads to
\be
\bar{\zeta} =\zeta_0 + \frac{1}{t} \zeta_1 + \mathcal{O}(1/t^2)\, , \quad \zeta_0 =  \frac{a-c}{\sqrt{(a-c)^2 + 4 a c z}}=  \frac{a-c}{\Delta}\, .
\ee
We introduced $\Delta=\sqrt{4 a c z+(a-c)^2}$. To get the correct normalization of the final result, we need to obtain the $\mathcal{O}(1/t)$ term. It reads
\be
\zeta_1 = 
\frac{\Delta \left(a^2 c (20 z-11)+3 a^3+3 a c^2 (4 z-1)+3 c^3\right)+2 a c (z-1) (a-c) (a+c)}{2 \left(4 a c z+(a-c)^2\right) \left(3 a^2+2 a c (8 z-5)+3 c^2\right)} \,.
\ee
The second derivative at $\bar{\zeta}$ up to first order
\be
S_{\bar{\zeta}}= \frac{((a-c)^2+4acz)^{3/2}}{2acz |a+c|} =  \frac{\Delta^{3}}{2acz |a+c|}\, .
\ee
We obtain the final form
\be
\braket{\mcircf_{\zeta t} | \mcircf_0(t) }_{\infty} = \delta_{\zeta-1} a^{2t} + \frac{C_{0,0}}{\sqrt{t}} \lambda^t e^{-t S_{\bar{\zeta}} (\zeta- \bar{\zeta})^2/2} \, ,
\label{eq:asyFin}
\ee
with
\be
\lambda=\frac{1}{4}
\left( (a+c)+\sqrt{4 a c z+(a-c)^2}\right)^2\, .
\ee
and
\be
C_{0,0}
=
  \frac{ 1 }{ 2}  \sqrt{\frac{v_0 }{\pi } \frac{ (1+\zeta_0) }{   (1-\zeta_0) (1-v_0-\zeta_0^2)} } \times 
   \left(\frac{a \left(\zeta _0 \left(\zeta _0+2 \sqrt{\zeta _0^2 (z-1) z+z}-2 \zeta _0 z\right)-1\right)}{c \left(\zeta _0^2-1\right)}\right)^{\zeta _1}\, ,
   \label{eq:asyNorm}
 \ee
 with $v_0 = v(\zeta_0)$. A comparison with the exact evaluation is shown in Fig. \ref{fig:Asympt}.

\begin{figure*}
 \centering
    \includegraphics[scale=0.55]{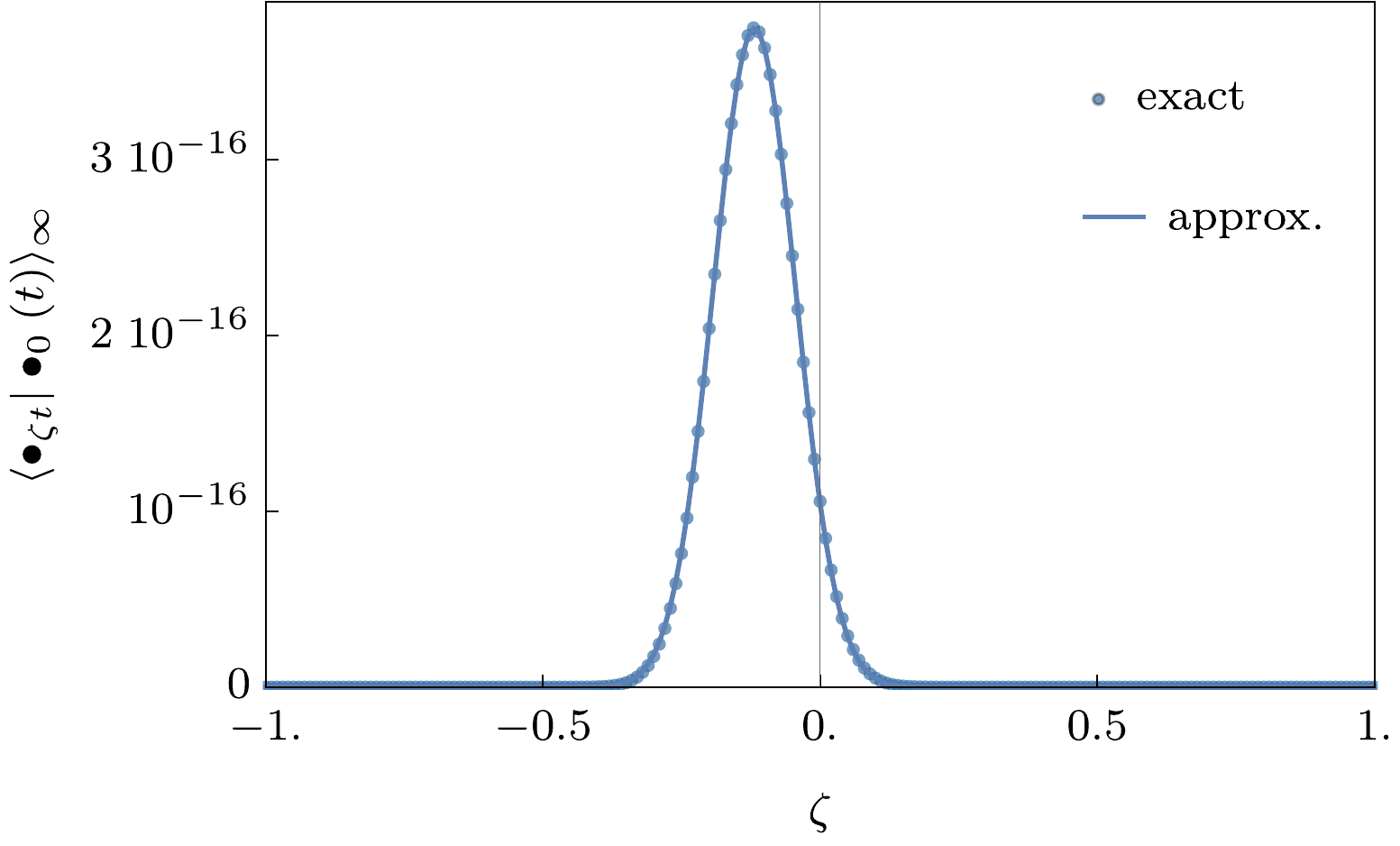}
    \includegraphics[scale=0.55]{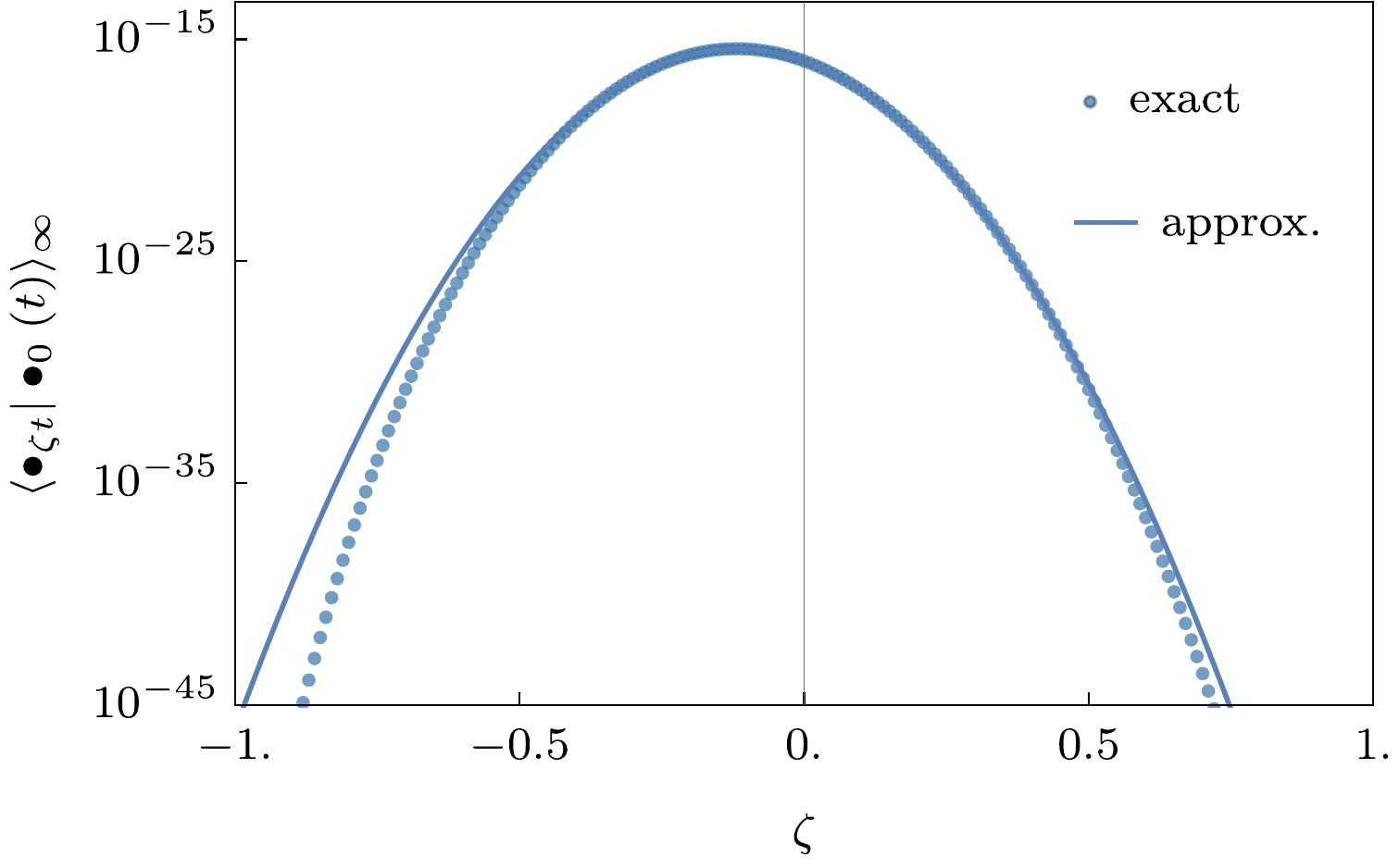}
    \caption{Comparison between asymptotic formula (\ref{eq:asyFin}) and exact evaluation (\ref{eq:corr1}). $t=100$, $a=0.4$, $c=0.5$ and $z=0.8$.}
    \label{fig:Asympt}
\end{figure*}

\subsubsection{Asymptotics between integer and half-integer indexed points}
\label{sec:AsyCorrCeo}
The calculation of the asymptotics for correlations with different endpoints follows the same lines. The final results read as 
\begin{align}
\braket{\mcircf_{x+1/2} | \mcircf_{1/2}(t) }_{\infty} &=\delta_{\zeta+1} c^{2t} +\frac{C_{1,1}}{\sqrt{t}} \lambda^t e^{-t S_{\bar{\zeta}} (\zeta- \bar{\zeta})^2/2}\, ,\\
\braket{\mcircf_{x+1/2} | \mcircf_{0}(t) }_{\infty} &=\frac{C_{0,1}}{\sqrt{t}} \lambda^t e^{-t S_{\bar{\zeta}} (\zeta- \bar{\zeta})^2/2}\, ,\\
\braket{\mcircf_{x} | \mcircf_{1/2}(t) }_{\infty} &=\frac{C_{1,0}}{\sqrt{t}} \lambda^t e^{-t S_{\bar{\zeta}} (\zeta- \bar{\zeta})^2/2}\, ,
\label{eq:asyFinCeo}
\end{align}
where $x\in \mathbb Z_L$, $C_{1,1}= \sqrt{\frac{S_{\bar\zeta}}{2\pi}}-C_{0,0}$ , and we do not report an explicit form of $C_{1,0}$ and $C_{0,1}$ as it is not needed in the main text.  

\section{Effective Markov matrix for the $m=1$ sector}
\label{app:Markov}
We begin by noting that $K^{(1)}_g(t)$ can be written as the trace of the following $2L\times 2L$ block-circulant Markov operator
\be
\mathbb M = \begin{pmatrix}
A & B & 0 & 0 & \cdots & 0 & C\\
C & A & B & 0 & \cdots & 0 & 0\\
0 & C & A & B & \cdots & 0 & 0\\
\vdots &  &  &  & \ddots &  & \vdots\\
B & 0 & 0 & 0 & \cdots & C & A\\
\end{pmatrix}\,,
\ee
where we introduced the following $2\times2$ matrices
\begin{align}
A = &
\begin{pmatrix}
\varepsilon_1\varepsilon_2 & c\varepsilon_2 \\
a\varepsilon_1 & \varepsilon_1\varepsilon_2\\
\end{pmatrix}\,,
 &  B = &
\begin{pmatrix}
c^2 & 0 \\
c\varepsilon_1 & 0\\
\end{pmatrix}\,,
& C = &
\begin{pmatrix}
0 & a\varepsilon_2 \\
0 & a^2\\
\end{pmatrix}\,.
\end{align}
The eigenvectors of $\mathbb M$ can then be written as 
\be
w_m^\pm =\begin{pmatrix}
 v_m^\pm\\
 e^{i \frac{2\pi m}{L}} v_m^\pm\\
\vdots\\
 e^{i \frac{2\pi m}{L}(L-1)}v_m^\pm
\end{pmatrix}\,,\qquad m= 0,\ldots, L-1,
\ee 
where $v_m^\pm$ are the two eigenvectors of 
\be
A+e^{i \frac{2\pi m}{L}}B+e^{-i \frac{2\pi m}{L}}C\,.
\ee
The eigenvalues are given by 
\be
\lambda^\pm_m= 
\frac{1}{4}  \left((a e^{-\frac{i \pi  m}{L}} + c e^{\frac{i \pi  m}{L}})  \pm  \sqrt{(a e^{-\frac{i \pi  m}{L}}-c e^{\frac{i \pi  m}{L}})^2+4 \varepsilon_1\varepsilon_2}\right)^2\,.
\label{eq:allLambdas}
\ee
In particular, the eigenvalue with largest absolute magnitude is one of the four choices:
\begin{align}
\lambda^{\pm}_0&=\frac{1}{4}  \left(a+c  \pm \sqrt{(a-c)^2+4\varepsilon_1\varepsilon_2}\right)^2\,, \quad
\lambda^{\pm}_{L/2}= -\frac{1}{4}  \left(a-c  \pm \sqrt{(a+c)^2-4\varepsilon_1\varepsilon_2}\right)^2\,.
\label{eq:largestLambda}
\end{align}
$\lambda^{+}_0$ is the relevant solution when all parameters are positive, and coincides with $\lambda$ in Eq.~\eqref{eq:lambda} of the main text. In this case we immediately get
\be
K^{(1)}_g(t) = {\rm tr}\left[\mathbb M^t\right] \simeq (\lambda^+_0)^t\,.
\ee

\section{Bound on $K^{(m)}_g(t)$}
\label{app:boundK}

In this appendix we use Property~\ref{prop:P1} and Eq.~\eqref{eq:C1asy} to find a bound for 
\be
K^{(2+)}_g(t) \equiv  \sum_{m=2}^{2L} K_{g}^{(m)}(t) \,.
\ee
Summing the 
\be
\binom{2L}{m} \leq \frac{(2L)^m}{m!}
\ee
$m$-point correlation functions in $K_g^{(m)}$ for $1<m<2L$ we get
\be
K_g^{(m)} < G (2L)^m \lambda^{mt} \left({\rm max}\left (1, \frac{g^{(m-1)t}}{(\varepsilon_1 \varepsilon_2+ a c)^{(m-1)t}}\right) \right) \, ,
\label{eq:Kmbound}
\ee
where $G$ is a constant. For the sector with $m=2L$ we instead find exactly $K^{(2L)}_g(t) = g^{4L t}$. Combining the bounds for all terms (\ref{eq:Kmbound}), we obtain
\be
K^{(2+)}_g(t) < G \gamma^2\frac{1}{1-\gamma} \left({\rm max}\left(1, \frac{g}{\varepsilon_1 \varepsilon_2+ a c}\right) \right)^{-t} \, , 
\label{eq:K2pBound}
\ee
with
\be
\gamma= 2L \lambda^t \left({\rm max}\left(1, \frac{g^t}{(\varepsilon_1 \varepsilon_2+ a c)^t}\right) \right)\, .
\ee 
This is much smaller than $K_1(t)$ for $\lambda\, {\rm max}[1, g/(\varepsilon_1 \varepsilon_2+ a c)]<1$. Empirically, this condition seems to hold in $83\% $ of randomly generated gates with no splittings  and positive weights. We compare the bound with a numerical evaluation of $K^{(2+)}_g(t)$ in Fig \ref{fig:Bound}.

\begin{figure}
    \centering
    \includegraphics[scale=1.]{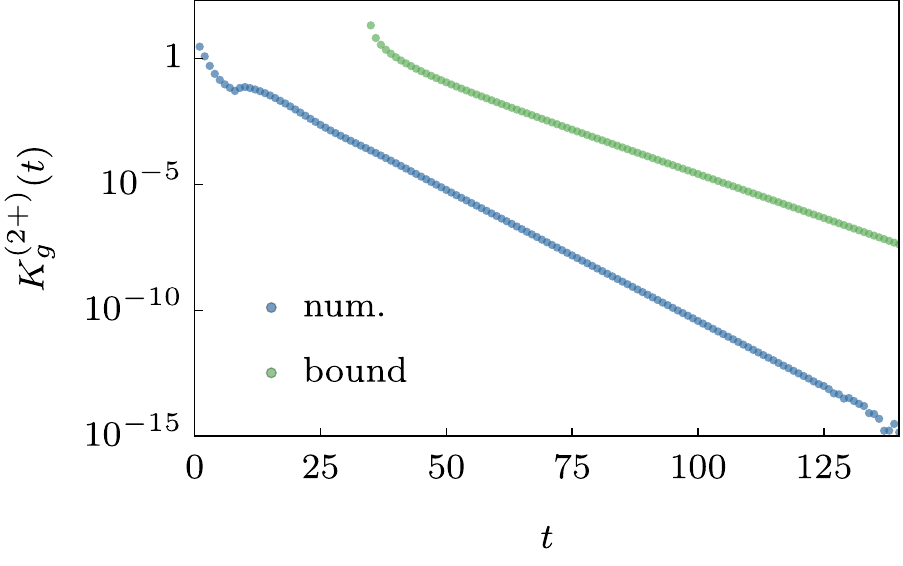}
    \caption{
    Exact numerical evaluation of $K^{(2+)}_g(t)$ 
     vs bound \ref{eq:K2pBound} (for $G=1$). We see that $K^{(2+)}_g(t)$ the bound is not tight.
Here we took $L=8$, and the gate's parameters are given in the second row of Tab.~\ref{tab:U1gates}.   }
    \label{fig:Bound}
\end{figure}

}

\subsection{Bound on $K^{(m)}_g(t)$ in the case of negative weights}

For gates with negative weights we have $K^{(1)}_g \simeq \lambda_1^t$, where $\lambda_1$ is the largest in the magnitude of the eigenvalues in \eqref{eq:largestLambda}. On the other hand, the two-point correlation function at long times are bounded from above by the asymptotic form for the effective gate, where the parameters of the gate are substituted by their absolute values. In particular, we have
\be
\lambda_{\rm eff} = \frac{1}{4}  \left(|a|+|c|  + \sqrt{(|a|-|c|)^2+4 |\varepsilon_1| |\varepsilon_2|}\right)^2.
\ee
Thus, we get
\be
K^{(2+)}_g(t) < G \gamma^2\frac{1}{1-\gamma} \left({\rm max}\left(1, \frac{|g|}{|\varepsilon_1| |\varepsilon_2|+ |a| |c|}\right) \right)^{-1} \, , 
\label{eq:K2pBound}
\ee
with
\be
\gamma= 2L \lambda_{\rm eff}^t \left({\rm max}\left(1,  \frac{|g|^t}{(|\varepsilon_1| |\varepsilon_2|+ |a| |c|)^t}\right) \right)\, .
\ee 
We see that $K^{(2+)}_g(t)$ is certainly exponentially smaller than $K^{(1)}_g(t)$ for
\be
\frac{\lambda_{\rm eff}^2}{|\lambda_1|} {\rm max}\left(1,  \frac{|g|}{|\varepsilon_1| |\varepsilon_2|+ |a| |c|}\right) < 1 \, ,
\ee
which holds in 90\% of randomly generated gates with no splittings and at least one of $a,c, \varepsilon_1, \varepsilon_2,g$ negative.

\section{$K_g(t)$ for $f=e=\varepsilon_{1}=0$}
\label{app:SWAPlike}

Here it is more convenient to look at the following representation ($L=4$, $t=2$)
\be
K_{g}(t) = 
\begin{tikzpicture}[baseline=(current  bounding  box.center), scale=0.7]
\def\eps{0};
\def\shift{11}
\def\shifty{-2.5}
\foreach \i in {1,...,4}{
\draw[thin, opacity =0.2,dashed] (2*\i+2-1.5+0.25,-2.5-0.1) -- (2*\i+2-1.5+0.25,1.5-0.1);
\draw[thin, opacity =0.2,dashed] (2*\i+2-0.5-0.25,-2.5-0.1) -- (2*\i+2-0.5-0.25,1.5-0.1);
}
\foreach \i in{3,...,3}{
\draw[thin] (\shift-1.5,2*\i-0.5-3.5+\shifty) arc (45:-90:0.15);
\draw[thin] (\shift-10+0.5+0,2*\i-0.5-3.5+\shifty) arc (-135:-270:0.15);
}
\foreach \i in{2,...,3}{
\draw[thin] (\shift-1.5,1+2*\i-0.5-3.5+\shifty) arc (-45:90:0.15);
\draw[thin] (\shift-10+0.5,1+2*\i-0.5-3.5+\shifty) arc (135:270:0.15);
}
\foreach \i in {1,...,4}
{
\draw[thin] (\shift-.5-2*\i,1+\shifty) -- (\shift+0.5-2*\i,0+\shifty);
\draw[thin] (\shift-0.5-2*\i,0+\shifty) -- (\shift+0.5-2*\i,1+\shifty);
\draw[ thick, fill=myY, rounded corners=2pt] (\shift-0.25-2*\i,0.25+\shifty) rectangle (\shift+.25-2*\i,0.75+\shifty);
\draw[thick] (\shift-2*\i,0.65+\shifty) -- (\shift+.15-2*\i,.65+\shifty) -- (\shift+.15-2*\i,0.5+\shifty);
\draw[thin] (2*\i+2-1.5,1.5) arc (135:-0:0.15);
\draw[thin] (2*\i+2-2.5,1.5) arc (-325:-180:0.15);
\draw[thin] (2*\i+2-1.5,-2.5) arc (-45:180:-0.15);
\draw[thin] (2*\i+2-0.5,-2.5) arc (45:-180:0.15);
}
\foreach \jj[evaluate=\jj as \j using -2*(ceil(\jj/2)-\jj/2)] in {0,...,2}
\foreach \i in {2,...,5}
{
\draw[thin] (\shift+.5-2*\i-1*\j,2+1*\jj+\shifty) -- (\shift+1-2*\i-1*\j,1.5+\jj+\shifty);
\draw[thin] (\shift+1-2*\i-1*\j,1.5+1*\jj+\shifty) -- (\shift+1.5-2*\i-1*\j,2+\jj+\shifty);
}
\foreach \i in {1,...,4}
{
\draw[thin] (\shift-.5-2*\i,1+\shifty) -- (\shift+0.5-2*\i,0+\shifty);
\draw[thin] (\shift-0.5-2*\i,0+\shifty) -- (\shift+0.5-2*\i,1+\shifty);
\draw[thick, fill=myY, rounded corners=2pt] (\shift-0.25-2*\i,0.25+\shifty) rectangle (\shift+.25-2*\i,0.75+\shifty);
\draw[thick] (\shift-2*\i,0.65+\shifty) -- (\shift+.15-2*\i,.65+\shifty) -- (\shift+.15-2*\i,0.5+\shifty);
}
\foreach \jj[evaluate=\jj as \j using -2*(ceil(\jj/2)-\jj/2)] in {0,...,2}
\foreach \i in {2,...,5}
{
\draw[thin] (\shift+.5-2*\i-1*\j,1+1*\jj+\shifty) -- (\shift+1-2*\i-1*\j,1.5+\jj+\shifty);
\draw[thin] (\shift+1-2*\i-1*\j,1.5+1*\jj+\shifty) -- (\shift+1.5-2*\i-1*\j,1+\jj+\shifty);
\draw[ thick, fill=myY, rounded corners=2pt] (\shift+0.75-2*\i-1*\j,1.75+\jj+\shifty) rectangle (\shift+1.25-2*\i-1*\j,1.25+\jj+\shifty);
\draw[thick] (\shift+1-2*\i-1*\j,1.65+1*\jj+\shifty) -- (\shift+1.15-2*\i-1*\j,1.65+1*\jj+\shifty) -- (\shift+1.15-2*\i-1*\j,1.5+1*\jj+\shifty);
}
\end{tikzpicture}\, .
\label{eq:SFFredtr}
\ee
Since we do not allow splittings ($e=f=0$) the merges will not appear in the evaluation of~\eqref{eq:SFFredtr}. Otherwise there would be a different number of operators at the bottom and the top. Furthermore, since $\varepsilon_{1}=0$, we cannot convert left mover to right mover. Since the number of right movers at the bottom and the top is the same, also $\varepsilon_{2}=0$. Therefore, the only allowed tiles are (see discussion around \eqref{eq:introducingTiles} for the introduction of tiles)
\begin{align}
\begin{tikzpicture}[baseline=(current  bounding  box.center), scale=0.7]
\begin{scope}[shift={(0,1.5)},rotate=45]\oGate{0}{0} \end{scope}
\begin{scope}[shift={(3,1.5)},rotate=45]\aGate{0}{0} \end{scope}
\begin{scope}[shift={(6,1.5)},rotate=45]\cGate{0}{0} \end{scope}
\begin{scope}[shift={(9,1.5)},rotate=45]\gGate{0}{0} \end{scope}
\Text[x=1.25,y=1.5, anchor = center]{$=1$}
\Text[x={3+1.25},y=1.5, anchor = center]{$=a$}
\Text[x={6+1.25},y=1.5, anchor = center]{$=c$}
\Text[x={1.25+9},y=1.5, anchor = center]{$=g$}
\end{tikzpicture} \, .
\label{eq:tiles}
\end{align}
We start by considering $\mcircf$ on a site and follow the evolution of a state on a single wire until the line closes on itself, which we call an \emph{orbit}. For example
\be
\begin{tikzpicture}[baseline=(current  bounding  box.center), scale=0.7]
\def\eps{0};
\def\shift{11}
\def\shifty{-2.5}
\foreach \i in {1,...,4}{
\draw[thin, opacity =0.2,dashed] (2*\i+2-1.5+0.25,-2.5-0.1) -- (2*\i+2-1.5+0.25,1.5-0.1);
\draw[thin, opacity =0.2,dashed] (2*\i+2-0.5-0.25,-2.5-0.1) -- (2*\i+2-0.5-0.25,1.5-0.1);
}
\foreach \i in{3,...,3}{
\draw[thin] (\shift-1.5,2*\i-0.5-3.5+\shifty) arc (45:-90:0.15);
\draw[thin] (\shift-10+0.5+0,2*\i-0.5-3.5+\shifty) arc (-135:-270:0.15);
}
\foreach \i in{2,...,3}{
\draw[thin] (\shift-1.5,1+2*\i-0.5-3.5+\shifty) arc (-45:90:0.15);
\draw[thin] (\shift-10+0.5,1+2*\i-0.5-3.5+\shifty) arc (135:270:0.15);
}
\foreach \i in {1,...,4}
{
\draw[thin] (\shift-.5-2*\i,1+\shifty) -- (\shift+0.5-2*\i,0+\shifty);
\draw[thin] (\shift-0.5-2*\i,0+\shifty) -- (\shift+0.5-2*\i,1+\shifty);
\draw[ thick, fill=myY, rounded corners=2pt] (\shift-0.25-2*\i,0.25+\shifty) rectangle (\shift+.25-2*\i,0.75+\shifty);
\draw[thick] (\shift-2*\i,0.65+\shifty) -- (\shift+.15-2*\i,.65+\shifty) -- (\shift+.15-2*\i,0.5+\shifty);

\draw[thin] (2*\i+2-1.5,1.5) arc (135:-0:0.15);
\draw[thin] (2*\i+2-2.5,1.5) arc (-325:-180:0.15);
\draw[thin] (2*\i+2-1.5,-2.5) arc (-45:180:-0.15);
\draw[thin] (2*\i+2-0.5,-2.5) arc (45:-180:0.15);
}
\foreach \jj[evaluate=\jj as \j using -2*(ceil(\jj/2)-\jj/2)] in {0,...,2}
\foreach \i in {2,...,5}
{
\draw[thin] (\shift+.5-2*\i-1*\j,2+1*\jj+\shifty) -- (\shift+1-2*\i-1*\j,1.5+\jj+\shifty);
\draw[thin] (\shift+1-2*\i-1*\j,1.5+1*\jj+\shifty) -- (\shift+1.5-2*\i-1*\j,2+\jj+\shifty);
}
\foreach \i in {1,...,4}
{
\draw[thin] (\shift-.5-2*\i,1+\shifty) -- (\shift+0.5-2*\i,0+\shifty);
\draw[thin] (\shift-0.5-2*\i,0+\shifty) -- (\shift+0.5-2*\i,1+\shifty);
\draw[thick, fill=myY, rounded corners=2pt] (\shift-0.25-2*\i,0.25+\shifty) rectangle (\shift+.25-2*\i,0.75+\shifty);
\draw[thick] (\shift-2*\i,0.65+\shifty) -- (\shift+.15-2*\i,.65+\shifty) -- (\shift+.15-2*\i,0.5+\shifty);
}
\foreach \jj[evaluate=\jj as \j using -2*(ceil(\jj/2)-\jj/2)] in {0,...,2}
\foreach \i in {2,...,5}
{
\draw[thin] (\shift+.5-2*\i-1*\j,1+1*\jj+\shifty) -- (\shift+1-2*\i-1*\j,1.5+\jj+\shifty);
\draw[thin] (\shift+1-2*\i-1*\j,1.5+1*\jj+\shifty) -- (\shift+1.5-2*\i-1*\j,1+\jj+\shifty);
\draw[ thick, fill=myY, rounded corners=2pt] (\shift+0.75-2*\i-1*\j,1.75+\jj+\shifty) rectangle (\shift+1.25-2*\i-1*\j,1.25+\jj+\shifty);
\draw[thick] (\shift+1-2*\i-1*\j,1.65+1*\jj+\shifty) -- (\shift+1.15-2*\i-1*\j,1.65+1*\jj+\shifty) -- (\shift+1.15-2*\i-1*\j,1.5+1*\jj+\shifty);
}
\def\colr{myblue}
\draw[thick,\colr] (2.5,-2.5) -- (6.5 ,1.5);
\draw[thick,\colr] (6.5,-2.5) -- (9.5 ,0.5);
\draw[thick,\colr] (1.5,0.5) -- (2.5 ,1.5);
\draw[thick,\colr] (2.76,-2.5 -.11) -- (2.76 ,1.5-0.1);
\draw[thick,\colr] (6.76,-2.5 -.11) -- (6.76 ,1.5-0.1);
\draw[thick,\colr] (1.5+0.06,0.25-0.01) -- (9.5-.05 ,0.75+0.01);
\draw[thick,\colr] (2.5,-2.5 -0.0) arc (-45:180:-0.15);
\draw[thick,\colr] (2.5,1.5 -0.0) arc (135:-0:0.15);
\draw[thick,\colr] (6.5,-2.5 -0.0) arc (-45:180:-0.15);
\draw[thick,\colr] (6.5,1.5 -0.0) arc (135:-0:0.15);
\draw[thick,\colr] (1.5,0.5 -0.0) arc (135:270:0.15);
\draw[thick, \colr] (9.5,0.5 -0.0) arc (-45:90:0.15);
\end{tikzpicture}\, ,
\label{eq:SFForbit}
\ee
where we used periodic boundary conditions in space and time (from trace). The state on a given wire travels straight along the light-cone. When it reaches the top, the trace acts as a periodic boundary condition in time. Therefore it continues to travel until it reaches the starting point.  When it reaches the starting point, it has traveled through $2\ell$ gates, with $\ell= {\rm lcm}(t,L)$ (${\rm lcm}$ stands for least common multiplier). 
There are $o=t L/\ell $ of distinct left(right)-moving orbits. Each orbit can either be occupied by the operator $\mcirc$ or $\mcircf$, which has weights $1$ or $c^{2\ell}$ ($a^{2\ell}$ for right moving orbits), respectively.

Consider first the case $g=0$. Here, we can have only left or right moving orbits filled with $\mcircf$, or else they cross and contribute $0$.
Therefore
\begin{align}
K_{g}(t) &= 1+ \sum_{n=1}^o \binom{o}{n} (a^{2n \ell} + c^{2n \ell}) = (1+a^{2 \ell})^o +(1+c^{2 \ell})^o-1\, ,
\end{align}
where $n$ counts the number of orbits occupied by $\mcircf$.

In case $g\neq0$, left and right moving orbits are allowed to cross, and we need to count the number of crossings and the combinatorial factors. The orbits cross $(\frac{n}{o} \,  \frac{m}{o}) 2 t L$ times, where $n$ and $m$ are the number of right and left moving orbits filled with $\mcircf$, respectively. Putting it all together, we find
\be
 K_{g}(t) =\!\!\!\!\! \sum_{n,m = 0}^{t L/\ell } \binom{t L/\ell}{n}\binom{t L/\ell }{m} a^{2 n \ell}  c^{2 m \ell} \left[\frac{g}{a c}\right]^{{2 \ell^2  n m}/{t L}} \, .
\label{eq:SWAP_like}
\ee

\section{Numerical Methods and Parameters of the Gates}
\label{app:numerics}

Direct numerical evaluations have been performed by exactly computing $K(t)$. We achieved that by exactly contracting the diagram (\ref{eq:SFFredtr}), using some basic functionalities of ITensor Library~\cite{ITensor}. 
The parameters of the gates used in numerical experiments are given in Table~\ref{tab:U1gates}.

\begin{table*}[h]
\centering
\begin {tabular} {c | ccccccc cc}
 & $a$ & $b$ & $c$ & $d$ & $e$ & $f$ & $g$ &$\varepsilon_1$ & $\varepsilon_2$\\
\hline
 \text {Gate red/yellow} & 0.761132 & 0 & 0.761132 & 0 & - 
     0.025732 & - 0.025732 & 0.701678 & 0.05 & 0.05 \\
\text {Gate green/blue} & .5 & 0 & 0.7 & 0 & \
0 & 0 & 0.25 &0.3 & 0.4 \\
\text {Gate splits case 1} & 0.040935   & 0.275939 & 0.040935 & 0.275939 & 0.275939 & 0.275939 & 0.037326 & 0.005513 & 0.005513 \\
\text {Gate splits case 2} & 0.262258 &   0.0745520 & 0.02958127 & -0.0277864 & 0.0448917 & -0.299859 & 0.361466 & 0.370204 & 0.146059
  \end {tabular}
  \caption{Parameters of the reduced gates, which were used in the figures. 
  }
  \label{tab:U1gates}
  \end{table*}

%


\begin{thebibliography}{99} 


\bibitem{arnold} V.~I.~Arnold and A.~Avez, \emph{Ergodic Problems of Classical Mechanics}, Addison-Wesley, Reprint Edition (1989).

\bibitem{arnold2} V.~I.~Arnold, \emph{Geometrical methods in the theory of ordinary differential equations}, Springer (1988).

\bibitem{ott} E.~Ott, \emph{Chaos in Dynamical Systems}, 2nd Edition, Cambridge University Press (2012).

\bibitem{cornfeld} 
I.~P.~Cornfeld,  S.~V.~Fomin, Y.~G.~Sinai, \emph{Ergodic Theory}, Springer (2012).
		
\bibitem{CGV80} G. Casati, F. Valz-Gris, and I. Guarneri, {\em On the connection between quantization of nonintegrable systems and statistical theory of spectra}, \href{https://doi.org/10.1007/BF02798790
}{Lett. Nuovo Cimento Soc. Ital. Fis. {\bf 28}, 279 (1980)}.
	
\bibitem{Berry81} M.~V.~Berry, {\em Quantizing a classically ergodic system: Sinai's billiard and the KKR method}, \href{https://doi.org/10.1016/0003-4916(81)90189-5}{Ann. Phys. (N.Y.) {\bf 131}, 163 (1981)}.
	
\bibitem{BGS84} O. Bohigas, M. J. Giannoni, and C. Schmit, {\em Characterization of chaotic quantum spectra and universality of level fluctuation laws}, \href{https://doi.org/10.1103/PhysRevLett.52.1}{Phys. Rev. Lett. {\bf 52}, 1 (1984)}.

\bibitem{HaPr07} 
P. Hayden and J. Preskill, \emph{Black holes as mirrors: quantum information in random subsystems},
\href{http://dx.doi.org/10.1088/1126-6708/2007/09/120}{J. High Energy Phys. {\bf 2007}, 120 (2007)}.

\bibitem{SeSu08} 
Y. Sekino and L. Susskind, \emph{Fast scramblers},
\href{http://dx.doi.org/10.1088/1126-6708/2008/10/065}{J. High Energy Phys. {\bf 2008}, 065 (2008)}.

\bibitem{KLP}  P.~Kos, M.~Ljubotina, and T.~Prosen, \emph{Many-body quantum chaos: Analytic connection to random matrix theory}, 
	\href{https://doi.org/10.1103/PhysRevX.8.021062}{Phys. Rev. X {\bf 8}, 021062} (2018).




\bibitem{Chalker} 
A.~Chan,~A.~De Luca,~J.~T.~Chalker,
\emph{Solution of a Minimal Model for Many-Body Quantum Chaos}, 
\href{https://doi.org/10.1103/PhysRevX.8.041019}{Phys. Rev. X {\bf 8}, 041019 (2018)}.

\bibitem{Chalker2}
A.~Chan,~A.~De Luca,~J.~T.~Chalker, \emph{Spectral Statistics in Spatially Extended Chaotic Quantum Many-Body Systems},  \href{https://doi.org/10.1103/PhysRevLett.121.060601}{Phys. Rev. Lett. {\bf 121}, 060601 (2018)}.

\bibitem{Chalker3}
A. J. Friedman, A. Chan, A. De Luca, J. T. Chalker, \emph{Spectral Statistics and Many-Body Quantum Chaos with Conserved Charge},  \href{http://dx.doi.org/10.1103/PhysRevLett.123.210603}{Phys. Rev. Lett. {\bf 123}, 210603 (2019)}. 


\bibitem{BKP:kickedIsing}
B. Bertini, P. Kos, and T. Prosen, \emph{Exact Spectral Form Factor in a Minimal Model of Many-Body Quantum Chaos}, 
\href{http://dx.doi.org/10.1103/PhysRevLett.121.264101}{Phys. Rev. Lett. {\bf 121}, 264101 (2018)}.

\bibitem{Chalker4}
S. J. Garratt, J. T. Chalker, \emph{Many-body quantum chaos and the local pairing of Feynman histories}, \href{https://arxiv.org/abs/2008.01697}{arXiv:2008.01697 (2020)}. 

\bibitem{Chan} S. Moudgalya, A. Prem, D. A. Huse, and A. Chan, \emph{Spectral statistics in constrained many-body quantum chaotic systems}, \href{https://arxiv.org/pdf/2009.11863.pdf}{arXiv:2009.11863 (2020)}. 

\bibitem{Keselman} A. Keselman, L. Nie, and E. Berg, \emph{Scrambling and Lyapunov Exponent in Unitary Networks with Tunable Interactions}, \href{https://arxiv.org/abs/2009.10104}{arXiv:2009.10104 (2020)}. 

\bibitem{RP20}
D. Roy and T. Prosen, \emph{Random Matrix Spectral Form Factor in Kicked Interacting Fermionic Chains}, \href{https://doi.org/10.1103/PhysRevE.102.060202}{Phys. Rev. E {\bf 102}, 060202(R) (2020)}. 


\bibitem{BKP:OEergodicandmixing}
B. Bertini, P. Kos, and T. Prosen, \emph{Operator Entanglement in Local Quantum Circuits I: Chaotic Dual-Unitary Circuits},  \href{https://doi.org/10.21468/SciPostPhys.8.4.067}{SciPost Phys. {\bf 8}, 067 (2020)}.

\bibitem{BKP:OEsolitons}
B. Bertini, P. Kos, and T. Prosen, \emph{Operator Entanglement in Local Quantum Circuits II: Solitons in Chains of Qubits},  \href{http://dx.doi.org/10.21468/SciPostPhys.8.4.068}{SciPost Phys. {\bf 8}, 068 (2020)}.

\bibitem{ClLa20} 
P. W. Claeys and A. Lamacraft, \emph{Maximum velocity quantum circuits},
\href{https://doi.org/10.1103/PhysRevResearch.2.033032}{Phys. Rev. Research {\bf 2}, 033032 (2020)}.

\bibitem{AlDM19} 
V. Alba, J. Dubail, and M. Medenjak, \emph{Operator Entanglement in Interacting Integrable Quantum Systems: The Case of the Rule 54 Chain},
\href{http://dx.doi.org/10.1103/PhysRevLett.122.250603}{Phys. Rev. Lett. {\bf 122}, 250603 (2019)}.

\bibitem{JHN:OperatorEntanglement}
C. Jonay, D. A. Huse, A. Nahum, \emph{Coarse-grained dynamics of operator and state entanglement},
\href{https://arxiv.org/abs/1803.00089}{arXiv:1803.00089 (2018)}. 

\bibitem{BP:tripartite}
B. Bertini and L. Piroli, \emph{Scrambling in Random Unitary Circuits: Exact Results}, \href{http://dx.doi.org/10.1103/PhysRevB.102.064305}{Phys. Rev. B {\bf 102}, 064305 (2020)}.

\bibitem{Nahum:operatorspreadingRU}
A. Nahum, S. Vijay, and J. Haah, \emph{Operator Spreading in Random Unitary Circuits}, \href{https://doi.org/10.1103/PhysRevX.8.021014}{Phys. Rev. X {\bf 8}, 021014 (2018)}.

\bibitem{Keyserlingk}
C. W. von Keyserlingk, T. Rakovszky, F. Pollmann, and S. L. Sondhi,
\emph{Operator Hydrodynamics, OTOCs, and Entanglement Growth in Systems without Conservation Laws}, \href{https://doi.org/10.1103/PhysRevX.8.021013}{Phys. Rev. X {\bf 8}, 021013 (2018)}. 

\bibitem{RandomCircuitsEnt} 
A. Nahum, J. Ruhman, S. Vijay, and J. Haah, \emph{Quantum Entanglement Growth under Random Unitary Dynamics}, \href{https://doi.org/10.1103/PhysRevX.7.031016}{Phys. Rev. X {\bf 7}, 031016 (2017)}.

\bibitem{BKP:dualunitary}
B. Bertini, P. Kos, and T. Prosen,
\emph{Exact Correlation Functions for Dual-Unitary Lattice Models in $1+1$ Dimensions}, \href{https://doi.org/10.1103/PhysRevLett.123.210601}{Phys. Rev. Lett. {\bf 123}, 210601 (2019)}. 

\bibitem{ShSt14_bh} 
S. H. Shenker and D. Stanford, \emph{Black holes and the butterfly effect},
\href{http://dx.doi.org/10.1007/JHEP03(2014)067}{J. High Energ. Phys. {\bf 2014}, 67 (2014)}.

\bibitem{RoSS15} 
D. A. Roberts, D. Stanford, and L. Susskind, \emph{Localized shocks},
\href{http://dx.doi.org/10.1007/JHEP03(2015)051}{J. High Energ. Phys. {\bf 2015}, 51 (2015)}.

\bibitem{MaSS16} 
J. Maldacena, S. H. Shenker, and D. Stanford, \emph{A bound on chaos},
\href{http://dx.doi.org/10.1007/JHEP08(2016)106}{J. High Energ. Phys. {\bf 2016}, 106 (2016)}.


\bibitem{PZ07} T. Prosen and M. \v Znidari\v c, \emph{Is the efficiency of classical simulations of quantum dynamics related to integrability?},
\href{https://doi.org/10.1103/PhysRevE.75.015202}{Phys. Rev. E {\bf 75}, 015202 (2007)}.


\bibitem{HQRY:tripartiteinfo} 
P. Hosur, X.-L. Qi, D. A. Roberts, and B. Yoshida, \emph{Chaos in quantum channels},
\href{http://dx.doi.org/10.1007/JHEP02(2016)004}{J. High Energ. Phys. {\bf 2016}, 4 (2016)}.

\bibitem{SSS:genSFF}
P.~Saad, S.~H.~Shenker, and D.~Stanford, \emph{A semiclassical ramp in SYK and in gravity},
\href{https://arxiv.org/abs/1806.06840}{arXiv:1806.06840 (2018)}. 

\bibitem{genSFF}
H. Gharibyan, M. Hanada, S. H. Shenker, and M. Tezuka, \emph{Onset of random matrix behavior in scrambling systems}, 
\href{https://doi.org/10.1007/JHEP07(2018)124}{J. High Energ. Phys. {\bf 2018}, 124 (2018)}. 


\bibitem{google}
F. Arute, K. Arya, R. Babbush, et al., \emph{Quantum supremacy using a programmable superconducting processor}, \href{https://doi.org/10.1038/s41586-019-1666-5}{Nature {\bf 574}, 505?510 (2019)}. 

\bibitem{nonSA} R. Prange, \emph{The Spectral Form Factor Is Not Self-Averaging}, \href{https://doi.org/10.1103/PhysRevLett.78.2280}{Phys. Rev. Lett. {\bf 78}, 2280 (1997)}.


\bibitem{CH} E. J. Torres-Herrera and L. F. Santos, \emph{Dynamical manifestations of quantum chaos: correlation hole and bulge}, \href{http://doi.org/10.1098/rsta.2016.0434}{Phil. Trans. R. Soc. A. {\bf 375}: 20160434 (2017)}.

\bibitem{LEReview} T.~Gorin, T.~Prosen, T.~H.~Seligman and M.~\v Znidari\v c, \emph{Dynamics of Loschmidt echoes and fidelity decay}, \href{https://www.sciencedirect.com/science/article/pii/S0370157306003310}{Phys. Rep. {\bf 435}, 33 (2006).}


\bibitem{Haake} F.~Haake, {\em Quantum Signatures of Chaos}, 2nd ed. (Springer, 2001).	

\bibitem{Mehtabook}
M. L. Mehta, \emph{Random Matrices and the Statistical Theory of Spectra}, 2nd ed. (Academic, New York, 1991).
	
\bibitem{COE:SFF} H. Kunz, \emph{The probability distribution of the spectral form factor in random matrix theory},
 \href{https://doi.org/10.1088/0305-4470/32/11/011}{J. Phys. A: Math. Gen. {\bf 32}, 2171 (1999)}.



\bibitem{KBP:CorrelationsPerturbed}
P. Kos, B. Bertini  and T. Prosen, \emph{Correlations in Perturbed Dual-Unitary Circuits: Efficient Path-Integral Formula},  \href{https://link.aps.org/doi/10.1103/PhysRevX.11.011022}{Phys. Rev. X {\bf 11}, 011022 (2021)}.

\bibitem{notecit}
A similar setting, for continuous time, has recently been considered in Ref.~\cite{Ren}. 

\bibitem{SM} See the Supplemental Material at [URL will be inserted by publisher], which includes Ref. \cite{ITensor}, for some useful complementary information.

\bibitem{Altshuler} Al'tshuler, O. L. and Shklovskii, B. I., \emph{Repulsion of energy levels and conductivity of small metal samples}, \href{http://www.jetp.ac.ru/cgi-bin/e/index/e/64/1/p127?a=list}{JETP, Vol. 64, No. 1, p. 127 (1986)}.

\bibitem{Smilansky} N. Argaman, Y. Imry, and U. Smilansky, \emph{Semiclassical analysis of spectral correlations in mesoscopic systems},
 \href{https://doi.org/10.1103/PhysRevB.47.4440}{Phys. Rev. B {\bf 47}, 4440 (1993)}.


\bibitem{ChaosvsMBL} J. \v Suntajs, J. Bon\v ca, T. Prosen, and L. Vidmar, \emph{Quantum chaos challenges many-body localization}, \href{https://doi.org/10.1103/PhysRevE.102.062144}{Phys. Rev. E {\bf 102}, 062144 (2020)}. 


\bibitem{Khemani} V. Khemani, A. Vishwanath, and D. A. Huse, \emph{Operator Spreading and the Emergence of Dissipative Hydrodynamics under Unitary Evolution with Conservation Laws}, \href{https://link.aps.org/doi/10.1103/PhysRevX.8.031057}{Phys. Rev. X {\bf 8}, 031057 (2018)}.

\bibitem{Rakovszky} T. Rakovszky, F. Pollmann, and C. W. von Keyserlingk, \emph{Diffusive Hydrodynamics of Out-of-Time-Ordered Correlators with Charge Conservation}, \href{https://link.aps.org/doi/10.1103/PhysRevX.8.031058}{Phys. Rev. X {\bf 8}, 031058 (2018)}.

\bibitem{Lenart} M. Vanicat, L. Zadnik, and T. Prosen, \emph{Integrable Trotterization: Local Conservation Laws and Boundary Driving}, \href{https://doi.org/10.1103/PhysRevLett.121.030606}{Phys. Rev. Lett. {\bf 121}, 030606 (2018)}.

\bibitem{note} It can be done up to exponential corrections with $L$-independent characteristic time. 


\bibitem{KI_JPA} T.~Prosen, \emph{Chaos and complexity of quantum motion}, \href{https://doi.org/10.1088/1751-8113/40/28/S02}{J. Phys. A {\bf 40}, 7881 (2007)}.

\bibitem{Mirko} M.~Degli Esposti, S.~Graffi, and S.~Isola, \emph{Classical limit of the quantized hyperbolic toral automorphisms},
\href{https://link.springer.com/article/10.1007/BF02101532}{Commun. Math. Phys, {\bf 167}, 471 (1995).}

\bibitem{Casati} G.~Casati and B.~V.~Chirikov, eds., \emph{Quantum chaos between order and disorder}, (Cambridge University Press, Cambridge 1995).

\bibitem{ITensor}
\mbox{ITensor Library} (version 3), \href{http://itensor.org}{http://itensor.org}.

\bibitem{Ren} J. Ren, Q. Li, W. Li, Z. Cai, and X. Wang, \emph{Noise-Driven Universal Dynamics towards an Infinite Temperature State},
\href{https://link.aps.org/doi/10.1103/PhysRevLett.124.130602}{Phys. Rev. Lett. {\bf 124} 130602 (2020)}.


\end{thebibliography}
\end{document}